\documentclass[twocolumn]{svjour3}          
\usepackage{amssymb}
\usepackage{amsmath}
\usepackage[latin1]{inputenc}
\usepackage{float}
\usepackage{cite}
\usepackage[ruled,vlined,algo2e]{algorithm2e}
\usepackage{subfigure}

\usepackage{eepic}
\usepackage{pictexwd} 

\usepackage{picins}

\smartqed  

\usepackage[dvips]{graphicx}       
\DeclareGraphicsExtensions{.eps,.EPS,.ps,.PS}   

\usepackage{verbatim}

\renewcommand{\Bar}[1]{\overline{#1}}

\newcommand{\ELIMINE}[1]{}

\newcommand{\ELIMINEPR}[1]{}

\newcommand{\Fset}[1]{\mathcal{F}}

\newcommand{\st}{\; | \;}

\newcommand{\Minima}[1]{{\cal M}({#1})}
\newcommand{\Card}[1]{|{#1}|}

\newcommand{\Vfunc}[1]{F^\ominus}

\SetKwInput{KwLocal}{Local}

\newtheorem{proper}[theorem]{Property}
\newtheorem{defn}[theorem]{Definition}

\journalname{JMIV}

\begin{document}
\title{On the equivalence between hierarchical segmentations and ultrametric watersheds}
\author{Laurent Najman}%
\institute{Universit\'e Paris-Est, Laboratoire d'Informatique Gaspard-Monge, Equipe A3SI, ESIEE Paris, France\\
\email{l.najman@esiee.fr}}

\markboth{Technical Report IGM 2009-10}{
\MakeLowercase{Ultrametric watersheds}}

\date{Received: date / Accepted: date}




\maketitle
\begin{abstract}
  We study hierarchical segmentation in the framework of edge-weighted
  graphs. We define ultrametric watersheds as topological watersheds
  null on the minima. We prove that there exists a bijection between
  the set of ultrametric watersheds and the set of hierarchical
  segmentations. We end this paper by showing how to use the proposed
  framework in practice on the example of constrained connectivity; in
  particular it allows to compute such a hierarchy following a
  classical watershed-based morphological scheme, which provides an
  efficient algorithm to compute the whole hierarchy.

\end{abstract}
%

\section*{Introduction}
This paper\footnote{This work was partially supported by ANR grant
SURF-NT05-2\_45825} is a contribution to a theory of hierarchical
(image) segmentation in the framework of edge-weighted graphs. Image
segmentation is a process of decomposing an image into regions which
are homogeneous according to some criteria. Intuitively, a
hierarchical segmentation represents an image at different resolution
levels.

In this paper, we introduce a subclass of edge-weigh\-ted graphs that
we call ultrametric watersheds. Theorem~\ref{th:onetoone} states that
there exists a one-to-one correspondence, also called a bijection,
between the set of indexed hierarchical segmentations and the set of
ultrametric watersheds. In other words, to any hierarchical
segmentation (whatever the way the hierarchy is built), it is possible
to associate a representation of that hierarchy by an ultrametric
watershed. Conversely, from any ultrametric watershed, one can infer a
indexed hierarchical segmentation.

\begin{figure*}[htbp]
    \begin{center}
      \begin{tabular}{ccc}
	\subfigure[Original image]{
	  \includegraphics[width=.4\textwidth]{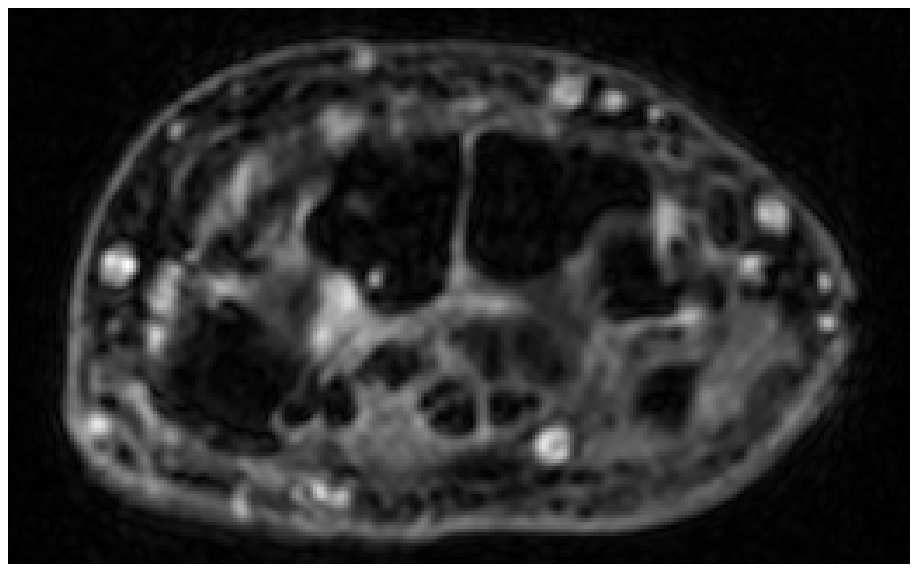}
	}
	&
	~~~~~~~
	&
	\subfigure[Dendrogram of the hierarchical segmentation]{
	  \includegraphics[width=.4\textwidth]{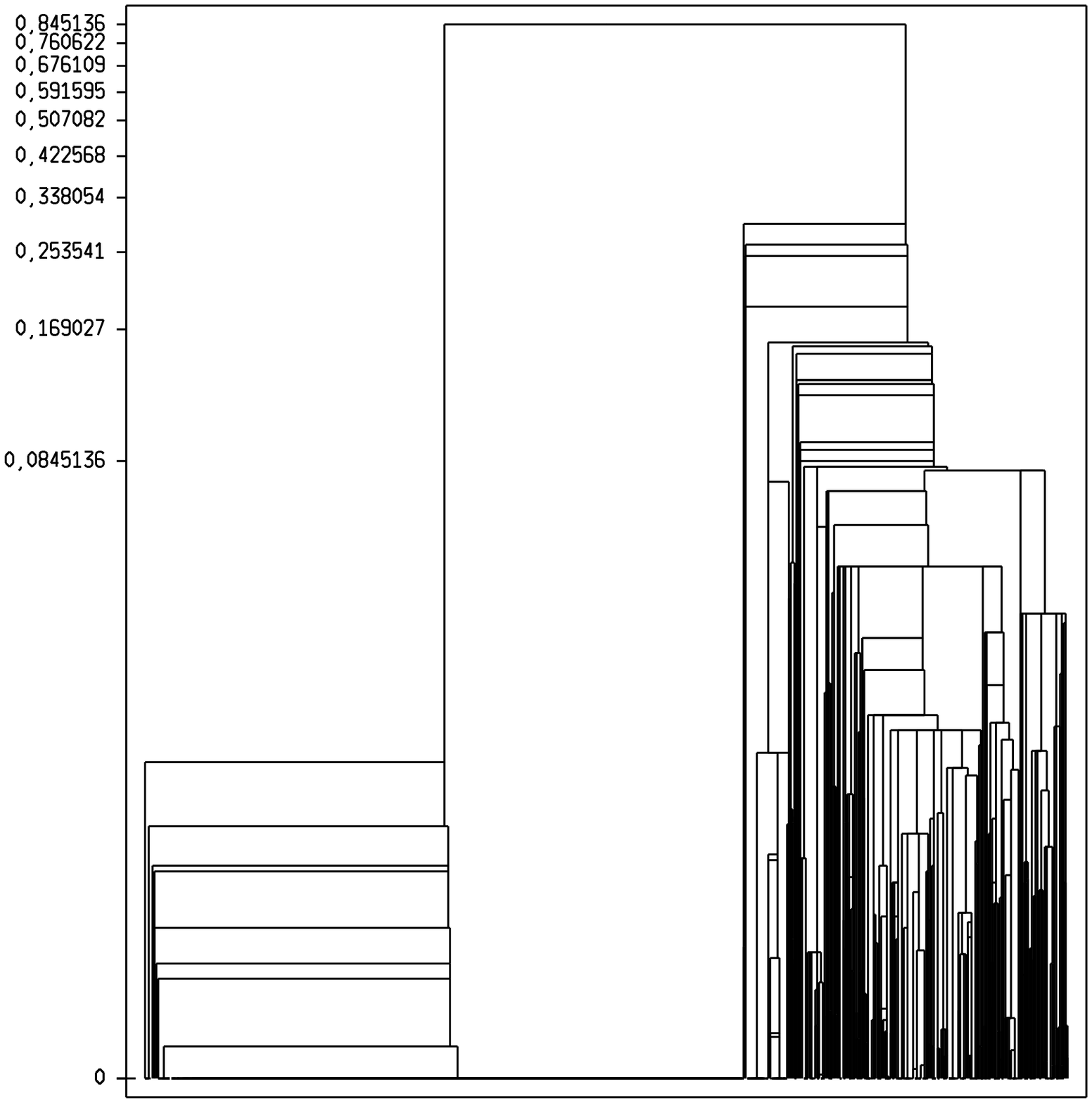}
	}
	\\
	\subfigure[One segmentation extracted from the hierarchy]{
	  \includegraphics[width=.4\textwidth]{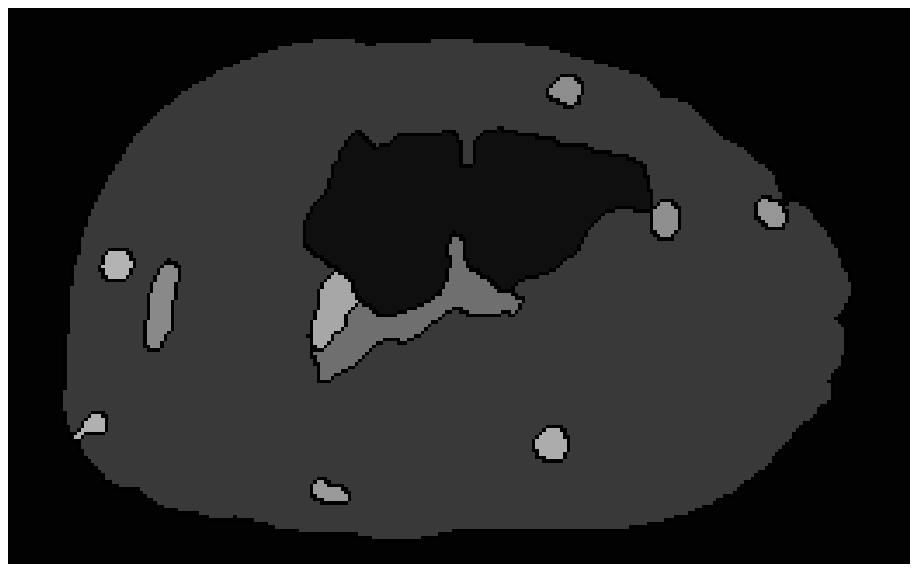}
	}
	&
	~~~~~~~
      	&
	\subfigure[An ultrametric watershed corresponding to 
	  the hierarchical segmentation]{
	  \includegraphics[width=.4\textwidth]{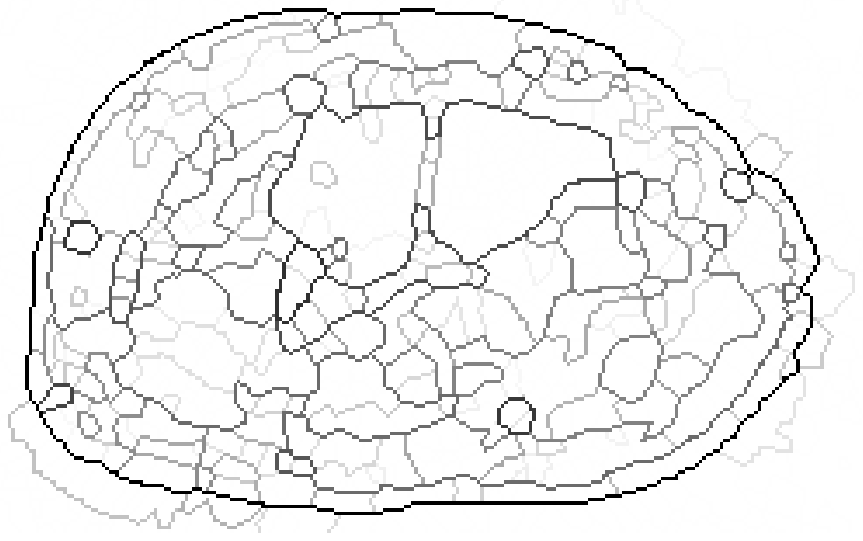}
	}
      \end{tabular}
    \end{center}
    \caption{An example of a hierarchical segmentation produced by the
      method of L.~Guigues~\cite{GuiguesCM06} {\em et al.}. The
      classical order for reading the images is (a), (b), (c),
      (d). But Theorem~\ref{th:onetoone} states that the reading order
      can also be (a), (d), (c), (b) (see text).}
    \label{fig:introduction}
\end{figure*}
This theorem is illustrated on Fig.~\ref{fig:introduction}, that is
produced using the method proposed in~\cite{GuiguesCM06}: what is
usually done is to compute from an original image
(Fig.~\ref{fig:introduction}.a) a hierarchical segmentation that can be
represented by a dendrogram (Fig.~\ref{fig:introduction}.b, see
section~\ref{sec:hier}). The borders of the segmentations extracted
from the hierarchy (such as the one seen in
Fig.~\ref{fig:introduction}.c) can be stacked to form a map
(Fig.~\ref{fig:introduction}.d) that allows for the visual
representation of the hierarchical
segmentation. Theorem~\ref{th:onetoone} gives a characterization of
the class of maps (called ultrametric watersheds) that represent a
hierarchical segmentation; more surprisingly,
Theorem~\ref{th:onetoone} also states that the dendrogram can be
obtained {\em after} the ultrametric watershed has been computed. 

Following~\cite{BBO2004}, we can say that, independently of its
theoretical interest,
such a bijection theorem is useful in practice. Any hierarchical
segmentation problem is a priori heterogeneous: assign to an
edge-weighted graph an indexed hierarchy. Theorem~\ref{th:onetoone}
allows such classification problem to become homogeneous: assign to an
edge-weighted graph a particular edge-weighted graph called
ultrametric watershed. Thus, Theorem~\ref{th:onetoone} gives a meaning
to questions like: which hierarchy is the closest to a given
edge-weighted graph with respect to a given measure or distance?

The paper is organised as follow.  Related works are examined in
section~\ref{sec:prevwork}. We introduce segmentation on edges in
section~\ref{sec:BasicNotions}, and in section~\ref{sec:topows}, we
adapt the topological watershed framework from the framework of gra\-phs
with discrete weights on the nodes to the one of graphs with
real-valued weights on the edges. We then define
(section~\ref{sec:hier}) hierarchies and ultrametric distances. In
section~\ref{sec:hierseg}, we introduce hierarchical segmentations and
ultrametric watersheds, the main result being the existence of a
bijection between these two sets (Th.~\ref{th:onetoone}).  In the last
part of the paper (section~\ref{sec:constcon}), we show how the
proposed framework can be used in practice. After proposing
(section~\ref{sec:representation}) a convenient way to represent
hierarchies as a discrete image, we demonstrate, using ultrametric
watersheds, how to compute constrained connectivity~\cite{Soille2007} as a
classical watershed-based morphological scheme; in particular, it
allows us to provide an efficient algorithm to compute the whole
constrained-connectivity hierarchy.

Apart when otherwise mentionned,
and to the best of the author's knowledge, all the properties and
theorems formally stated in this paper are new. This paper is an
extended version of~\cite{IGMI_Naj09}.

\section{Related works}
\label{sec:prevwork}

This section positions the proposed approach with respect to what has
been done in various different fields. When reading the paper for the
first time, it can be skipped. Readers with a background in
classification will be interested in section~\ref{sec:prevwork1},
those with a background in hierarchical image clustering
section~\ref{sec:prevwork2}, and those with a background in
mathematical morphology by section~\ref{sec:prevwork3}.

\subsection{Hierarchical clustering}
\label{sec:prevwork1}
From its beginning in image processing, hierarchical segmentation has
been thoug\-ht of as a particular instance of hierarchical
classification~\cite{Benzecri73}. One of the fundamental theorems for
hierarchical clustering states that there exists a one-to-one
correspondence between the set of indexed hierarchical classification
and a particular subset of dissimilarity measures called ultrametric
distances; This theorem is generally attributed to
Johnson~\cite{Johnson67}, Jardine {\em et al.}~\cite{JardineSibson71}
and Benz\'ecri~\cite{Benzecri73}.  Since then, numerous generalisations
of that bijection theorem have been proposed (see~\cite{BBO2004} for a
recent review).

Theorem~\ref{th:onetoone} (see below) is an extension to hierarchical
segmentation of this fundamental hierarchical clustering theorem.
Note that the direction of this extension is different from what is
done classically in hierarchical clustering. For example,
E.~Diday~\cite{Diday2008} looks for proper dissimilarities that are
compatible with the underlying lattice. An ultrametric watershed $F$
is not a proper dissimilarity, {\em i.e.} $F(x,y)=0$ does not imply
that $x=y$ (see section~\ref{sec:hier}). But $F$ is an ultrametric
distance (and thus a proper dissimilarity) on the set of connected
components of $\{(x,y) | F(x,y)=0\}$, those connected components being
the regions of a segmentation.

Another point of view on our extension is the following: some authors
assimilate classification and segmentation. We advocate that there
exists a fundamental difference: in classification, we work on the
complete graph, {\em i.e.} the underlying connectivity of the image
(like the four-connectivity) is not used, and some points can be put
in the same class because for example, their coordinates are
correlated in some way with their color; thus a class is not always
connected for the underlying graph. In the framework of
segmentation, any region of any level of a hierarchy of
segmentations is connected for the underlying graph. In other
words, our approach yields a constrained classification, the
constraint being the four-connectivity of the classes, or more
generally any connection defining a graph (for the notion of
connection and its links with segmentation,
see~\cite{Serra-2006,Ronse-2008}.)


\subsection{Hierarchical segmentation}
\label{sec:prevwork2}
There exist many methods for building a hierarchical
segmentation~\cite{Pavlidis-1979}, which can be divided in three
classes: bottom-up , top-down or split-and-merge. A recent review of
some of those approaches can be found in~\cite{Soille2007}.  A useful
representation of hierarchical segmentations was introduced
in~\cite{NS96} under the name of {\em saliency map}. This
representation has been used (under several names) by several authors,
for example for visualisation purposes~\cite{GuiguesCM06} or for
comparing hierarchies~\cite{Arbelaez-Cohen-2006}.

In this paper, we show that any saliency map is an ultrametric
watershed, and conversely.


\subsection{Watersheds}
\label{sec:prevwork3}
For bottom-up approaches, a generic way to build a hierarchical
segmentation is to start from an initial segmentation and
progressively merge regions together~\cite{PavlidisCh5-1977}.  Often,
this initial segmentation is obtained through a
watershed~\cite{meyer-beucher90,NS96,Meyer2005}.
See~\cite{meyer.najman:segmentation} for a recent review of these
notions in the context of mathematical morphology~\cite{NajTal10}.

Among many others~\cite{RM01}, topological watershed~\cite{Ber05} is
an original approach to watersheding that modifies a map (e.g., a
grayscale image) while preserving the connectivity of each lower
cross-section. It has been proved~\cite{Ber05,NCB05} that this
approach is the only one that preserves altitudes of the passes (named
connection values in this paper) between regions of the
segmentation. Pass altitudes are fundamental for hierarchical
schemes~\cite{NS96}.  On the other hand, topological watersheds may be
thick. A study of the properties of different kinds of graphs with
respect to the thinness of watersheds can be found
in~\cite{CouBerCou2008,CouNajBer2008}. An useful framework is that of
edge-weighted graphs, where watersheds are {\em de facto} thin ({\em
  i.e.} of thickness 1); furthermore, in that framework, a subclass of
topological watersheds satisfies both the drop of water principle and
a property of global optimality~\cite{CouBerNaj2008a}. In this
subclass of topological watersheds, some of them can be seen as the
limit, when the power of the weights tends to infinity for some
specific energy function, of classical algorithms like graph cuts or
random walkers~\cite{CGNT09iccv,CouGraNaj09c}.

In this paper, we translate topological watersheds from the framework
of vertice-weigthed-graphs to the one of edge-weighted graphs, and we
identify ultrametric watersheds, a subclass of topological watersheds
that is convenient for hierarchical segmentation.


\section{Segmentation on edges}
\label{sec:BasicNotions}
This paper is settled in the framework of edge-weighted
graphs. Following the notations of \cite{Diestel97}, we present 
some basic definitions to handle such kind of
graphs.
\subsection{Basic notions}

We define a {\em graph} as a pair~$X = (V,E)$ where~$V$ is a finite
set and~$E$ is composed of unordered pairs of~$V$, {\em i.e.},~$E$ is
a subset of~$\left\{\{x,y\} \subseteq V \st x\neq y \right\}$. We
denote by $\Card{V}$ the cardinal of $V$, {\em i.e}, the number of
elements of $V$. Each element of~$V$ is called a {\em vertex or a
point (of~$X$)}, and each element of~$E$ is called an {\em edge (of
$X$)}. If~$V \neq \emptyset$, we say that~$X$ is {\em non-empty}. \\
As several graphs are considered in this paper, whenever this is
necessary, we denote by $V(X)$ and by $E(X)$ the vertex and edge set
of a graph $X$.\\
A graph $X$ is said {\em complete} if $E=V(X)\times V(X)$.\\
Let~$X$ be a graph. If~$u = \{x,y\}$ is an edge of~$X$, we say
that~$x$ and~$y$ are {\em adjacent (for~$X$)}. Let~$\pi=\langle x_0,
\ldots , x_\ell \rangle$ be an ordered sequence of vertices
of~$X$,~$\pi$ is {\em a path from~$x_0$ to~$x_\ell$ in~$X$ (or
in~$V$)} if for any~$i \in [1,\ell]$,~$x_i$ is adjacent to
$x_{i-1}$. In this case, we say that {\em~$x_0$ and~$x_\ell$ are
linked for~$X$}.
We say that {\em~$X$ is connected} if any two vertices of~$X$ are
linked for~$X$.\\
Let~$X$ and~$Y$ be two graphs. If~$V(Y) \subseteq V(X)$ and~$E(Y)$
$\subseteq E(X)$, we say that {\em~$Y$ is a subgraph of~$X$} and we
write~$Y \subseteq X$. We say that~$Y$ is a {\em
connected component of~$X$}, or simply a {\em component of~$X$},
if~$Y$ is a connected subgraph of~$X$ which is maximal for this
property, {\em i.e.}, for any connected graph~$Z$,~$Y \subseteq Z
\subseteq X$ implies~$Z = Y$.\\
Let $X$ be a graph, and let~$S \subseteq E(X)$. The {\em graph induced
by $S$} is the graph whose edge set is~$S$ and whose vertex set is
made of all points that belong to an edge in~$S$, {\em i.e.},~$(\{x
\in V(X) \st \exists u \in S, x \in u\}, S)$.

{\bf Important remark.} {\em Throughout this paper~$G=(V,E)$ denotes a
  connected graph, and the letter $V$ ({\em resp.} $E$) will always
  refer to the vertex set ({\em resp.} the edge set) of $G$. We will
  also assume that~$E \neq \emptyset$.\\
  Let $S\subset E$. In the following, when no confusion may occur, the
  graph induced by~$S$ is also denoted by~$S$.
}

Typically, in applications to image segmentation,~$V$ is the set of
picture elements (pixels) and~$E$ is any of the usual adjacency
relations, {\em e.g.}, the 4- or 8-adjacency in 2D~\cite{KR89}.


If~$S\subset E$, we denote by {\em~$\Bar{S}$ the complementary set
of~$S$ in~$E$}, {\em i.e.},~$\Bar{S} = E \setminus S$.

\subsection{Segmentation in edge-weigthed graphs}
\label{sec:edgesegmentation}
A deep insight on our work is that we are working with edges and not
with points: the minimal unit which we want to modify is an
edge. Indeed, what we need is a discrete space in which we can draw
the border of a segmentation, so that we can represent that
segmentation by its border; in other words, we want to be able to obtain
the regions from their borders, and conversely. In that context, a
desirable property is that the regions of the segmentation are the
connected components of the complement of the border. 

As illustrated in Fig.~\ref{fig:segmentation}.b, this is not possible
to achieve with the classical definition of a point-cut. Indeed,
recall that a {\em partition} of $V$ is a collection ($V_i$) of
non-empty subsets of $V$ such that any element of $V$ is exactly in
one of these subsets, and that a {\em point-cut} is the set of edges
crossing a partition. Even if we add the hypothesis that any $(V_i,
(V_i\times V_i)\cap E)$ is a connected graph, a $V_i$ can be reduced
to an isolated vertice, as the circled grey-point of
Fig.~\ref{fig:segmentation}.b. In that case, the complement of the
point-cut, being a set of edges, does not contain that isolated
vertice. The correct space to work with is the one of edges, and this
motivates the following definitions.

\begin{defn}
\label{def:edgecut}
A set $C\subset E$ is an {\em (edge-)cut (of $G$)} if each edge of $C$
is adjacent to two different nonempty connected components of
$\Bar{C}$. A graph $S$ is called an {\em segmentation (of $G$)}
if $\Bar{E(S)}$ is a cut.  Any connected component of a segmentation
$S$ is called a {\em region (of $S$)}. 
\end{defn}

\begin{figure*}[htbp]
  \begin{center}
    \begin{tabular}{ccc}
      \subfigure[]{
	\includegraphics[width=.3\textwidth]{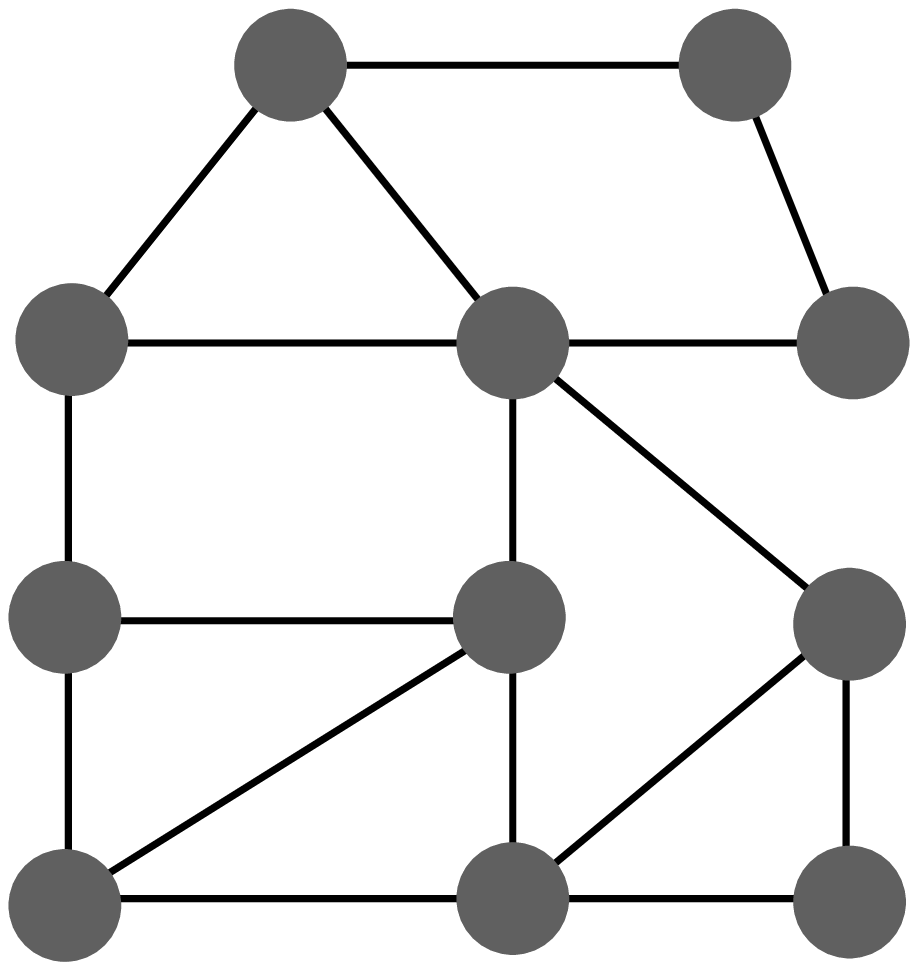}
      } &
      \subfigure[]{
	\includegraphics[width=.3\textwidth]{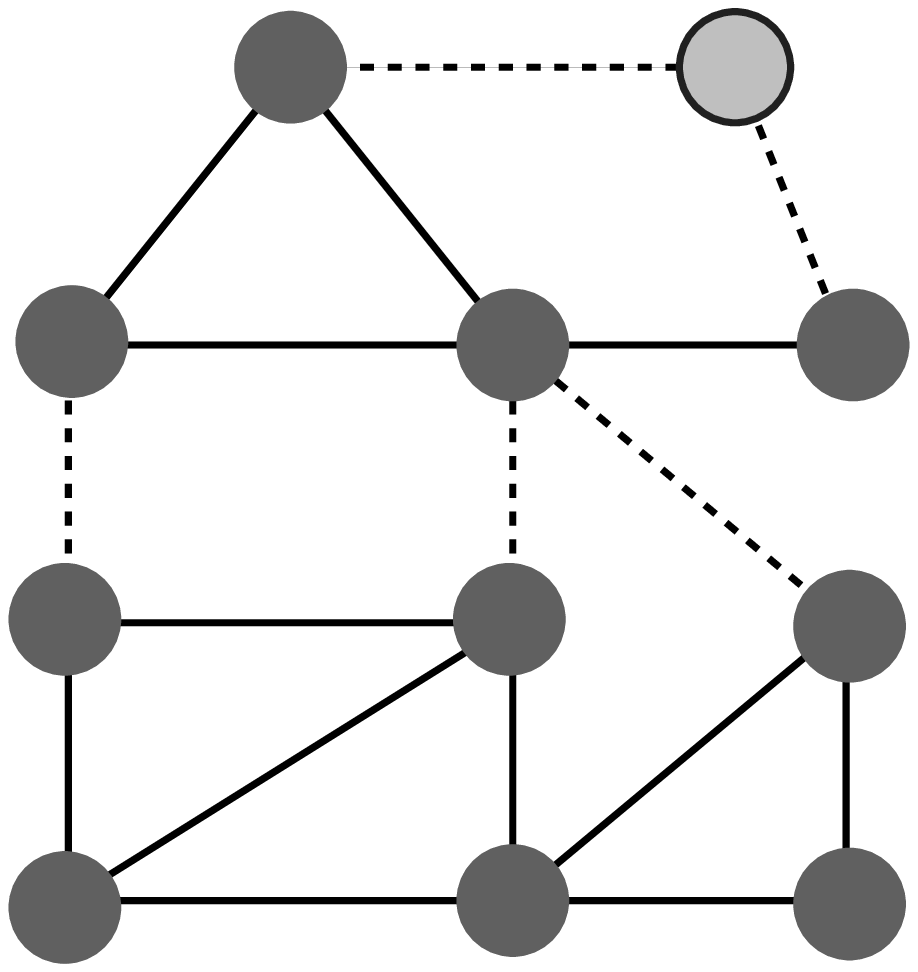}
      } &
      \subfigure[]{
	\includegraphics[width=.3\textwidth]{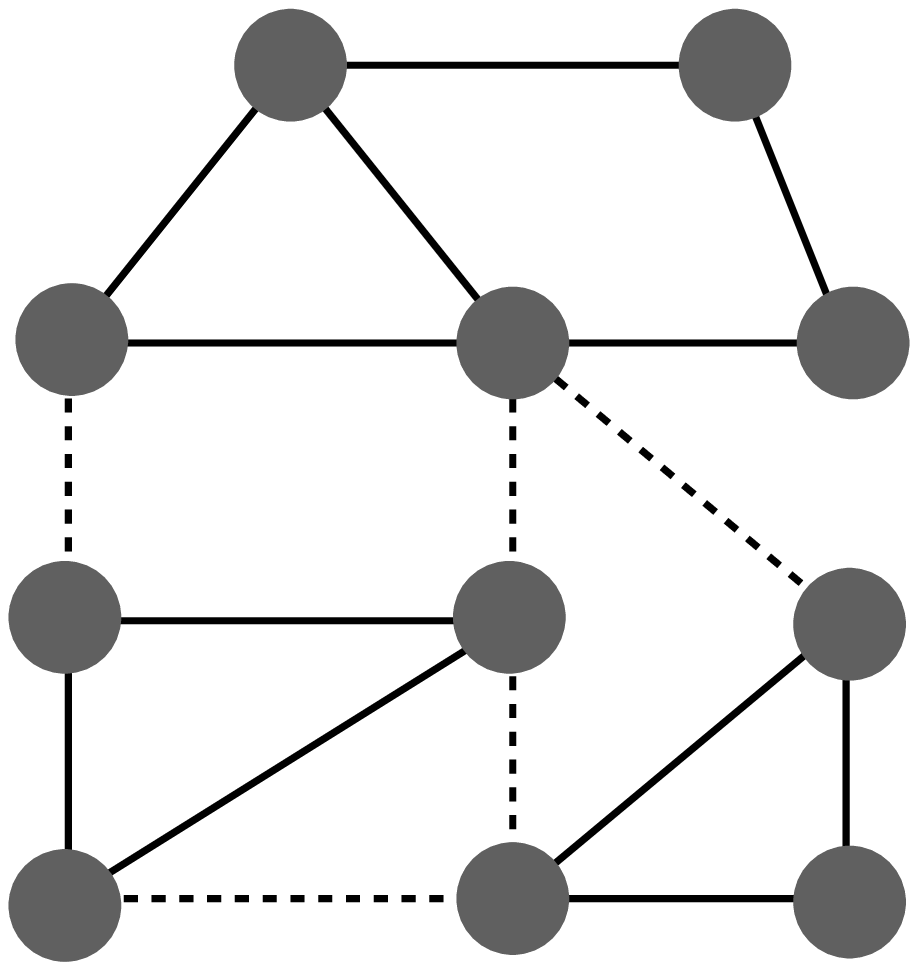}
      }
    \end{tabular}
  \end{center}
  \caption{Illustration of segmentation and edge-cut. (a)~A graph
  $X$. 
  (b)~A subgraph of $X$ which is not a segmentation of $X$: the
    circled grey-point is isolated, and if the point-cut $D$ is the
    set of dotted-lines edges, $\Bar{D}$ contains only two connected
    components, instead of the expected three (see text).
    (c)~An segmentation of $X$; the set $C$ of dotted-lines edges
  is the associated edge-cut of $X$.
}
  \label{fig:segmentation}
\end{figure*}
As mentioned above, the previous definitions of cut and segmentation
(illustrated on Fig.~\ref{fig:segmentation}.c) are not the usual
ones. One can remark the complement of the complement of a cut is the
cut itself, and that any segmentation gives a partition, the converse
being false. In particular, Prop.~\ref{pr:segm}.i below states that
there is no isolated point in an segmentation. If we need an
isolated point $x$, it is always possible to replace $x$ with an edge
$\{x',y'\}$.  An application of the framework of hierarchical
segmentation to constrained connectivity (where isolated points
are present) is described in section~\ref{sec:constcon}.

It is interesting to state the definition of a segmentation from the
point of view of vertices of the graph.  A graph $X$ is said to be
{\em spanning (for $V$)} if $V(X)=V$. We denote by $\phi$ the map that
associates, to any $X\subset G$, the graph $\phi(X) = \{V(X),
\{\{x,y\}\in E | x\in V(X), y\in V(X)\}\}$. We observe that $\phi(X)$ is
maximal among all subgraphs of $G$ that are spanning for $V(X)$, it is
thus a closing on the lattice of subgraphs of
$G$~\cite{IGMI_CouNajSer09}. We call $\phi$ the edge-closing.
\begin{proper}
\label{pr:segm}
A graph $S\subseteq G=(V,E)$ is a segmentation of $G$ if and only if
\begin{enumerate}
\item[(i)] The graph induced by $E(S)$ is $S$; 
\item[(ii)] $S$ is spanning for $V$; 
\item[(iii)] for any connected component $X$ of $S$, $X=\phi(X)$.
\end{enumerate}
\end{proper}
\begin{proof}
  Let $S$ be a segmentation of $G$. Then $\Bar{S}$ is a cut, in other
  word, any edge $v=\{x,y\}\not\in E(S)$ is such that $x$ an $y$ are
  in two different connected components 
  of $\Bar{S}$. As $G$ is connected, that implies that $S$ is spanning for
  $V$. Moreover, $E(S)$ is the set of all edges of $S$, and as $S$ is
  spanning for $G$, the graph induced by $E(S)$ is $(V,E(S))=S$. Let
  $X$ be a connected component of $S$, suppose that there exists
  $v=\{x,y\}\in E$ such that $x$ and $y$ belong to $X$ and $v\not\in
  E(X)$. But then $v\not\in E(S)$ and thus $x$ and $y$ are in two
  different connected components of $S$, a contradiction.

  Conversely, let $S$ be a graph satisfying (i), (ii) and (iii) and
  let $v=\{x,y\}\not\in E(S)$.  As, by (ii), $S$ is spanning for $V$,
  assertion (iii) implies that $x$ and $y$ are in two different
  connected components of $\Bar{E(S)}$. Assertion (i) implies that
  there is no isolated points in $S$, thus $\Bar{S}$ is a cut and thus
  $S$ is a segmentation of $G$.$\qed$
\end{proof}







\subsection{Binary watershed}
Let $X$ be a subgraph of $G$. We note $X+u=(V(X)\cup u, E(X)\cup
\{u\})$. In other words, $X+u$ is the graph whose vertice-set is composed
by the points of $V(X)$ and the points of $u$, and whose edge-set is
composed by the edges of $E(X)$ and $u$. An edge $u\in \Bar{E(X)}$ is
said to be {\em W-simple (for $X$)} (see~\cite{Ber05}) if $X$ has the
same number of connected components as $X+u$.\\
A subgraph $X'$ of $G$ is a {\em thickening (of $X$)} if:
\begin{itemize}
\item $X'=X$, or if
\item there exists a graph $X''$ which is a thickening of $X$
and there exists an edge $u$ W-simple for $X''$ and $X'=X''+u$.
\end{itemize}
Thus, informally, a thickening $X'$ of $X$ is obtained by iteratively
adding to $X$ a sequence of edges $u_1,\ldots,u_n$, {\em i.e.}
$X'=X+u_1+\ldots+u_n$, with the constraint that in the sequence
$X_0=X$, $X_{n+1}=X_n+u_{n+1}$, the edge $u_{n+1}$ is W-simple for
$X_n$.

A subgraph $X$ of $G$ such that there does not exist a W-simple edge
for $X$ is called a {\em binary watershed (of $G$)}.

The following property is a consequence of the definitions of
segmentation and binary watershed.
\begin{proper}
\label{pr:segmwshed}
A graph $X\subseteq G=(V,E)$ is a segmentation of $G$ if and only if
$X$ is a binary watershed of $G$ and if $X$ is induced by $E(X)$.
\end{proper}
\begin{proof}
If $X$ is a segmentation, then $\Bar{E(X)}$ is a cut; let $u\in
\Bar{E(X)}$, $u$ is adjacent to two different non-empty connected
components of $E(X)$, in other word $u$ is not W-simple for $X$. Thus
any segmentation is a binary watershed.

Conversely, let $X$ be a binary watershed, any $u\not\in E(X)$ is not
W-simple for $X$ (and thus $u$ is adjacent to two different connected
components of $X$). If furthermore $X$ is induced by $E(X)$ then
$\Bar{E(X)}$ is a cut.$\qed$
\end{proof}

Thus, starting from a set of edges $X$, a segmentation is obtained by
iterative thickening steps until idempotence. The next section extends
the binary watershed approach to edge-weighted graphs.

\section{Topological watershed}
\label{sec:topows}
\subsection{Edge-weighted graphs}
We denote by~$\Fset{E}$ the set of all maps from~$E$ to~$\mathbb{R}^+$
Given any $F\in\Fset{E}$, the positive numbers $F(u)$ for $u\in E$ are
called the {\em weights} and the pair~$(G,F)$ an {\em edge-weighted
  graph}. Whenever no confusion can occur, we will denote the
edge-weighted graph $(G,F)$ by $F$.

For applications to image segmentation, we take for weight $F(u)$,
where $u=\{x,y\}$ is an edge between two pixels~$x$ and~$y$, a
dissimilarity measure between~$x$ and~$y$ ({\em e.g.},~$F(u)$ equals
the absolute difference of intensity between~$x$ and~$y$;
see~\cite{CouBerNaj09} for a more complete discussion on different
ways to set the map~$F$ for image segmentation). Thus, we suppose that
the salient contours are located on the highest edges of $(G,F)$.

Let $\lambda \in \mathbb{R}^+$ and $F\in\mathcal{F}$, we define
$F[\lambda]=\{v \in E \st F(v)\leq \lambda\}$. The graph (induced by)
$F[\lambda]$ is called a {\em (cross)-section} of $F$. A connected
component of a section $F[\lambda]$ is called a {\em component of~$F$
  (at level $\lambda$)}.

We define ${\cal C}(F)$ as the set composed of all the pairs
$[\lambda,C]$, where $\lambda\in\mathbb{R}^+$ and $C$ is a component
of the graph $F[\lambda]$. 
We call {\em altitude of~$[\lambda,C]$\/} the number~$\lambda$. We
note that one can reconstruct $F$ from ${\cal C}(F)$; more precisely,
we have:
\begin{equation}
F(v) = \min\{\lambda \st [\lambda,C]\in{\cal C}(F), v\in E(C) \}
\label{eq:ct1}
\end{equation}
For any component $C$ of $F$, we set $h(C)= \min\{\lambda \st
[\lambda,C]$ $\in{\cal C}(F)\}$. We define ${\cal C}^\star(F)$ as the set
composed by all $[h(C),C]$ where $C$ is a component of $F$.  The set
${\cal C}^\star(F)$, called the component tree of
$F$~\cite{Salembier-Oliveras-Garrido-1998,NajCou2006}, is a finite
subset of ${\cal C}(F)$ that is widely used in practice for image
filtering. Note that the previous equation~(\ref{eq:ct1}) also holds for ${\cal
C}^\star(F)$:
\begin{equation}
F(v) = \min\{\lambda \st [\lambda,C]\in{\cal C}^\star(F), v\in E(C) \}
\end{equation}
We will make use of the component tree in the proof of
Pr.~\ref{pr:saliency}.

A {\em (regional) minimum of~$F$} is a component $X$ of the graph
$F[\lambda]$ such that for all $\lambda_1<\lambda$, $F[\lambda_1] \cap
E(X) = \emptyset$. We remark that a minimum of $F$ is a subgraph of
$G$ and not a subset of vertices of $G$; we also remark that any
minimum $X$ of $F$ is such that $\Card{V(X)}>1$.

We denote by~$\Minima{F}$ the graph whose vertex set and edge set are,
respectively, the union of the vertex sets and edge sets of all minima
of $F$. In Fig.~\ref{fig:wtopo}, boxes are drawn around each of the
minimum of $\Minima{F}$. Note that $\Minima{F}$ is induced by
$E(\Minima{F})$. As a convenient notation, and when no confusion can
occur, we will sometimes write $X\in\Minima{F}$ if $X$ is a connected
component of $\Minima{F}$.

\subsection{Topological watersheds on edge-weighted graphs}
In that section, we extend the definition of topological
watershed~\cite{Ber05} to edge-weighted graphs, and we give an
original characterization of topological watersheds in that framework
(Th.~\ref{th:wtopocarac}).

Let $F\in\Fset{E}$. An edge $u$ such that $F(u)=\lambda$ is said to
be {\em W-destructible (for $F$) with lowest value $\lambda_0$} if
there exists $\lambda_0$ such that, for all $\lambda_1$, $\lambda_0<
\lambda_1\leq \lambda$, $u$ is W-simple for $F[\lambda_1]$ and if $u$
is not W-simple for $F[\lambda_0]$.

A {\em topological watershed (on $G$)} is a map that contains no
W-destructible edges.

\begin{figure*}[htbp]
  \begin{center}
    \begin{tabular}{ccc}
      \subfigure[]{
	\includegraphics[width=.4\textwidth]{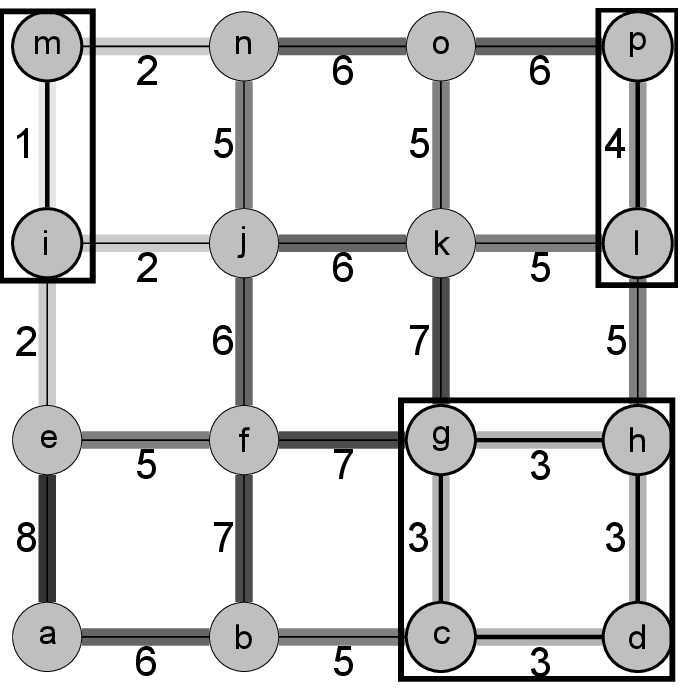}
      } &
      ~~~~~~~
      &
      \subfigure[]{
	\includegraphics[width=.4\textwidth]{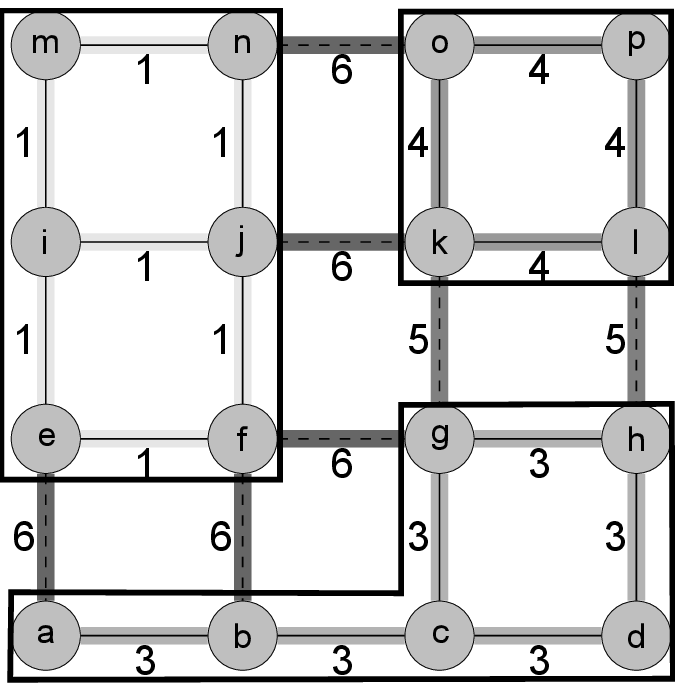}
      }
    \end{tabular}
  \end{center}
  \caption{Illustration of topological watershed. (a)~An edge-weighted
  graph $F$. (b)~A topological watershed of $F$. The minima of (a) are
  $(\{m,i\}), (\{p,l\})$, $(\{g,h\},\{c,d\},\{g,c\},\{h,d\})$. A box is 
  drawn around each one of the minimum in (a) and (b).}
  \label{fig:wtopo}
\end{figure*}
An illustration of a topological watershed can be found in
Fig.~\ref{fig:wtopo}.

A practical way to obtain a topological watershed from any given map
is to apply a topological thinning, that, informally, consists in
lowering W-destructible ed\-ges. More precisely, a map $F'$ is a {\em
  topological thinning (of $F$)} if:
\begin{itemize}
\item $F'=F$, or if
\item there exists a map $F''$ which is a topological thinning of $F$
and there exists an edge $u$ W-destructible for $F''$ with lowest
value $\lambda$ such that $\forall v\neq u, F'(v)=F''(v)$ and
$F'(v)=\lambda_0$, with $\lambda\leq \lambda_0< F''(v)$.
\end{itemize}

A characterization of a W-destructible edge is provided through the
connection value. The {\em connection value} between $x\in V$ and
$y\in V$ is the number
\begin{equation}
\label{eq:connection}
F(x,y)= \min\{\lambda \st [\lambda,C]\in{\cal C}(F), x \in V(C), y\in
V(C) \}
\end{equation}
In other words, $F(x,y)$ is the altitude of the lowest element
$[\lambda,C]$ of ${\cal C}(F)$ such that $x$ and $y$ belong to $C$
(rule of the least common ancestor). 

In Fig.~\ref{fig:wtopo}.a and Fig.~\ref{fig:wtopo}.b, it can be seen
that the connection value between the points $m$ and $p$ is 6, that
the one between $m$ and $d$ is 6, and that the one between $p$ and $d$
is $5$.

The connection value is a practical way to know if an edge is
W-destructible. The following property is a translation of prop.~2
in~\cite{CNB05} to the framework of edge-weighted graphs.
\begin{proper}[Prop.~2 in~\cite{CNB05}]
\label{pr:W-destructible}
Let $F\in\Fset{E}$. An edge $v=\{x,y\}\in E$ is W-destructible for $F$ with
lowest value $\lambda$ if and only if $\lambda=F(x,y)<F(v)$.
\end{proper}

Two points $x$ and $y$ are {\em separated (for $F$)} if $F(x,y) >
\max\{\lambda_1,\lambda_2\}$, where $\lambda_1$ (resp. $\lambda_2$) is
the altitude of the lowest element $[\lambda_1,c_1]$
(resp. $[\lambda_2,c_2]$) of ${\cal C}(F)$ such that $x\in c_1$
(resp. $y\in c_2)$.  The points $x$ and $y$ are {\em
$\lambda$-separated (for $F$)} if they are separated and
$\lambda=F(x,y)$.

The map $F'$ is a {\em separation} of $F$ if, whenever two points are
$\lambda$-separated for $F$, they are $\lambda$-separated for $F'$.
%
%
%

If $X$ and $Y$ are two subgraphs of $G$, we set
$F(X,Y) = \min\{F(x,y) \st  x\in X, y\in Y\}$. 
\begin{theorem}[Restriction to minima~\cite{Ber05}]
\label{th:minima}
Let $F'\leq F$ be two elements of $\Fset{E}$. The map $F'$ is a
separation of $F$ if and only if, for all distinct minima $X$ and $Y$
of $\Minima{F}$, we have $F'(X,Y)=F(X,Y)$.
\end{theorem}

A graph $X$ is {\em flat (for $F$)} if for all $u,v\in E(X)$,
$F(u)=F(v)$.  If $X$ is flat, the {\em altitude} of $X$ is the number
$F(X)$ such that $F(X)=F(v)$ for any $v\in E(X)$.

We say that $F'$ is a {\em strong separation} of $F$ if $F'$ is a
separation of $F$ and if, for each $X'\in\Minima{F'}$, there exists
$X\in\Minima{F}$ such that $X\subseteq X'$ and $F(X)=F(X')$.
\begin{theorem}[strong separation~\cite{Ber05}]
\label{th:separ}
Let $F$ and $F'$ in $\Fset{E}$ with $F'\leq F$. Then $F'$ is a
topological thinning of $F$ if and only if $F'$ is a strong separation
of~$F$.
\end{theorem}
In other words, topological thinnings are the only way to obtain a
watershed that preserves connection values.

In the framework of edge-weighted graphs, topological watersheds
allows for a simple characterization.
\begin{theorem}
\label{th:wtopocarac}
A map $F$ is a topological watershed if and only if:
\begin{enumerate}
\item[(i)] $\Minima{F}$ is a segmentation of $G$;
\item[(ii)] for any edge $v=\{x,y\}$, if there exist $X$ and $Y$ in
$\Minima{F}$, $X\neq Y$, such that $x\in V(X)$ and $y\in V(Y)$, then
$F(v)=F(X,Y)$.
\end{enumerate}
\end{theorem}
\begin{proof}
Let $F$ be a topological watershed. Thus there does not exist any edge
W-destructible for $F$.
\begin{itemize}
\item Suppose that $\Minima{F}$ is not a segmentation of $G$. That
means that there exists an edge $u=\{x,y\}\in\Bar{E(\Minima{F})}$ such
that $x$ and $y$ belongs to the same connected component $X$ of
$\Minima{F}$. That implies that $F(u)>F(X)=F(x,y)$.
By Pr.~\ref{pr:W-destructible}, that implies that the edge $u$ is
W-destructible for $F$, a contradiction. Thus $\Minima{F}$ is a
segmentation of $G$.
\item As $F$ is a topological watershed, we have by
Pr.~\ref{pr:W-destructible} that for any $v=\{x,y\}\in E$,
$F(x,y)=F(v)$. In particular, if there exist $X$ and $Y$ in
$\Minima{F}$, $X\neq Y$, such that $x\in V(X)$ and $y\in V(Y)$, then
$F(v)=F(X,Y)$.
\end{itemize}
Conversely, suppose that $F$ satisfies (i) and (ii). By
Pr.~\ref{pr:W-destructible}, for any edge $v=\{x,y\}\in
E(\Minima{F})$, $F(v)=F(x,y)=F(X)$, and thus $\Minima{F}$ does not
contain any edge W-destructible for $F$. As, by (i), $\Minima{F}$ is a
segmentation, any edge $v\not\in E(\Minima{F})$ satisfies (ii). By
Pr.~\ref{pr:W-destructible}, such an edge $v$ is not W-destructible.
Thus $F$ contains no W-destructible edge and is a topological
watershed.  $\qed$
\end{proof}
Note that if $F$ is a topological watershed, then for any edge
$v=\{x,y\}$ such that there exists $X\in\Minima{F}$ with $x\in V(X)$
and $y\in V(X)$, we have $F(v)=F(X)$. 

\section{Hierarchies and ultrametric distances}
\label{sec:hier}
Let $\Omega$ be a finite set. A {\em hierarchy} $H$ on $\Omega$ is a
set of parts of $\Omega$ such that
\begin{enumerate}
\item[(i)] $\Omega\in H$
\item[(ii)] for every $\omega\in\Omega, \{\omega\}\in H$
\item[(iii)] for each pair $(h,h')\in H^2$, $h\cap
h'\neq\emptyset\implies$ $h\subset h'$ or $h'\subset h$.
\end{enumerate}
The (iii) can be expressed by saying that two elements of a hierarchy
are either disjoint or nested.

An {\em indexed hierarchy} on $\Omega$ is a pair $(H,\mu)$, where $H$
denotes a given hierarchy on $\Omega$ and $\mu$ is a positive function,
defined on $H$ and satisfying the following conditions:
\begin{enumerate}
\item[(i)] $\mu(h)=0$ if and only if $h$ is reduced to a singleton of
$\Omega$;
\item[(ii)] if $h\subset h'$, then $\mu(h)<\mu(h')$.
\end{enumerate}

\begin{figure}[htbp]
  \begin{center}
    \begin{tabular}{ccc}
      \subfigure[Hierarchy]{
	\includegraphics[height=.3\columnwidth]{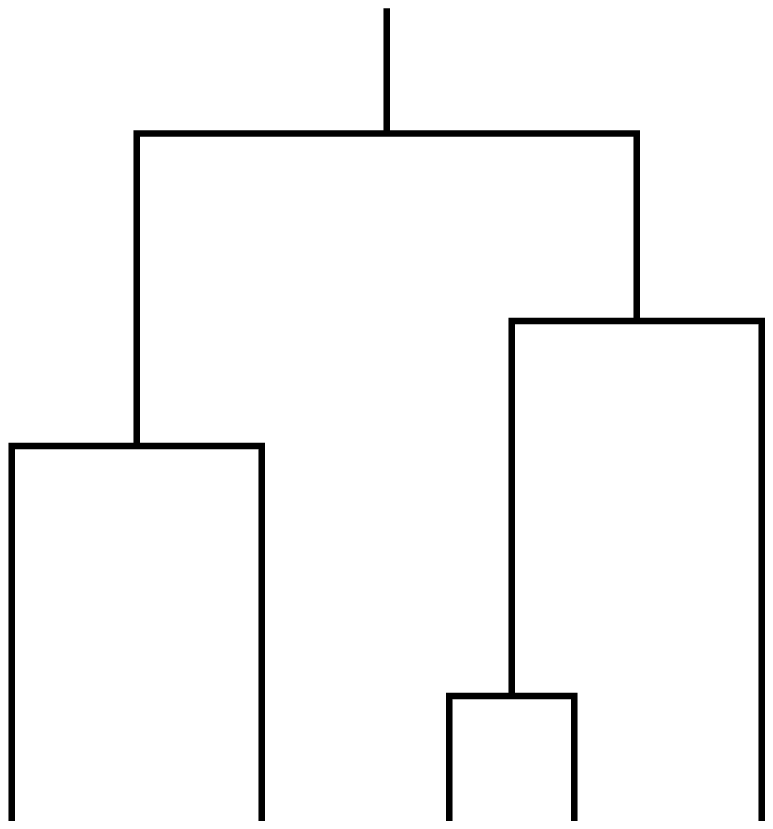}
      }
      &
      ~~~~~~~
      &
      \subfigure[Indexed hierarchy]{
	\includegraphics[height=.3\columnwidth]{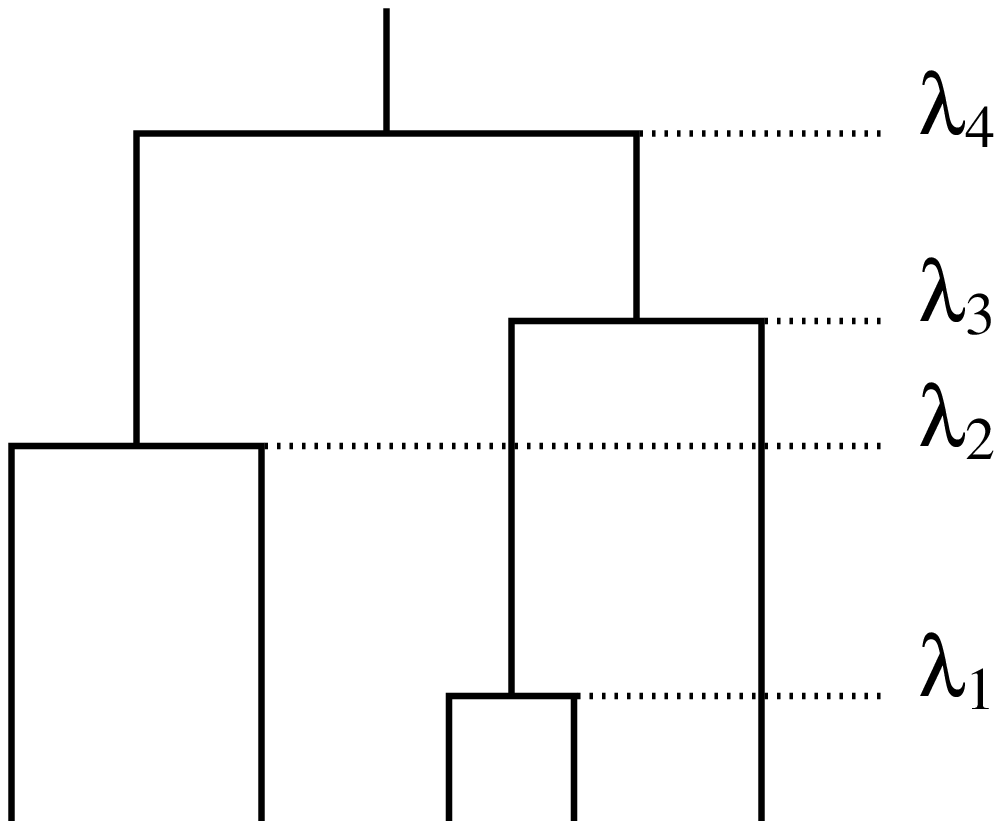}
      }
    \end{tabular}
  \end{center}
  \caption{Hierarchical trees. We have $\lambda_1 < 
  \lambda_3 < \lambda_4$ and $\lambda_2 < \lambda_4$.}
 \label{fig:dendrogram}
\end{figure}

Hierarchy are usually represented using a special type of tree called
{\em dendrograms} (Fig.~\ref{fig:dendrogram}). The leafs of the tree
are the data that are to be classified, while the branching point (the
junctions) are the agglomeration of all the data that are below that
point. In that sense, one can see that, for a given $h$, $\mu(h)$
corresponds to the ``level'' of aggregation, where the elements of $h$
have been aggregated for the first time.

Recall that a {\em dissimilarity on $\Omega$} is a map $d$ from the
Cartesian product $\Omega\times\Omega$ to the set $\mathbb{R}$ of real
numbers such that: $d(\omega_1, \omega_2) = d(\omega_2, \omega_1)$,
$d(\omega_1, \omega_1) = 0$ and $d(\omega_1, \omega_2) \geq 0$ for all
$\omega_1,\omega_2,\omega_3\in\Omega$. The dissimilarity $d$ is said
to be proper whenever $d(\omega_1, \omega_2) = 0$ implies $\omega_1 =
\omega_2$.

A {\em distance $d$ (on $\Omega$)} is a proper dissimilarity that
obeys the triangular inequality $d(\omega_1,\omega_2)\leq
d(\omega_1,\omega_3) + d(\omega_3,\omega_2)$ where $\omega_1,\omega_2$
and $\omega_3$ are any three points of the space.

The ultrametric inequality~\cite{Krasner1944} is stronger than the
triangular inequality.  An {\em ultrametric distance (on $\Omega$)} is
a proper dissimilarity such that, for all
$\omega_1,\omega_2,\omega_3\in\Omega$, $d(\omega_1,\omega_2)\leq
\max(d(\omega_1,\omega_3), d(\omega_2,\omega_3))$

Note that any given partition ($\Omega_i$) of the set $\Omega$ induces
a large number of trivial ultrametric distances:
$d(\omega_1,\omega_1)=0, d(\omega_1,\omega_2)=1$ if
$\omega_1\in\Omega_i$, $\omega_2\in\Omega_j$, $i\neq j$, and
$d(\omega_1,\omega_2)=a$ if $i=j$, $0<a<1$. The general connection
between indexed hierarchies and ultrametric distances goes back to
Benz\'ecri~\cite{Benzecri73} and Johnson~\cite{Johnson67}. They proved
there is a bijection between indexed hierarchies and ultrametric
distances, both defined on the same set. Indeed, associated with each
indexed hierarchy $(H,\mu)$ on $\Omega$ is the following ultrametric
distance:
\begin{equation}
\label{eq:ultrametric}
d(\omega_1,\omega_2) = \min\{\mu(h) \st h\in H, \omega_1\in h,
\omega_2\in h \}.
\end{equation}
In other words, the distance $d(\omega_1,\omega_2)$ between two
elements $\omega_1$ and $\omega_2$ in $\Omega$ is given by the
smallest element in $H$ which contains both $\omega_1$ and
$\omega_2$. Conversely, each ultrametric distance $d$ is associated
with one and only one indexed hierarchy.

Observe the similarity between Eq.~\ref{eq:ultrametric} and
Eq.~\ref{eq:connection}. Indeed, connection value is an ultrametric
distance on $V$ whenever $F>0$. More precisely, we can state the
following property, whose proof is a simple consequence of
Eq.\ref{eq:ultrametric} and Eq.~\ref{eq:connection}.
\begin{proper}
\label{pr:Fultra}
Let $F\in\Fset{E}$. Then $F(X,Y)$ is an ultrametric distance on $\Minima{F}$.
If furthemore, $F>0$, then $F(x,y)$ is an ultrametric distance on $V$.
\end{proper}

Let $\Psi$ be the application that associates to any $F\in{\cal F}$
the map $\Psi(F)$ such that for any edge $\{x,y\}\in E$,
$\Psi(F)(\{x,y\})=F(x,y)$. It is straightforward to see that
$\Psi(F)\leq F$, that $\Psi(\Psi(F))=\Psi(F)$ and that if $F'\leq F$,
$\Psi(F')\leq\Psi(F)$. Thus $\Psi$ is an opening on the lattice
$({\cal F}, \leq)$~\cite{Leclerc81}. We observe that the subset of
strictly positive maps that are defined on the complete graph
$(V,V\times V)$ and that are open with respect to $\Psi$ is the set of
ultrametric distances on $V$.  The mapping $\Psi$ is known under
several names, including ``subdominant ultrametric'' and ``ultrametric
opening''. It is well known that $\Psi$ is associated to the simplest
method for hierarchical classification called single linkage
clustering~\cite{JardineSibson71,GR69}, closely related to Kruskal's
algorithm~\cite{Kruskal56} for computing a minimum spanning tree.


Thanks to Th.~\ref{th:wtopocarac}, we observe that if $F$ is a
topological watershed, then $\Psi(F)=F$. However, an ultrametric
distance $d$ may have plateaus, and thus the weighted complete graph
$(V,V\times V, d)$ is not always a topological
watershed. Nevertheless, those results underline that topological
watersheds are related to hierarchical classification, but not yet to
hierarchical segmentation; the study of such relations is the
subject of the rest of the paper.

\section{Hierarchical segmentations, saliency and ultrametric watersheds}
\label{sec:hierseg}
Informally, a hierarchical segmentation is a hierarchy of connected
regions. However, in our framework, if a segmentation induces a partition,
the converse is not true (see Pr.~\ref{pr:segm}); thus, as the union
of two disjoint connected subgraphs of $G$ is not a connected subgraph
of $G$, the formal definition is slightly more involved.

A {\em hierarchical segmentation (on $G$)} is an indexed hierarchy
$(H,\mu)$ on the set of regions of a segmentation $S$ of $G$, such
that for any $h\in H$, $\phi(\cup_{X\in h} X)$ is connected ($\phi$
being the edge-closing defined in section~\ref{sec:BasicNotions}).

For any $\lambda\geq 0$, we denote by $H[\lambda]$ the graph induced
by $\{\phi(\cup_{X\in h}X)| h\in H, \mu(h)\leq\lambda\}$. The following
property is an easy consequence of the definition of a hierarchical
segmentation.
\begin{proper}
\label{pr:hierseg}
Let $(H,\mu)$ be a  hierarchical segmentation. Then for any
$\lambda\geq 0$, the graph $H[\lambda]$ is a segmentation of~$G$.
\end{proper}
\begin{proof}
Let $(H,\mu)$ be a hierarchical segmentation, and let $\lambda\geq 0$.
Suppose that $H[\lambda]$ is not a segmentation, {\em i.e.} that
$\Bar{H[\lambda]}$ is not a cut. Then there exists a connected
component $X$ of $H[\lambda]$ and $v=\{x,y\}\in\Bar{H[\lambda]}$ such
that $x\in X$ and $y\in X$. That implies that $\phi(X)\neq X$, a
contradiction with the definition of a hierarchical segmentation.
$\qed$
\end{proof}

Prop.~\ref{pr:Fultra} implies that the connection value defines a
hierarchy on the set of minima of $F$.  If $F$ is a topological
watershed, then by Th.~\ref{th:wtopocarac}, $\Minima{F}$ is a
segmentation of $G$, and thus from any topological watershed, one can
infer a hierarchical segmentation. However, $F[\lambda]$ is not always
a segmentation: if there exists a minimum $X$ of $F$ such that
$F(X)=\lambda_0>0$, for any $\lambda_1<\lambda_0$, $F[\lambda_1]$
contains at least two connected components $X_1$ and $X_2$ such that
$\Card{V(X_1)}=\Card{V(X_2)}=1$. Note that the value of $F$ on the
minima of $F$ is not related to the position of the divide nor to the
associated hierarchy of minima/segmentations. This leads us to
introduce the following definition.

\begin{defn}
\label{def:ultraw}
A map $F\in\Fset{E}$ is an {\em ultrametric watershed} if
$F$ is a topological watershed, and if furthemore, for any
$X\in\Minima{F}$, $F(X)=0$.
\end{defn}
Definition~\ref{def:ultraw} directly yields to the nice following
property, illustrated in~Fig.~\ref{fig:hierarchy}, that states that
any level of an ultrametric watershed is a segmentation and
conversely.
\begin{proper}
\label{pr:watultra}
A map $F$ is an ultrametric watershed if and only if for all
$\lambda\geq 0$, $F[\lambda]$ is a segmentation of $G$.
\end{proper}
\begin{proof}
Suppose that $F$ is an ultrametric watershed, then it is a topological
watershed, and by Th.~\ref{th:wtopocarac}.(i), $\Minima{F}$ is a
segmentation of $G$. But as the value of $F$ on its minima is null,
then any cross-section of $F$ is a segmentation of $G$.

Conversely, if for any $\lambda\geq 0$, $F[\lambda]$ is a segmentation of
$G$, then $F$ contains no W-destructible edge for $F$. Indeed, suppose
that there exists an edge $v$ W-destructible for $F$, let
$\lambda=F(v)$, then $v$ is W-simple for $F[\lambda]$. In other words,
adding $v$ to $F[\lambda]$ does not change the number of connected
components of $F[\lambda]$. This is a contradiction with the definition
of a segmentation. Hence $F$ is a topological watershed. Furthermore,
as $F[\lambda]$ is a segmentation for any $\lambda \geq 0$, the value of
$F$ on its minima is null, hence $F$ is an ultrametric watershed.
$\qed$
\end{proof}
\begin{figure*}[htbp]
\begin{center}
  \begin{tabular}{ccc}
    \subfigure[Ultrametric watershed $F$]{
      \includegraphics[width=.4\textwidth]{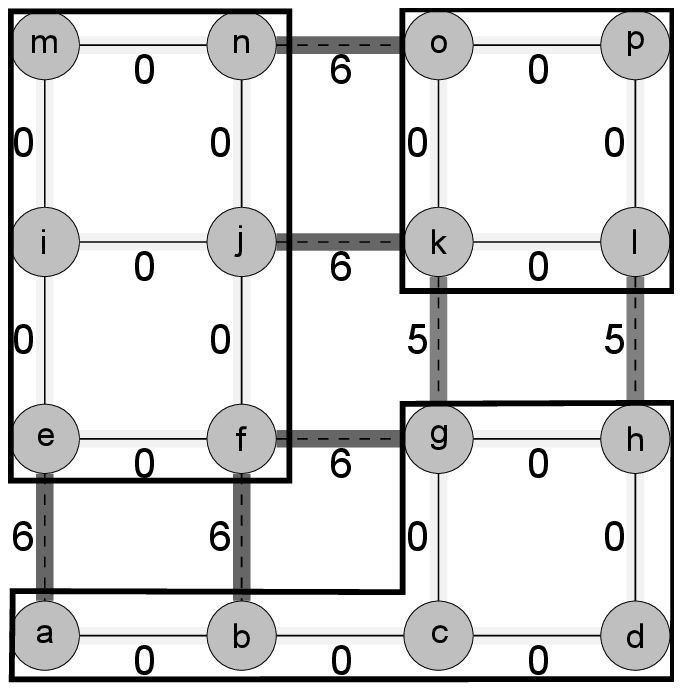}
    }
    &
      ~~~~~~~
    &
    \subfigure[cross-section of $F$ at level 5]{
      \includegraphics[width=.4\textwidth]{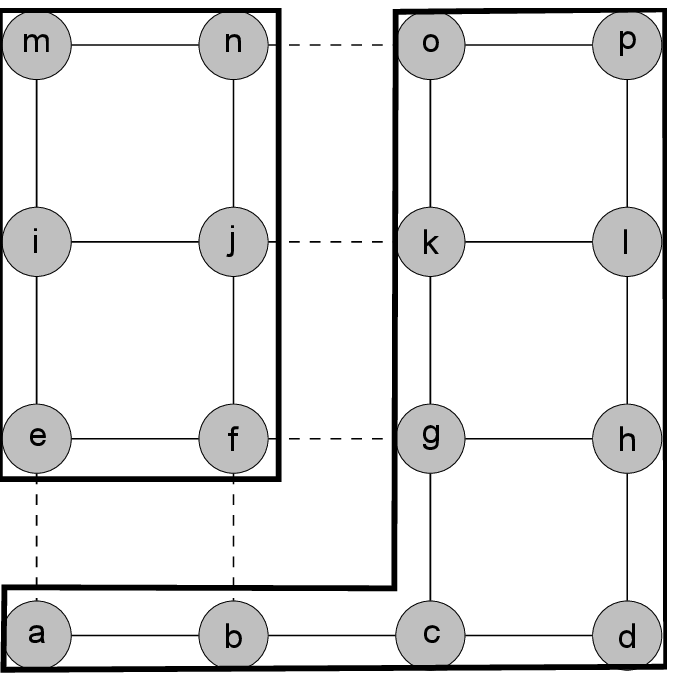}
    }
  \end{tabular}
\end{center}
 \caption{An example of an ultrametric watershed $F$ and a
 cross-section of $F$. In (a), a box is drawn around each 
 one of the minima of $F$, and in (b), a box is drawn 
 around each one of the connected components of the 
 cross-section of $F$. Remark that the minima of $F$, 
 as well as any cross-section of $F$,
 form a segmentation of the graph.}
 \label{fig:hierarchy}
\end{figure*}

By definition of a hierarchy, two elements of $H$ are either disjoint
or nested. If furthermore $(H,\mu)$ is a
hierarchical segmentation, the graphs $\Bar{E(H[\lambda])}$ can be
stacked to form a map.  We call {\em saliency map}~\cite{NS96} the
result of such a stacking, {\em i.e.} a saliency map is a map $F$ such
that there exists $(H,\mu)$ a hierarchical segmentation with
$F(v)=\min\{\lambda| v\in E(H[\lambda])\}$.

\begin{proper}
\label{pr:saliency}
A map $F$ is a saliency map if and only if $F$ is an ultrametric
watershed.
\end{proper}
\begin{proof}
If $F$ is a saliency map, then there exists $(H,\mu)$ a hierarchical
segmentation such that $F(v)=\min\{\lambda \st v\in E(H(\lambda))\}$.
But $F[\lambda] = \{v\st F(v)\leq \lambda\}=\{v\st \min\{\lambda \st
v\in E(H[\lambda])\}\leq\lambda\} = H[\lambda]$ and thus by
Pr.~\ref{pr:hierseg}, for any $\lambda\geq 0$, $F[\lambda]$ is a
segmentation. By Pr.~\ref{pr:watultra}, $F$ is an ultrametric
watershed.




Conversely, let $F$ be an ultrametric watershed, and let ${\cal
C}^\star(F)$ be the component tree of $F$. We build the pair $(H,\mu)$
in the following way: $h\in H$ if and only if there exists
$[\lambda,C]\in {\cal C}^\star(F)$ such that $h=\{X_i \st X_i \in
\Minima{F} \mbox{~and~} X_i\subset C\}$; in that case, we set
$\mu(h)=\lambda$.

Then $(H,\mu)$ is a hierarchical segmentation. Indeed, let $h$
and $h'$ two elements of $H$ such that there exists $[\lambda,X]$ and
$[\lambda',Y]$ in ${\cal C}^\star(F)$ with $h=\{X_0,\ldots,X_p\st
X_i\in\Minima{F} \mbox{~and~} X_i\subset X\}$ and with
$h'=\{Y_0,\ldots,Y_n\st Y_i\in\Minima{F} \mbox{~and~} Y_i\subset Y\}$.

\begin{itemize}
\item by Th.~\ref{th:wtopocarac}, $\Minima{F}$ is a segmentation,
\item we set $\lambda_{\max}=\max\{F(v)\st v\in E\}$, it easy to see
that $[\lambda_{\max},(V,E)]\in{\cal C}^\star(F)$, thus
$\{\cup_{X\in\Minima{F}}\{X\}\}\in H$;
\item any minimum $X$ of $F$ is such that $[0,X]$ belongs to ${\cal
C}^\star(F)$, thus $\{X\}\in H$;
\item furthermore, $h$ and $h'$ are either disjoint or nested:
\begin{itemize}
\item either disjoint: suppose that $X\cap Y=\emptyset$, in that case
$h$ and $h'$ are also disjoint;
\item or nested: suppose that $X\cap Y\neq\emptyset$, then as $X$ and
$Y$ are two connected components of the cross-sections of $F$, either
$X\subset Y$ or $Y\subset X$; suppose that $X\subset Y$; by reordering
the $X_i$ and the $Y_i$, that means that $X_i=Y_i$ for $i=0,\ldots,p$,
$p<n$. In other words, $h\subset h'$.
\end{itemize}
\item by construction, $\mu(h)=0$ if and only if there exists
$X\in\Minima{F}$ such that $h=\{X\}$;
\item If $h\subset h'$, then $\mu(h)<\mu(h')$, because in that case,
$X\subset Y$ and thus $\lambda<\lambda'$.
\end{itemize} 
Thus $(H,\mu)$ is a indexed hierarchy on $\Minima{F}$. 

Furthermore, $\phi(\cup_{X_i\in h}X_i)$ is connected: more precisely,
as $\Minima{F}$ is a segmentation, and as $X$ is a connected component
of the cross-sections of $F$, we have $\phi(\cup_{X_i\in h}X_i)=X$.
Thus $(H,\mu)$ is a hierarchical segmentation.  $\qed$


\end{proof}


The following theorem, a corrolary of Prop.~\ref{pr:saliency}, states
the equivalence between hierarchical segmentations and ultrametric
watersheds. It is the main result of this paper.
\begin{theorem}
\label{th:onetoone}
There exists a bijection between the set of hierarchical
segmen\-tations on $G$ and the set of ultrametric watersheds
on~$G$.
\end{theorem}
\begin{proof}
By Pr.~\ref{pr:saliency}, any ultrametric watershed is a saliency map,
thus for any ultrametric watershed, there exists an associated
hierarchical segmentation. 

Conversely, for any hierarchical segmentation, there exists a
unique saliency map, thus by Pr.~\ref{pr:saliency}, a unique
ultrametric watershed.  $\qed$
\end{proof}
Th.~\ref{th:onetoone} states that any hierarchical segmentation can be
represented by an ultrametric watershed. Such a representation can
easily be built by stacking the border of the regions of the hierarchy
(see Pr.~\ref{pr:hierseg} and~\ref{pr:saliency}, but
also~\cite{NS96,GuiguesCM06,Arbelaez-Cohen-2006}). More interestingly,
Th.~\ref{th:onetoone} also states that any ultrametric watershed
yields a hierarchical segmentation. As the definition of topological
watershed is constructive, 
this is an incentive to searching for algorithmic schemes that
directly compute the whole hierarchy. An exemple of such an
application of Th.~\ref{th:onetoone} is developped in
section~\ref{sec:constcon}.

As there exists a one-to-one correspondence between the set of indexed
hierarchies and the set of ultrametric distances, it is interesting to
search if there exists a similar property for the set of hierarchical
segmentations. Let $d$ be the ultrametric distance associated to a
hierarchical segmentation $(H,\mu)$. We call {\em ultrametric contour
map (associated to $(H,\mu)$)} the map $d_E$ such that:
\begin{enumerate}
\item for any edge $v\in E(H[0])$, then $d_E(v)=0$;
\item for any edge $v=\{x,y\}\in\Bar{E(H[0])}$, $d_E(v)=d(X,Y)$ where $X$
(resp. $Y$) is the connected component of $H[0]$ that contains $x$
(resp. $y$).
\end{enumerate}
\begin{proper}
\label{pr:restriction}
A map $F$ is an ultrametric watershed if and only if $F$ is the
ultrametric contour map associated to a
hierarchical segmentation.
\end{proper}
\begin{proof}
Let $F$ be an ultrametric watershed. By Pr.~\ref{pr:Fultra}, $F(X,Y)$
is an ultrametric distance on $\Minima{F}$. By Pr.~\ref{pr:watultra},
$F$ is a saliency map, hence there exists a hierarchical segmentation
$(H,\mu)$ such that $F(v)=\min\{\lambda \st v\in E(H[\lambda])\}$. In
particular,
\begin{enumerate}
\item for any edge $v\in E(H[0])$, then $F(v)=0$;
\item for any edge $v=\{x,y\}\in\Bar{E(H[0])}$, $F(v)=F(X,Y)$ where $X$
(resp. $Y$) is the connected component of $H[0]$ that contains $x$
(resp. $y$).
\end{enumerate}
Hence $F$ is an ultrametric contour map associated to a hierarchical
segmentation.

Conversely, let $d_E$ be an ultrametric contour map associated to a
hierarchical segmentation $(H,\mu)$. Then by Th.~\ref{th:wtopocarac},
$d_E$ is a topological watershed. Indeed, as $H$ is a hierarchical
segmentation, $H[0]=\Minima{d_E}$ is a segmentation of $G$, and
furthermore for any edge $v=\{x,y\}$, if there exist $X$ and $Y$ in
$\Minima{d_E}$, $X\neq Y$, such that $x\in V(X)$ and $y\in V(Y)$, then
$d_E(v)=d(X,Y)=d_E(X,Y)$. 

Moreover, for any $v\in\Minima{d_E}$,
$d_E(v)=0$, hence $d_E$ is an ultrametric watershed.
$\qed$
\end{proof}



\section{How to use the ultrametric watershed in practice: the example of constrained connectivity}
\label{sec:constcon}
Let us illustrate the usefulness of the proposed framework by
providing an original way of revisiting constrained connectivity
hierarchical segmentations~\cite{Soille2007}, whi\-ch leads to efficient
algorithms. This section is meant as an illustration of our framework,
and, although it is self-sufficient, technical details can be somewhat
difficult to grasp for someone not familiar with the watershed-based
segmentation framework of mathematical
morphology~\cite{meyer.najman:segmentation}. We plan to provide more
information in an extended version of that section.

In this section, we propose to compute an ultrametric watershed that
corresponds to the constrained connectivity hierarchy of a given
image. We show that, in the framework of edge-weighted segmentations,
constrained connectivity can be thought as a classical morphological
scheme, that consists of:
\begin{itemize}
\item computing a gradient;
\item filtering this gradient by attribute filtering;
\item computing a watershed of the filtered gradient.
\end{itemize}
We first discuss how to represent hierarchical segmentations, then
give the formal definition of constrained connectivity, and then we
move on to using ultrametric watersheds for computing such a
hierarchy. In the last part of the section, we will show some other
examples related to the classical watershed-based segmentation schemes.

\subsection{Representations of hierarchical segmentations}
\label{sec:representation}
As we mentionned in section~\ref{sec:edgesegmentation}, one of the
motivations of this work is to be able to imbed the hierarchical
segmentation in a discrete space in a way that can be represented.
Until now, we have used the classical representation of a graph for
all of our examples.

For the purpose of visualisation, it is enough to represent the image
by a grid of double resolution. For example, with the usual four
connectivity in 2D, each pixel will be the center of a 3x3
neighborhood, and if two pixels share an edge, the two corresponding
neighborhoods will share 3 elements corresponding to that edge. The
representation of an ultrametric watershed with double resolution can
be seen in Fig.~\ref{fig:representation}.a.

\begin{figure*}[htbp]
\begin{center}
  \begin{tabular}{ccc}
    \subfigure[Ultrametric watershed $F$ seen with double resolution]{
      \includegraphics[width=.35\textwidth]{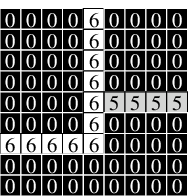}
    }
    &
      ~~~~~~~
    &
    \subfigure[Ultrametric watershed $F$ seen in the Khalimski grid]{
      \includegraphics[width=.35\textwidth]{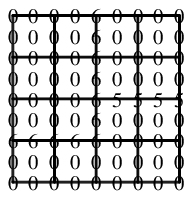}
    }
  \end{tabular}
\end{center}
 \caption{Two possible representations of the ultrametric watershed
 $F$ of Fig.~\ref{fig:hierarchy}.a.}
 \label{fig:representation}
\end{figure*}

\noindent
{\bf Remark: }
A convenient interpretation of the doubling of the
resolution can be given in the framework of cubical complexes, that
have been popularized in computer vision by E.~Khalimski~\cite{KKM90},
but can be found earlier in the literature, originally in the work of
P.S.~Alexandroff~\cite{AleHopf37,Ale37}.

Intuitively, a cubical complexe can be seen as a set of elements of
various dimensions (cubes, squares, segments and points) with specific
rules between those elements. The traditional vision of a numerical
image as being composed of pixels (elementary squares) in 2D or voxels
(elementary cubes) in 3D leads to a natural link between numerical
images and complexes. The representation of an ultrametric watershed
in the Khalimski grid can be seen in Fig.~\ref{fig:representation}.b.

The framework of complexes is useful in the study of topological
properties~\cite{Ber07}.  It is indeed possible to provide a formal
treatment of watersheds in complexes, which we will not do in this
paper. The interested reader can have a look at~\cite{CBCN09iwcia}.



\subsection{Constrained connectivity}
This section is a reminder of P. Soille's approach~\cite{Soille2007},
using the same notations. 

Let $f$ be an application from $V$ to $\mathbb{R}$, {\em i.e.}  an
image with values on the points.  For any set of points $U\subseteq
V$, we set
\begin{equation}
R_f(U) = \sup\{f(x)-f(y)| x,y\in U\}.
\end{equation}
The number $R_f(U)$ is called the {\em range} of $U$ (for $f$).

For any $x\in V$, and for any $\alpha\geq 0$, define~\cite{NMI-1979}
the $\alpha$-connected component $\alpha$-CC(x) as the set:
\begin{eqnarray}
\alpha\mbox{-CC}(x) = \{x\}\cup\{y\in V & | & \mbox{~there exists a
path~}\nonumber\\
& & \pi=\{x_0=x,\ldots,x_n=y\}, \nonumber\\
& & n>0, \mbox{~such that~}\nonumber\\
& & R_f(\{x_i,x_{i+1}\})\leq\alpha,\nonumber\\
& & \mbox{~for all~}0\leq i<n \}
\end{eqnarray}

An essential property of the $\alpha$-connected components of a point
$x$ is that they form an ordered sequence ({\em i.e} a hierarchy) when
increasing the value of $\alpha$:
\begin{equation}
\alpha\mbox{-CC}(x) \subseteq \beta\mbox{-CC}(x)
\end{equation}
whenever $\beta\geq\alpha$. An example of such a hierarchy is given in 
Fig.~\ref{fig:alphapami}.

\begin{figure*}[htbp]
\begin{center}
  \begin{tabular}{ccc}	
    \subfigure[]{
      \setlength{\unitlength}{0.5cm}
\begin{picture}(7,7)(0,0)
\setcoordinatesystem units <0.5cm,0.5cm> point at 0 0
\setplotarea x from 0 to 7, y from 0 to 7
\put(0.5,0.5){\makebox(0,0)[c]{0}}
\put(0.5,0.5){\circle{0.8}}
\put(1.5,0.5){\makebox(0,0)[c]{2}}
\put(1.5,0.5){\circle{0.8}}
\put(2.5,0.5){\makebox(0,0)[c]{9}}
\put(2.5,0.5){\circle{0.8}}
\put(3.5,0.5){\makebox(0,0)[c]{3}}
\put(3.5,0.5){\circle{0.8}}
\put(4.5,0.5){\makebox(0,0)[c]{8}}
\put(4.5,0.5){\circle{0.8}}
\put(5.5,0.5){\makebox(0,0)[c]{5}}
\put(5.5,0.5){\circle{0.8}}
\put(6.5,0.5){\makebox(0,0)[c]{9}}
\put(6.5,0.5){\circle{0.8}}
\put(0.5,1.5){\makebox(0,0)[c]{1}}
\put(0.5,1.5){\circle{0.8}}
\put(1.5,1.5){\makebox(0,0)[c]{0}}
\put(1.5,1.5){\circle{0.8}}
\put(2.5,1.5){\makebox(0,0)[c]{8}}
\put(2.5,1.5){\circle{0.8}}
\put(3.5,1.5){\makebox(0,0)[c]{4}}
\put(3.5,1.5){\circle{0.8}}
\put(4.5,1.5){\makebox(0,0)[c]{9}}
\put(4.5,1.5){\circle{0.8}}
\put(5.5,1.5){\makebox(0,0)[c]{6}}
\put(5.5,1.5){\circle{0.8}}
\put(6.5,1.5){\makebox(0,0)[c]{7}}
\put(6.5,1.5){\circle{0.8}}
\put(0.5,2.5){\makebox(0,0)[c]{3}}
\put(0.5,2.5){\circle{0.8}}
\put(1.5,2.5){\makebox(0,0)[c]{2}}
\put(1.5,2.5){\circle{0.8}}
\put(2.5,2.5){\makebox(0,0)[c]{7}}
\put(2.5,2.5){\circle{0.8}}
\put(3.5,2.5){\makebox(0,0)[c]{9}}
\put(3.5,2.5){\circle{0.8}}
\put(4.5,2.5){\makebox(0,0)[c]{9}}
\put(4.5,2.5){\circle{0.8}}
\put(5.5,2.5){\makebox(0,0)[c]{1}}
\put(5.5,2.5){\circle{0.8}}
\put(6.5,2.5){\makebox(0,0)[c]{1}}
\put(6.5,2.5){\circle{0.8}}
\put(0.5,3.5){\makebox(0,0)[c]{1}}
\put(0.5,3.5){\circle{0.8}}
\put(1.5,3.5){\makebox(0,0)[c]{1}}
\put(1.5,3.5){\circle{0.8}}
\put(2.5,3.5){\makebox(0,0)[c]{9}}
\put(2.5,3.5){\circle{0.8}}
\put(3.5,3.5){\makebox(0,0)[c]{3}}
\put(3.5,3.5){\circle{0.8}}
\put(4.5,3.5){\makebox(0,0)[c]{4}}
\put(4.5,3.5){\circle{0.8}}
\put(5.5,3.5){\makebox(0,0)[c]{2}}
\put(5.5,3.5){\circle{0.8}}
\put(6.5,3.5){\makebox(0,0)[c]{6}}
\put(6.5,3.5){\circle{0.8}}
\put(0.5,4.5){\makebox(0,0)[c]{1}}
\put(0.5,4.5){\circle{0.8}}
\put(1.5,4.5){\makebox(0,0)[c]{0}}
\put(1.5,4.5){\circle{0.8}}
\put(2.5,4.5){\makebox(0,0)[c]{4}}
\put(2.5,4.5){\circle{0.8}}
\put(3.5,4.5){\makebox(0,0)[c]{1}}
\put(3.5,4.5){\circle{0.8}}
\put(4.5,4.5){\makebox(0,0)[c]{1}}
\put(4.5,4.5){\circle{0.8}}
\put(5.5,4.5){\makebox(0,0)[c]{2}}
\put(5.5,4.5){\circle{0.8}}
\put(6.5,4.5){\makebox(0,0)[c]{5}}
\put(6.5,4.5){\circle{0.8}}
\put(0.5,5.5){\makebox(0,0)[c]{2}}
\put(0.5,5.5){\circle{0.8}}
\put(1.5,5.5){\makebox(0,0)[c]{1}}
\put(1.5,5.5){\circle{0.8}}
\put(2.5,5.5){\makebox(0,0)[c]{9}}
\put(2.5,5.5){\circle{0.8}}
\put(3.5,5.5){\makebox(0,0)[c]{8}}
\put(3.5,5.5){\circle{0.8}}
\put(4.5,5.5){\makebox(0,0)[c]{8}}
\put(4.5,5.5){\circle{0.8}}
\put(5.5,5.5){\makebox(0,0)[c]{9}}
\put(5.5,5.5){\circle{0.8}}
\put(6.5,5.5){\makebox(0,0)[c]{1}}
\put(6.5,5.5){\circle{0.8}}
\put(0.5,6.5){\makebox(0,0)[c]{1}}
\put(0.5,6.5){\circle{0.8}}
\put(1.5,6.5){\makebox(0,0)[c]{3}}
\put(1.5,6.5){\circle{0.8}}
\put(2.5,6.5){\makebox(0,0)[c]{8}}
\put(2.5,6.5){\circle{0.8}}
\put(3.5,6.5){\makebox(0,0)[c]{7}}
\put(3.5,6.5){\circle{0.8}}
\put(4.5,6.5){\makebox(0,0)[c]{8}}
\put(4.5,6.5){\circle{0.8}}
\put(5.5,6.5){\makebox(0,0)[c]{8}}
\put(5.5,6.5){\circle{0.8}}
\put(6.5,6.5){\makebox(0,0)[c]{2}}
\put(6.5,6.5){\circle{0.8}}
\allinethickness{1.5pt}
\put(0,0){\line(1,0){7}}
\put(0,0){\line(0,1){7}}
\put(7,7){\line(-1,0){7}}
\put(7,7){\line(0,-1){7}}
\put(1,7){\line(0,-1){1}}
\put(2,7){\line(0,-1){1}}
\put(3,7){\line(0,-1){1}}
\put(4,7){\line(0,-1){1}}
\put(5,6.5){\line(1,0){0.11}}
\put(5,6.5){\line(-1,0){0.11}}
\put(6,7){\line(0,-1){1}}
\put(1,6){\line(0,-1){1}}
\put(2,6){\line(0,-1){1}}
\put(3,6){\line(0,-1){1}}
\put(4,5.5){\line(1,0){0.11}}
\put(4,5.5){\line(-1,0){0.11}}
\put(5,6){\line(0,-1){1}}
\put(6,6){\line(0,-1){1}}
\put(1,5){\line(0,-1){1}}
\put(2,5){\line(0,-1){1}}
\put(3,5){\line(0,-1){1}}
\put(4,4.5){\line(1,0){0.11}}
\put(4,4.5){\line(-1,0){0.11}}
\put(5,5){\line(0,-1){1}}
\put(6,5){\line(0,-1){1}}
\put(1,3.5){\line(1,0){0.11}}
\put(1,3.5){\line(-1,0){0.11}}
\put(2,4){\line(0,-1){1}}
\put(3,4){\line(0,-1){1}}
\put(4,4){\line(0,-1){1}}
\put(5,4){\line(0,-1){1}}
\put(6,4){\line(0,-1){1}}
\put(1,3){\line(0,-1){1}}
\put(2,3){\line(0,-1){1}}
\put(3,3){\line(0,-1){1}}
\put(4,2.5){\line(1,0){0.11}}
\put(4,2.5){\line(-1,0){0.11}}
\put(5,3){\line(0,-1){1}}
\put(6,2.5){\line(1,0){0.11}}
\put(6,2.5){\line(-1,0){0.11}}
\put(1,2){\line(0,-1){1}}
\put(2,2){\line(0,-1){1}}
\put(3,2){\line(0,-1){1}}
\put(4,2){\line(0,-1){1}}
\put(5,2){\line(0,-1){1}}
\put(6,2){\line(0,-1){1}}
\put(1,1){\line(0,-1){1}}
\put(2,1){\line(0,-1){1}}
\put(3,1){\line(0,-1){1}}
\put(4,1){\line(0,-1){1}}
\put(5,1){\line(0,-1){1}}
\put(6,1){\line(0,-1){1}}
\put(0,6){\line(1,0){1}}
\put(0,5){\line(1,0){1}}
\put(0.5,4){\line(0,1){0.11}}
\put(0.5,4){\line(0,-1){0.11}}
\put(0,3){\line(1,0){1}}
\put(0,2){\line(1,0){1}}
\put(0,1){\line(1,0){1}}
\put(1,6){\line(1,0){1}}
\put(1,5){\line(1,0){1}}
\put(1,4){\line(1,0){1}}
\put(1,3){\line(1,0){1}}
\put(1,2){\line(1,0){1}}
\put(1,1){\line(1,0){1}}
\put(2,6){\line(1,0){1}}
\put(2,5){\line(1,0){1}}
\put(2,4){\line(1,0){1}}
\put(2,3){\line(1,0){1}}
\put(2,2){\line(1,0){1}}
\put(2,1){\line(1,0){1}}
\put(3,6){\line(1,0){1}}
\put(3,5){\line(1,0){1}}
\put(3,4){\line(1,0){1}}
\put(3,3){\line(1,0){1}}
\put(3,2){\line(1,0){1}}
\put(3,1){\line(1,0){1}}
\put(4.5,6){\line(0,1){0.11}}
\put(4.5,6){\line(0,-1){0.11}}
\put(4,5){\line(1,0){1}}
\put(4,4){\line(1,0){1}}
\put(4,3){\line(1,0){1}}
\put(4.5,2){\line(0,1){0.11}}
\put(4.5,2){\line(0,-1){0.11}}
\put(4,1){\line(1,0){1}}
\put(5,6){\line(1,0){1}}
\put(5,5){\line(1,0){1}}
\put(5.5,4){\line(0,1){0.11}}
\put(5.5,4){\line(0,-1){0.11}}
\put(5,3){\line(1,0){1}}
\put(5,2){\line(1,0){1}}
\put(5,1){\line(1,0){1}}
\put(6,6){\line(1,0){1}}
\put(6,5){\line(1,0){1}}
\put(6,4){\line(1,0){1}}
\put(6,3){\line(1,0){1}}
\put(6,2){\line(1,0){1}}
\put(6,1){\line(1,0){1}}
\end{picture}
    }
    &
    \subfigure[]{
      \setlength{\unitlength}{0.5cm}
\begin{picture}(7,7)(0,0)
\setcoordinatesystem units <0.5cm,0.5cm> point at 0 0
\setplotarea x from 0 to 7, y from 0 to 7
\put(0.5,0.5){\makebox(0,0)[c]{0}}
\put(0.5,0.5){\circle{0.8}}
\put(1.5,0.5){\makebox(0,0)[c]{2}}
\put(1.5,0.5){\circle{0.8}}
\put(2.5,0.5){\makebox(0,0)[c]{9}}
\put(2.5,0.5){\circle{0.8}}
\put(3.5,0.5){\makebox(0,0)[c]{3}}
\put(3.5,0.5){\circle{0.8}}
\put(4.5,0.5){\makebox(0,0)[c]{8}}
\put(4.5,0.5){\circle{0.8}}
\put(5.5,0.5){\makebox(0,0)[c]{5}}
\put(5.5,0.5){\circle{0.8}}
\put(6.5,0.5){\makebox(0,0)[c]{9}}
\put(6.5,0.5){\circle{0.8}}
\put(0.5,1.5){\makebox(0,0)[c]{1}}
\put(0.5,1.5){\circle{0.8}}
\put(1.5,1.5){\makebox(0,0)[c]{0}}
\put(1.5,1.5){\circle{0.8}}
\put(2.5,1.5){\makebox(0,0)[c]{8}}
\put(2.5,1.5){\circle{0.8}}
\put(3.5,1.5){\makebox(0,0)[c]{4}}
\put(3.5,1.5){\circle{0.8}}
\put(4.5,1.5){\makebox(0,0)[c]{9}}
\put(4.5,1.5){\circle{0.8}}
\put(5.5,1.5){\makebox(0,0)[c]{6}}
\put(5.5,1.5){\circle{0.8}}
\put(6.5,1.5){\makebox(0,0)[c]{7}}
\put(6.5,1.5){\circle{0.8}}
\put(0.5,2.5){\makebox(0,0)[c]{3}}
\put(0.5,2.5){\circle{0.8}}
\put(1.5,2.5){\makebox(0,0)[c]{2}}
\put(1.5,2.5){\circle{0.8}}
\put(2.5,2.5){\makebox(0,0)[c]{7}}
\put(2.5,2.5){\circle{0.8}}
\put(3.5,2.5){\makebox(0,0)[c]{9}}
\put(3.5,2.5){\circle{0.8}}
\put(4.5,2.5){\makebox(0,0)[c]{9}}
\put(4.5,2.5){\circle{0.8}}
\put(5.5,2.5){\makebox(0,0)[c]{1}}
\put(5.5,2.5){\circle{0.8}}
\put(6.5,2.5){\makebox(0,0)[c]{1}}
\put(6.5,2.5){\circle{0.8}}
\put(0.5,3.5){\makebox(0,0)[c]{1}}
\put(0.5,3.5){\circle{0.8}}
\put(1.5,3.5){\makebox(0,0)[c]{1}}
\put(1.5,3.5){\circle{0.8}}
\put(2.5,3.5){\makebox(0,0)[c]{9}}
\put(2.5,3.5){\circle{0.8}}
\put(3.5,3.5){\makebox(0,0)[c]{3}}
\put(3.5,3.5){\circle{0.8}}
\put(4.5,3.5){\makebox(0,0)[c]{4}}
\put(4.5,3.5){\circle{0.8}}
\put(5.5,3.5){\makebox(0,0)[c]{2}}
\put(5.5,3.5){\circle{0.8}}
\put(6.5,3.5){\makebox(0,0)[c]{6}}
\put(6.5,3.5){\circle{0.8}}
\put(0.5,4.5){\makebox(0,0)[c]{1}}
\put(0.5,4.5){\circle{0.8}}
\put(1.5,4.5){\makebox(0,0)[c]{0}}
\put(1.5,4.5){\circle{0.8}}
\put(2.5,4.5){\makebox(0,0)[c]{4}}
\put(2.5,4.5){\circle{0.8}}
\put(3.5,4.5){\makebox(0,0)[c]{1}}
\put(3.5,4.5){\circle{0.8}}
\put(4.5,4.5){\makebox(0,0)[c]{1}}
\put(4.5,4.5){\circle{0.8}}
\put(5.5,4.5){\makebox(0,0)[c]{2}}
\put(5.5,4.5){\circle{0.8}}
\put(6.5,4.5){\makebox(0,0)[c]{5}}
\put(6.5,4.5){\circle{0.8}}
\put(0.5,5.5){\makebox(0,0)[c]{2}}
\put(0.5,5.5){\circle{0.8}}
\put(1.5,5.5){\makebox(0,0)[c]{1}}
\put(1.5,5.5){\circle{0.8}}
\put(2.5,5.5){\makebox(0,0)[c]{9}}
\put(2.5,5.5){\circle{0.8}}
\put(3.5,5.5){\makebox(0,0)[c]{8}}
\put(3.5,5.5){\circle{0.8}}
\put(4.5,5.5){\makebox(0,0)[c]{8}}
\put(4.5,5.5){\circle{0.8}}
\put(5.5,5.5){\makebox(0,0)[c]{9}}
\put(5.5,5.5){\circle{0.8}}
\put(6.5,5.5){\makebox(0,0)[c]{1}}
\put(6.5,5.5){\circle{0.8}}
\put(0.5,6.5){\makebox(0,0)[c]{1}}
\put(0.5,6.5){\circle{0.8}}
\put(1.5,6.5){\makebox(0,0)[c]{3}}
\put(1.5,6.5){\circle{0.8}}
\put(2.5,6.5){\makebox(0,0)[c]{8}}
\put(2.5,6.5){\circle{0.8}}
\put(3.5,6.5){\makebox(0,0)[c]{7}}
\put(3.5,6.5){\circle{0.8}}
\put(4.5,6.5){\makebox(0,0)[c]{8}}
\put(4.5,6.5){\circle{0.8}}
\put(5.5,6.5){\makebox(0,0)[c]{8}}
\put(5.5,6.5){\circle{0.8}}
\put(6.5,6.5){\makebox(0,0)[c]{2}}
\put(6.5,6.5){\circle{0.8}}
\allinethickness{1.5pt}
\put(0,0){\line(1,0){7}}
\put(0,0){\line(0,1){7}}
\put(7,7){\line(-1,0){7}}
\put(7,7){\line(0,-1){7}}
\put(1,7){\line(0,-1){1}}
\put(2,7){\line(0,-1){1}}
\put(3,6.5){\line(1,0){0.11}}
\put(3,6.5){\line(-1,0){0.11}}
\put(4,6.5){\line(1,0){0.11}}
\put(4,6.5){\line(-1,0){0.11}}
\put(5,6.5){\line(1,0){0.11}}
\put(5,6.5){\line(-1,0){0.11}}
\put(6,7){\line(0,-1){1}}
\put(1,5.5){\line(1,0){0.11}}
\put(1,5.5){\line(-1,0){0.11}}
\put(2,6){\line(0,-1){1}}
\put(3,5.5){\line(1,0){0.11}}
\put(3,5.5){\line(-1,0){0.11}}
\put(4,5.5){\line(1,0){0.11}}
\put(4,5.5){\line(-1,0){0.11}}
\put(5,5.5){\line(1,0){0.11}}
\put(5,5.5){\line(-1,0){0.11}}
\put(6,6){\line(0,-1){1}}
\put(1,4.5){\line(1,0){0.11}}
\put(1,4.5){\line(-1,0){0.11}}
\put(2,5){\line(0,-1){1}}
\put(3,5){\line(0,-1){1}}
\put(4,4.5){\line(1,0){0.11}}
\put(4,4.5){\line(-1,0){0.11}}
\put(5,4.5){\line(1,0){0.11}}
\put(5,4.5){\line(-1,0){0.11}}
\put(6,5){\line(0,-1){1}}
\put(1,3.5){\line(1,0){0.11}}
\put(1,3.5){\line(-1,0){0.11}}
\put(2,4){\line(0,-1){1}}
\put(3,4){\line(0,-1){1}}
\put(4,3.5){\line(1,0){0.11}}
\put(4,3.5){\line(-1,0){0.11}}
\put(5,4){\line(0,-1){1}}
\put(6,4){\line(0,-1){1}}
\put(1,2.5){\line(1,0){0.11}}
\put(1,2.5){\line(-1,0){0.11}}
\put(2,3){\line(0,-1){1}}
\put(3,3){\line(0,-1){1}}
\put(4,2.5){\line(1,0){0.11}}
\put(4,2.5){\line(-1,0){0.11}}
\put(5,3){\line(0,-1){1}}
\put(6,2.5){\line(1,0){0.11}}
\put(6,2.5){\line(-1,0){0.11}}
\put(1,1.5){\line(1,0){0.11}}
\put(1,1.5){\line(-1,0){0.11}}
\put(2,2){\line(0,-1){1}}
\put(3,2){\line(0,-1){1}}
\put(4,2){\line(0,-1){1}}
\put(5,2){\line(0,-1){1}}
\put(6,1.5){\line(1,0){0.11}}
\put(6,1.5){\line(-1,0){0.11}}
\put(1,1){\line(0,-1){1}}
\put(2,1){\line(0,-1){1}}
\put(3,1){\line(0,-1){1}}
\put(4,1){\line(0,-1){1}}
\put(5,1){\line(0,-1){1}}
\put(6,1){\line(0,-1){1}}
\put(0.5,6){\line(0,1){0.11}}
\put(0.5,6){\line(0,-1){0.11}}
\put(0.5,5){\line(0,1){0.11}}
\put(0.5,5){\line(0,-1){0.11}}
\put(0.5,4){\line(0,1){0.11}}
\put(0.5,4){\line(0,-1){0.11}}
\put(0,2){\line(1,0){1}}
\put(0.5,1){\line(0,1){0.11}}
\put(0.5,1){\line(0,-1){0.11}}
\put(1,6){\line(1,0){1}}
\put(1.5,5){\line(0,1){0.11}}
\put(1.5,5){\line(0,-1){0.11}}
\put(1.5,4){\line(0,1){0.11}}
\put(1.5,4){\line(0,-1){0.11}}
\put(1.5,3){\line(0,1){0.11}}
\put(1.5,3){\line(0,-1){0.11}}
\put(1,2){\line(1,0){1}}
\put(1,1){\line(1,0){1}}
\put(2.5,6){\line(0,1){0.11}}
\put(2.5,6){\line(0,-1){0.11}}
\put(2,5){\line(1,0){1}}
\put(2,4){\line(1,0){1}}
\put(2,3){\line(1,0){1}}
\put(2.5,2){\line(0,1){0.11}}
\put(2.5,2){\line(0,-1){0.11}}
\put(2.5,1){\line(0,1){0.11}}
\put(2.5,1){\line(0,-1){0.11}}
\put(3.5,6){\line(0,1){0.11}}
\put(3.5,6){\line(0,-1){0.11}}
\put(3,5){\line(1,0){1}}
\put(3,4){\line(1,0){1}}
\put(3,3){\line(1,0){1}}
\put(3,2){\line(1,0){1}}
\put(3.5,1){\line(0,1){0.11}}
\put(3.5,1){\line(0,-1){0.11}}
\put(4.5,6){\line(0,1){0.11}}
\put(4.5,6){\line(0,-1){0.11}}
\put(4,5){\line(1,0){1}}
\put(4,4){\line(1,0){1}}
\put(4,3){\line(1,0){1}}
\put(4.5,2){\line(0,1){0.11}}
\put(4.5,2){\line(0,-1){0.11}}
\put(4.5,1){\line(0,1){0.11}}
\put(4.5,1){\line(0,-1){0.11}}
\put(5.5,6){\line(0,1){0.11}}
\put(5.5,6){\line(0,-1){0.11}}
\put(5,5){\line(1,0){1}}
\put(5.5,4){\line(0,1){0.11}}
\put(5.5,4){\line(0,-1){0.11}}
\put(5.5,3){\line(0,1){0.11}}
\put(5.5,3){\line(0,-1){0.11}}
\put(5,2){\line(1,0){1}}
\put(5.5,1){\line(0,1){0.11}}
\put(5.5,1){\line(0,-1){0.11}}
\put(6.5,6){\line(0,1){0.11}}
\put(6.5,6){\line(0,-1){0.11}}
\put(6,5){\line(1,0){1}}
\put(6.5,4){\line(0,1){0.11}}
\put(6.5,4){\line(0,-1){0.11}}
\put(6,3){\line(1,0){1}}
\put(6,2){\line(1,0){1}}
\put(6,1){\line(1,0){1}}
\end{picture}
    }
    &
    \subfigure[]{
      \setlength{\unitlength}{0.5cm}
\begin{picture}(7,7)(0,0)
\setcoordinatesystem units <0.5cm,0.5cm> point at 0 0
\setplotarea x from 0 to 7, y from 0 to 7
\put(0.5,0.5){\makebox(0,0)[c]{0}}
\put(0.5,0.5){\circle{0.8}}
\put(1.5,0.5){\makebox(0,0)[c]{2}}
\put(1.5,0.5){\circle{0.8}}
\put(2.5,0.5){\makebox(0,0)[c]{9}}
\put(2.5,0.5){\circle{0.8}}
\put(3.5,0.5){\makebox(0,0)[c]{3}}
\put(3.5,0.5){\circle{0.8}}
\put(4.5,0.5){\makebox(0,0)[c]{8}}
\put(4.5,0.5){\circle{0.8}}
\put(5.5,0.5){\makebox(0,0)[c]{5}}
\put(5.5,0.5){\circle{0.8}}
\put(6.5,0.5){\makebox(0,0)[c]{9}}
\put(6.5,0.5){\circle{0.8}}
\put(0.5,1.5){\makebox(0,0)[c]{1}}
\put(0.5,1.5){\circle{0.8}}
\put(1.5,1.5){\makebox(0,0)[c]{0}}
\put(1.5,1.5){\circle{0.8}}
\put(2.5,1.5){\makebox(0,0)[c]{8}}
\put(2.5,1.5){\circle{0.8}}
\put(3.5,1.5){\makebox(0,0)[c]{4}}
\put(3.5,1.5){\circle{0.8}}
\put(4.5,1.5){\makebox(0,0)[c]{9}}
\put(4.5,1.5){\circle{0.8}}
\put(5.5,1.5){\makebox(0,0)[c]{6}}
\put(5.5,1.5){\circle{0.8}}
\put(6.5,1.5){\makebox(0,0)[c]{7}}
\put(6.5,1.5){\circle{0.8}}
\put(0.5,2.5){\makebox(0,0)[c]{3}}
\put(0.5,2.5){\circle{0.8}}
\put(1.5,2.5){\makebox(0,0)[c]{2}}
\put(1.5,2.5){\circle{0.8}}
\put(2.5,2.5){\makebox(0,0)[c]{7}}
\put(2.5,2.5){\circle{0.8}}
\put(3.5,2.5){\makebox(0,0)[c]{9}}
\put(3.5,2.5){\circle{0.8}}
\put(4.5,2.5){\makebox(0,0)[c]{9}}
\put(4.5,2.5){\circle{0.8}}
\put(5.5,2.5){\makebox(0,0)[c]{1}}
\put(5.5,2.5){\circle{0.8}}
\put(6.5,2.5){\makebox(0,0)[c]{1}}
\put(6.5,2.5){\circle{0.8}}
\put(0.5,3.5){\makebox(0,0)[c]{1}}
\put(0.5,3.5){\circle{0.8}}
\put(1.5,3.5){\makebox(0,0)[c]{1}}
\put(1.5,3.5){\circle{0.8}}
\put(2.5,3.5){\makebox(0,0)[c]{9}}
\put(2.5,3.5){\circle{0.8}}
\put(3.5,3.5){\makebox(0,0)[c]{3}}
\put(3.5,3.5){\circle{0.8}}
\put(4.5,3.5){\makebox(0,0)[c]{4}}
\put(4.5,3.5){\circle{0.8}}
\put(5.5,3.5){\makebox(0,0)[c]{2}}
\put(5.5,3.5){\circle{0.8}}
\put(6.5,3.5){\makebox(0,0)[c]{6}}
\put(6.5,3.5){\circle{0.8}}
\put(0.5,4.5){\makebox(0,0)[c]{1}}
\put(0.5,4.5){\circle{0.8}}
\put(1.5,4.5){\makebox(0,0)[c]{0}}
\put(1.5,4.5){\circle{0.8}}
\put(2.5,4.5){\makebox(0,0)[c]{4}}
\put(2.5,4.5){\circle{0.8}}
\put(3.5,4.5){\makebox(0,0)[c]{1}}
\put(3.5,4.5){\circle{0.8}}
\put(4.5,4.5){\makebox(0,0)[c]{1}}
\put(4.5,4.5){\circle{0.8}}
\put(5.5,4.5){\makebox(0,0)[c]{2}}
\put(5.5,4.5){\circle{0.8}}
\put(6.5,4.5){\makebox(0,0)[c]{5}}
\put(6.5,4.5){\circle{0.8}}
\put(0.5,5.5){\makebox(0,0)[c]{2}}
\put(0.5,5.5){\circle{0.8}}
\put(1.5,5.5){\makebox(0,0)[c]{1}}
\put(1.5,5.5){\circle{0.8}}
\put(2.5,5.5){\makebox(0,0)[c]{9}}
\put(2.5,5.5){\circle{0.8}}
\put(3.5,5.5){\makebox(0,0)[c]{8}}
\put(3.5,5.5){\circle{0.8}}
\put(4.5,5.5){\makebox(0,0)[c]{8}}
\put(4.5,5.5){\circle{0.8}}
\put(5.5,5.5){\makebox(0,0)[c]{9}}
\put(5.5,5.5){\circle{0.8}}
\put(6.5,5.5){\makebox(0,0)[c]{1}}
\put(6.5,5.5){\circle{0.8}}
\put(0.5,6.5){\makebox(0,0)[c]{1}}
\put(0.5,6.5){\circle{0.8}}
\put(1.5,6.5){\makebox(0,0)[c]{3}}
\put(1.5,6.5){\circle{0.8}}
\put(2.5,6.5){\makebox(0,0)[c]{8}}
\put(2.5,6.5){\circle{0.8}}
\put(3.5,6.5){\makebox(0,0)[c]{7}}
\put(3.5,6.5){\circle{0.8}}
\put(4.5,6.5){\makebox(0,0)[c]{8}}
\put(4.5,6.5){\circle{0.8}}
\put(5.5,6.5){\makebox(0,0)[c]{8}}
\put(5.5,6.5){\circle{0.8}}
\put(6.5,6.5){\makebox(0,0)[c]{2}}
\put(6.5,6.5){\circle{0.8}}
\allinethickness{1.5pt}
\put(0,0){\line(1,0){7}}
\put(0,0){\line(0,1){7}}
\put(7,7){\line(-1,0){7}}
\put(7,7){\line(0,-1){7}}
\put(1,6.5){\line(1,0){0.11}}
\put(1,6.5){\line(-1,0){0.11}}
\put(2,7){\line(0,-1){1}}
\put(3,6.5){\line(1,0){0.11}}
\put(3,6.5){\line(-1,0){0.11}}
\put(4,6.5){\line(1,0){0.11}}
\put(4,6.5){\line(-1,0){0.11}}
\put(5,6.5){\line(1,0){0.11}}
\put(5,6.5){\line(-1,0){0.11}}
\put(6,7){\line(0,-1){1}}
\put(1,5.5){\line(1,0){0.11}}
\put(1,5.5){\line(-1,0){0.11}}
\put(2,6){\line(0,-1){1}}
\put(3,5.5){\line(1,0){0.11}}
\put(3,5.5){\line(-1,0){0.11}}
\put(4,5.5){\line(1,0){0.11}}
\put(4,5.5){\line(-1,0){0.11}}
\put(5,5.5){\line(1,0){0.11}}
\put(5,5.5){\line(-1,0){0.11}}
\put(6,6){\line(0,-1){1}}
\put(1,4.5){\line(1,0){0.11}}
\put(1,4.5){\line(-1,0){0.11}}
\put(2,5){\line(0,-1){1}}
\put(3,5){\line(0,-1){1}}
\put(4,4.5){\line(1,0){0.11}}
\put(4,4.5){\line(-1,0){0.11}}
\put(5,4.5){\line(1,0){0.11}}
\put(5,4.5){\line(-1,0){0.11}}
\put(6,5){\line(0,-1){1}}
\put(1,3.5){\line(1,0){0.11}}
\put(1,3.5){\line(-1,0){0.11}}
\put(2,4){\line(0,-1){1}}
\put(3,4){\line(0,-1){1}}
\put(4,3.5){\line(1,0){0.11}}
\put(4,3.5){\line(-1,0){0.11}}
\put(5,3.5){\line(1,0){0.11}}
\put(5,3.5){\line(-1,0){0.11}}
\put(6,4){\line(0,-1){1}}
\put(1,2.5){\line(1,0){0.11}}
\put(1,2.5){\line(-1,0){0.11}}
\put(2,3){\line(0,-1){1}}
\put(3,2.5){\line(1,0){0.11}}
\put(3,2.5){\line(-1,0){0.11}}
\put(4,2.5){\line(1,0){0.11}}
\put(4,2.5){\line(-1,0){0.11}}
\put(5,3){\line(0,-1){1}}
\put(6,2.5){\line(1,0){0.11}}
\put(6,2.5){\line(-1,0){0.11}}
\put(1,1.5){\line(1,0){0.11}}
\put(1,1.5){\line(-1,0){0.11}}
\put(2,2){\line(0,-1){1}}
\put(3,2){\line(0,-1){1}}
\put(4,2){\line(0,-1){1}}
\put(5,2){\line(0,-1){1}}
\put(6,1.5){\line(1,0){0.11}}
\put(6,1.5){\line(-1,0){0.11}}
\put(1,0.5){\line(1,0){0.11}}
\put(1,0.5){\line(-1,0){0.11}}
\put(2,1){\line(0,-1){1}}
\put(3,1){\line(0,-1){1}}
\put(4,1){\line(0,-1){1}}
\put(5,1){\line(0,-1){1}}
\put(0.5,6){\line(0,1){0.11}}
\put(0.5,6){\line(0,-1){0.11}}
\put(0.5,5){\line(0,1){0.11}}
\put(0.5,5){\line(0,-1){0.11}}
\put(0.5,4){\line(0,1){0.11}}
\put(0.5,4){\line(0,-1){0.11}}
\put(0.5,3){\line(0,1){0.11}}
\put(0.5,3){\line(0,-1){0.11}}
\put(0.5,2){\line(0,1){0.11}}
\put(0.5,2){\line(0,-1){0.11}}
\put(0.5,1){\line(0,1){0.11}}
\put(0.5,1){\line(0,-1){0.11}}
\put(1.5,6){\line(0,1){0.11}}
\put(1.5,6){\line(0,-1){0.11}}
\put(1.5,5){\line(0,1){0.11}}
\put(1.5,5){\line(0,-1){0.11}}
\put(1.5,4){\line(0,1){0.11}}
\put(1.5,4){\line(0,-1){0.11}}
\put(1.5,3){\line(0,1){0.11}}
\put(1.5,3){\line(0,-1){0.11}}
\put(1.5,2){\line(0,1){0.11}}
\put(1.5,2){\line(0,-1){0.11}}
\put(1.5,1){\line(0,1){0.11}}
\put(1.5,1){\line(0,-1){0.11}}
\put(2.5,6){\line(0,1){0.11}}
\put(2.5,6){\line(0,-1){0.11}}
\put(2,5){\line(1,0){1}}
\put(2,4){\line(1,0){1}}
\put(2.5,3){\line(0,1){0.11}}
\put(2.5,3){\line(0,-1){0.11}}
\put(2.5,2){\line(0,1){0.11}}
\put(2.5,2){\line(0,-1){0.11}}
\put(2.5,1){\line(0,1){0.11}}
\put(2.5,1){\line(0,-1){0.11}}
\put(3.5,6){\line(0,1){0.11}}
\put(3.5,6){\line(0,-1){0.11}}
\put(3,5){\line(1,0){1}}
\put(3.5,4){\line(0,1){0.11}}
\put(3.5,4){\line(0,-1){0.11}}
\put(3,3){\line(1,0){1}}
\put(3,2){\line(1,0){1}}
\put(3.5,1){\line(0,1){0.11}}
\put(3.5,1){\line(0,-1){0.11}}
\put(4.5,6){\line(0,1){0.11}}
\put(4.5,6){\line(0,-1){0.11}}
\put(4,5){\line(1,0){1}}
\put(4,3){\line(1,0){1}}
\put(4.5,2){\line(0,1){0.11}}
\put(4.5,2){\line(0,-1){0.11}}
\put(4.5,1){\line(0,1){0.11}}
\put(4.5,1){\line(0,-1){0.11}}
\put(5.5,6){\line(0,1){0.11}}
\put(5.5,6){\line(0,-1){0.11}}
\put(5,5){\line(1,0){1}}
\put(5.5,4){\line(0,1){0.11}}
\put(5.5,4){\line(0,-1){0.11}}
\put(5.5,3){\line(0,1){0.11}}
\put(5.5,3){\line(0,-1){0.11}}
\put(5,2){\line(1,0){1}}
\put(5.5,1){\line(0,1){0.11}}
\put(5.5,1){\line(0,-1){0.11}}
\put(6.5,6){\line(0,1){0.11}}
\put(6.5,6){\line(0,-1){0.11}}
\put(6,5){\line(1,0){1}}
\put(6.5,4){\line(0,1){0.11}}
\put(6.5,4){\line(0,-1){0.11}}
\put(6,3){\line(1,0){1}}
\put(6,2){\line(1,0){1}}
\put(6.5,1){\line(0,1){0.11}}
\put(6.5,1){\line(0,-1){0.11}}
\end{picture}
    }
\\
    \subfigure[]{
      \setlength{\unitlength}{0.5cm}
\begin{picture}(7,7)(0,0)
\setcoordinatesystem units <0.5cm,0.5cm> point at 0 0
\setplotarea x from 0 to 7, y from 0 to 7
\put(0.5,0.5){\makebox(0,0)[c]{0}}
\put(0.5,0.5){\circle{0.8}}
\put(1.5,0.5){\makebox(0,0)[c]{2}}
\put(1.5,0.5){\circle{0.8}}
\put(2.5,0.5){\makebox(0,0)[c]{9}}
\put(2.5,0.5){\circle{0.8}}
\put(3.5,0.5){\makebox(0,0)[c]{3}}
\put(3.5,0.5){\circle{0.8}}
\put(4.5,0.5){\makebox(0,0)[c]{8}}
\put(4.5,0.5){\circle{0.8}}
\put(5.5,0.5){\makebox(0,0)[c]{5}}
\put(5.5,0.5){\circle{0.8}}
\put(6.5,0.5){\makebox(0,0)[c]{9}}
\put(6.5,0.5){\circle{0.8}}
\put(0.5,1.5){\makebox(0,0)[c]{1}}
\put(0.5,1.5){\circle{0.8}}
\put(1.5,1.5){\makebox(0,0)[c]{0}}
\put(1.5,1.5){\circle{0.8}}
\put(2.5,1.5){\makebox(0,0)[c]{8}}
\put(2.5,1.5){\circle{0.8}}
\put(3.5,1.5){\makebox(0,0)[c]{4}}
\put(3.5,1.5){\circle{0.8}}
\put(4.5,1.5){\makebox(0,0)[c]{9}}
\put(4.5,1.5){\circle{0.8}}
\put(5.5,1.5){\makebox(0,0)[c]{6}}
\put(5.5,1.5){\circle{0.8}}
\put(6.5,1.5){\makebox(0,0)[c]{7}}
\put(6.5,1.5){\circle{0.8}}
\put(0.5,2.5){\makebox(0,0)[c]{3}}
\put(0.5,2.5){\circle{0.8}}
\put(1.5,2.5){\makebox(0,0)[c]{2}}
\put(1.5,2.5){\circle{0.8}}
\put(2.5,2.5){\makebox(0,0)[c]{7}}
\put(2.5,2.5){\circle{0.8}}
\put(3.5,2.5){\makebox(0,0)[c]{9}}
\put(3.5,2.5){\circle{0.8}}
\put(4.5,2.5){\makebox(0,0)[c]{9}}
\put(4.5,2.5){\circle{0.8}}
\put(5.5,2.5){\makebox(0,0)[c]{1}}
\put(5.5,2.5){\circle{0.8}}
\put(6.5,2.5){\makebox(0,0)[c]{1}}
\put(6.5,2.5){\circle{0.8}}
\put(0.5,3.5){\makebox(0,0)[c]{1}}
\put(0.5,3.5){\circle{0.8}}
\put(1.5,3.5){\makebox(0,0)[c]{1}}
\put(1.5,3.5){\circle{0.8}}
\put(2.5,3.5){\makebox(0,0)[c]{9}}
\put(2.5,3.5){\circle{0.8}}
\put(3.5,3.5){\makebox(0,0)[c]{3}}
\put(3.5,3.5){\circle{0.8}}
\put(4.5,3.5){\makebox(0,0)[c]{4}}
\put(4.5,3.5){\circle{0.8}}
\put(5.5,3.5){\makebox(0,0)[c]{2}}
\put(5.5,3.5){\circle{0.8}}
\put(6.5,3.5){\makebox(0,0)[c]{6}}
\put(6.5,3.5){\circle{0.8}}
\put(0.5,4.5){\makebox(0,0)[c]{1}}
\put(0.5,4.5){\circle{0.8}}
\put(1.5,4.5){\makebox(0,0)[c]{0}}
\put(1.5,4.5){\circle{0.8}}
\put(2.5,4.5){\makebox(0,0)[c]{4}}
\put(2.5,4.5){\circle{0.8}}
\put(3.5,4.5){\makebox(0,0)[c]{1}}
\put(3.5,4.5){\circle{0.8}}
\put(4.5,4.5){\makebox(0,0)[c]{1}}
\put(4.5,4.5){\circle{0.8}}
\put(5.5,4.5){\makebox(0,0)[c]{2}}
\put(5.5,4.5){\circle{0.8}}
\put(6.5,4.5){\makebox(0,0)[c]{5}}
\put(6.5,4.5){\circle{0.8}}
\put(0.5,5.5){\makebox(0,0)[c]{2}}
\put(0.5,5.5){\circle{0.8}}
\put(1.5,5.5){\makebox(0,0)[c]{1}}
\put(1.5,5.5){\circle{0.8}}
\put(2.5,5.5){\makebox(0,0)[c]{9}}
\put(2.5,5.5){\circle{0.8}}
\put(3.5,5.5){\makebox(0,0)[c]{8}}
\put(3.5,5.5){\circle{0.8}}
\put(4.5,5.5){\makebox(0,0)[c]{8}}
\put(4.5,5.5){\circle{0.8}}
\put(5.5,5.5){\makebox(0,0)[c]{9}}
\put(5.5,5.5){\circle{0.8}}
\put(6.5,5.5){\makebox(0,0)[c]{1}}
\put(6.5,5.5){\circle{0.8}}
\put(0.5,6.5){\makebox(0,0)[c]{1}}
\put(0.5,6.5){\circle{0.8}}
\put(1.5,6.5){\makebox(0,0)[c]{3}}
\put(1.5,6.5){\circle{0.8}}
\put(2.5,6.5){\makebox(0,0)[c]{8}}
\put(2.5,6.5){\circle{0.8}}
\put(3.5,6.5){\makebox(0,0)[c]{7}}
\put(3.5,6.5){\circle{0.8}}
\put(4.5,6.5){\makebox(0,0)[c]{8}}
\put(4.5,6.5){\circle{0.8}}
\put(5.5,6.5){\makebox(0,0)[c]{8}}
\put(5.5,6.5){\circle{0.8}}
\put(6.5,6.5){\makebox(0,0)[c]{2}}
\put(6.5,6.5){\circle{0.8}}
\allinethickness{1.5pt}
\put(0,0){\line(1,0){7}}
\put(0,0){\line(0,1){7}}
\put(7,7){\line(-1,0){7}}
\put(7,7){\line(0,-1){7}}
\put(1,6.5){\line(1,0){0.11}}
\put(1,6.5){\line(-1,0){0.11}}
\put(2,7){\line(0,-1){1}}
\put(3,6.5){\line(1,0){0.11}}
\put(3,6.5){\line(-1,0){0.11}}
\put(4,6.5){\line(1,0){0.11}}
\put(4,6.5){\line(-1,0){0.11}}
\put(5,6.5){\line(1,0){0.11}}
\put(5,6.5){\line(-1,0){0.11}}
\put(6,7){\line(0,-1){1}}
\put(1,5.5){\line(1,0){0.11}}
\put(1,5.5){\line(-1,0){0.11}}
\put(2,6){\line(0,-1){1}}
\put(3,5.5){\line(1,0){0.11}}
\put(3,5.5){\line(-1,0){0.11}}
\put(4,5.5){\line(1,0){0.11}}
\put(4,5.5){\line(-1,0){0.11}}
\put(5,5.5){\line(1,0){0.11}}
\put(5,5.5){\line(-1,0){0.11}}
\put(6,6){\line(0,-1){1}}
\put(1,4.5){\line(1,0){0.11}}
\put(1,4.5){\line(-1,0){0.11}}
\put(2,5){\line(0,-1){1}}
\put(3,4.5){\line(1,0){0.11}}
\put(3,4.5){\line(-1,0){0.11}}
\put(4,4.5){\line(1,0){0.11}}
\put(4,4.5){\line(-1,0){0.11}}
\put(5,4.5){\line(1,0){0.11}}
\put(5,4.5){\line(-1,0){0.11}}
\put(6,4.5){\line(1,0){0.11}}
\put(6,4.5){\line(-1,0){0.11}}
\put(1,3.5){\line(1,0){0.11}}
\put(1,3.5){\line(-1,0){0.11}}
\put(2,4){\line(0,-1){1}}
\put(3,4){\line(0,-1){1}}
\put(4,3.5){\line(1,0){0.11}}
\put(4,3.5){\line(-1,0){0.11}}
\put(5,3.5){\line(1,0){0.11}}
\put(5,3.5){\line(-1,0){0.11}}
\put(1,2.5){\line(1,0){0.11}}
\put(1,2.5){\line(-1,0){0.11}}
\put(2,3){\line(0,-1){1}}
\put(3,2.5){\line(1,0){0.11}}
\put(3,2.5){\line(-1,0){0.11}}
\put(4,2.5){\line(1,0){0.11}}
\put(4,2.5){\line(-1,0){0.11}}
\put(5,3){\line(0,-1){1}}
\put(6,2.5){\line(1,0){0.11}}
\put(6,2.5){\line(-1,0){0.11}}
\put(1,1.5){\line(1,0){0.11}}
\put(1,1.5){\line(-1,0){0.11}}
\put(2,2){\line(0,-1){1}}
\put(3,2){\line(0,-1){1}}
\put(4,2){\line(0,-1){1}}
\put(5,1.5){\line(1,0){0.11}}
\put(5,1.5){\line(-1,0){0.11}}
\put(6,1.5){\line(1,0){0.11}}
\put(6,1.5){\line(-1,0){0.11}}
\put(1,0.5){\line(1,0){0.11}}
\put(1,0.5){\line(-1,0){0.11}}
\put(2,1){\line(0,-1){1}}
\put(3,1){\line(0,-1){1}}
\put(4,1){\line(0,-1){1}}
\put(5,0.5){\line(1,0){0.11}}
\put(5,0.5){\line(-1,0){0.11}}
\put(0.5,6){\line(0,1){0.11}}
\put(0.5,6){\line(0,-1){0.11}}
\put(0.5,5){\line(0,1){0.11}}
\put(0.5,5){\line(0,-1){0.11}}
\put(0.5,4){\line(0,1){0.11}}
\put(0.5,4){\line(0,-1){0.11}}
\put(0.5,3){\line(0,1){0.11}}
\put(0.5,3){\line(0,-1){0.11}}
\put(0.5,2){\line(0,1){0.11}}
\put(0.5,2){\line(0,-1){0.11}}
\put(0.5,1){\line(0,1){0.11}}
\put(0.5,1){\line(0,-1){0.11}}
\put(1.5,6){\line(0,1){0.11}}
\put(1.5,6){\line(0,-1){0.11}}
\put(1.5,5){\line(0,1){0.11}}
\put(1.5,5){\line(0,-1){0.11}}
\put(1.5,4){\line(0,1){0.11}}
\put(1.5,4){\line(0,-1){0.11}}
\put(1.5,3){\line(0,1){0.11}}
\put(1.5,3){\line(0,-1){0.11}}
\put(1.5,2){\line(0,1){0.11}}
\put(1.5,2){\line(0,-1){0.11}}
\put(1.5,1){\line(0,1){0.11}}
\put(1.5,1){\line(0,-1){0.11}}
\put(2.5,6){\line(0,1){0.11}}
\put(2.5,6){\line(0,-1){0.11}}
\put(2,5){\line(1,0){1}}
\put(2,4){\line(1,0){1}}
\put(2.5,3){\line(0,1){0.11}}
\put(2.5,3){\line(0,-1){0.11}}
\put(2.5,2){\line(0,1){0.11}}
\put(2.5,2){\line(0,-1){0.11}}
\put(2.5,1){\line(0,1){0.11}}
\put(2.5,1){\line(0,-1){0.11}}
\put(3.5,6){\line(0,1){0.11}}
\put(3.5,6){\line(0,-1){0.11}}
\put(3,5){\line(1,0){1}}
\put(3.5,4){\line(0,1){0.11}}
\put(3.5,4){\line(0,-1){0.11}}
\put(3,3){\line(1,0){1}}
\put(3,2){\line(1,0){1}}
\put(3.5,1){\line(0,1){0.11}}
\put(3.5,1){\line(0,-1){0.11}}
\put(4.5,6){\line(0,1){0.11}}
\put(4.5,6){\line(0,-1){0.11}}
\put(4,5){\line(1,0){1}}
\put(4.5,4){\line(0,1){0.11}}
\put(4.5,4){\line(0,-1){0.11}}
\put(4,3){\line(1,0){1}}
\put(4.5,2){\line(0,1){0.11}}
\put(4.5,2){\line(0,-1){0.11}}
\put(4.5,1){\line(0,1){0.11}}
\put(4.5,1){\line(0,-1){0.11}}
\put(5.5,6){\line(0,1){0.11}}
\put(5.5,6){\line(0,-1){0.11}}
\put(5,5){\line(1,0){1}}
\put(5.5,4){\line(0,1){0.11}}
\put(5.5,4){\line(0,-1){0.11}}
\put(5.5,3){\line(0,1){0.11}}
\put(5.5,3){\line(0,-1){0.11}}
\put(5,2){\line(1,0){1}}
\put(5.5,1){\line(0,1){0.11}}
\put(5.5,1){\line(0,-1){0.11}}
\put(6.5,6){\line(0,1){0.11}}
\put(6.5,6){\line(0,-1){0.11}}
\put(6,5){\line(1,0){1}}
\put(6.5,4){\line(0,1){0.11}}
\put(6.5,4){\line(0,-1){0.11}}
\put(6,2){\line(1,0){1}}
\put(6.5,1){\line(0,1){0.11}}
\put(6.5,1){\line(0,-1){0.11}}
\end{picture}
    }
    &
    \subfigure[]{
      \setlength{\unitlength}{0.5cm}
\begin{picture}(7,7)(0,0)
\setcoordinatesystem units <0.5cm,0.5cm> point at 0 0
\setplotarea x from 0 to 7, y from 0 to 7
\put(0.5,0.5){\makebox(0,0)[c]{0}}
\put(0.5,0.5){\circle{0.8}}
\put(1.5,0.5){\makebox(0,0)[c]{2}}
\put(1.5,0.5){\circle{0.8}}
\put(2.5,0.5){\makebox(0,0)[c]{9}}
\put(2.5,0.5){\circle{0.8}}
\put(3.5,0.5){\makebox(0,0)[c]{3}}
\put(3.5,0.5){\circle{0.8}}
\put(4.5,0.5){\makebox(0,0)[c]{8}}
\put(4.5,0.5){\circle{0.8}}
\put(5.5,0.5){\makebox(0,0)[c]{5}}
\put(5.5,0.5){\circle{0.8}}
\put(6.5,0.5){\makebox(0,0)[c]{9}}
\put(6.5,0.5){\circle{0.8}}
\put(0.5,1.5){\makebox(0,0)[c]{1}}
\put(0.5,1.5){\circle{0.8}}
\put(1.5,1.5){\makebox(0,0)[c]{0}}
\put(1.5,1.5){\circle{0.8}}
\put(2.5,1.5){\makebox(0,0)[c]{8}}
\put(2.5,1.5){\circle{0.8}}
\put(3.5,1.5){\makebox(0,0)[c]{4}}
\put(3.5,1.5){\circle{0.8}}
\put(4.5,1.5){\makebox(0,0)[c]{9}}
\put(4.5,1.5){\circle{0.8}}
\put(5.5,1.5){\makebox(0,0)[c]{6}}
\put(5.5,1.5){\circle{0.8}}
\put(6.5,1.5){\makebox(0,0)[c]{7}}
\put(6.5,1.5){\circle{0.8}}
\put(0.5,2.5){\makebox(0,0)[c]{3}}
\put(0.5,2.5){\circle{0.8}}
\put(1.5,2.5){\makebox(0,0)[c]{2}}
\put(1.5,2.5){\circle{0.8}}
\put(2.5,2.5){\makebox(0,0)[c]{7}}
\put(2.5,2.5){\circle{0.8}}
\put(3.5,2.5){\makebox(0,0)[c]{9}}
\put(3.5,2.5){\circle{0.8}}
\put(4.5,2.5){\makebox(0,0)[c]{9}}
\put(4.5,2.5){\circle{0.8}}
\put(5.5,2.5){\makebox(0,0)[c]{1}}
\put(5.5,2.5){\circle{0.8}}
\put(6.5,2.5){\makebox(0,0)[c]{1}}
\put(6.5,2.5){\circle{0.8}}
\put(0.5,3.5){\makebox(0,0)[c]{1}}
\put(0.5,3.5){\circle{0.8}}
\put(1.5,3.5){\makebox(0,0)[c]{1}}
\put(1.5,3.5){\circle{0.8}}
\put(2.5,3.5){\makebox(0,0)[c]{9}}
\put(2.5,3.5){\circle{0.8}}
\put(3.5,3.5){\makebox(0,0)[c]{3}}
\put(3.5,3.5){\circle{0.8}}
\put(4.5,3.5){\makebox(0,0)[c]{4}}
\put(4.5,3.5){\circle{0.8}}
\put(5.5,3.5){\makebox(0,0)[c]{2}}
\put(5.5,3.5){\circle{0.8}}
\put(6.5,3.5){\makebox(0,0)[c]{6}}
\put(6.5,3.5){\circle{0.8}}
\put(0.5,4.5){\makebox(0,0)[c]{1}}
\put(0.5,4.5){\circle{0.8}}
\put(1.5,4.5){\makebox(0,0)[c]{0}}
\put(1.5,4.5){\circle{0.8}}
\put(2.5,4.5){\makebox(0,0)[c]{4}}
\put(2.5,4.5){\circle{0.8}}
\put(3.5,4.5){\makebox(0,0)[c]{1}}
\put(3.5,4.5){\circle{0.8}}
\put(4.5,4.5){\makebox(0,0)[c]{1}}
\put(4.5,4.5){\circle{0.8}}
\put(5.5,4.5){\makebox(0,0)[c]{2}}
\put(5.5,4.5){\circle{0.8}}
\put(6.5,4.5){\makebox(0,0)[c]{5}}
\put(6.5,4.5){\circle{0.8}}
\put(0.5,5.5){\makebox(0,0)[c]{2}}
\put(0.5,5.5){\circle{0.8}}
\put(1.5,5.5){\makebox(0,0)[c]{1}}
\put(1.5,5.5){\circle{0.8}}
\put(2.5,5.5){\makebox(0,0)[c]{9}}
\put(2.5,5.5){\circle{0.8}}
\put(3.5,5.5){\makebox(0,0)[c]{8}}
\put(3.5,5.5){\circle{0.8}}
\put(4.5,5.5){\makebox(0,0)[c]{8}}
\put(4.5,5.5){\circle{0.8}}
\put(5.5,5.5){\makebox(0,0)[c]{9}}
\put(5.5,5.5){\circle{0.8}}
\put(6.5,5.5){\makebox(0,0)[c]{1}}
\put(6.5,5.5){\circle{0.8}}
\put(0.5,6.5){\makebox(0,0)[c]{1}}
\put(0.5,6.5){\circle{0.8}}
\put(1.5,6.5){\makebox(0,0)[c]{3}}
\put(1.5,6.5){\circle{0.8}}
\put(2.5,6.5){\makebox(0,0)[c]{8}}
\put(2.5,6.5){\circle{0.8}}
\put(3.5,6.5){\makebox(0,0)[c]{7}}
\put(3.5,6.5){\circle{0.8}}
\put(4.5,6.5){\makebox(0,0)[c]{8}}
\put(4.5,6.5){\circle{0.8}}
\put(5.5,6.5){\makebox(0,0)[c]{8}}
\put(5.5,6.5){\circle{0.8}}
\put(6.5,6.5){\makebox(0,0)[c]{2}}
\put(6.5,6.5){\circle{0.8}}
\allinethickness{1.5pt}
\put(0,0){\line(1,0){7}}
\put(0,0){\line(0,1){7}}
\put(7,7){\line(-1,0){7}}
\put(7,7){\line(0,-1){7}}
\put(1,6.5){\line(1,0){0.11}}
\put(1,6.5){\line(-1,0){0.11}}
\put(2,7){\line(0,-1){1}}
\put(3,6.5){\line(1,0){0.11}}
\put(3,6.5){\line(-1,0){0.11}}
\put(4,6.5){\line(1,0){0.11}}
\put(4,6.5){\line(-1,0){0.11}}
\put(5,6.5){\line(1,0){0.11}}
\put(5,6.5){\line(-1,0){0.11}}
\put(6,7){\line(0,-1){1}}
\put(1,5.5){\line(1,0){0.11}}
\put(1,5.5){\line(-1,0){0.11}}
\put(2,6){\line(0,-1){1}}
\put(3,5.5){\line(1,0){0.11}}
\put(3,5.5){\line(-1,0){0.11}}
\put(4,5.5){\line(1,0){0.11}}
\put(4,5.5){\line(-1,0){0.11}}
\put(5,5.5){\line(1,0){0.11}}
\put(5,5.5){\line(-1,0){0.11}}
\put(6,6){\line(0,-1){1}}
\put(1,4.5){\line(1,0){0.11}}
\put(1,4.5){\line(-1,0){0.11}}
\put(2,4.5){\line(1,0){0.11}}
\put(2,4.5){\line(-1,0){0.11}}
\put(3,4.5){\line(1,0){0.11}}
\put(3,4.5){\line(-1,0){0.11}}
\put(4,4.5){\line(1,0){0.11}}
\put(4,4.5){\line(-1,0){0.11}}
\put(5,4.5){\line(1,0){0.11}}
\put(5,4.5){\line(-1,0){0.11}}
\put(6,4.5){\line(1,0){0.11}}
\put(6,4.5){\line(-1,0){0.11}}
\put(1,3.5){\line(1,0){0.11}}
\put(1,3.5){\line(-1,0){0.11}}
\put(2,4){\line(0,-1){1}}
\put(3,4){\line(0,-1){1}}
\put(4,3.5){\line(1,0){0.11}}
\put(4,3.5){\line(-1,0){0.11}}
\put(5,3.5){\line(1,0){0.11}}
\put(5,3.5){\line(-1,0){0.11}}
\put(6,3.5){\line(1,0){0.11}}
\put(6,3.5){\line(-1,0){0.11}}
\put(1,2.5){\line(1,0){0.11}}
\put(1,2.5){\line(-1,0){0.11}}
\put(2,3){\line(0,-1){1}}
\put(3,2.5){\line(1,0){0.11}}
\put(3,2.5){\line(-1,0){0.11}}
\put(4,2.5){\line(1,0){0.11}}
\put(4,2.5){\line(-1,0){0.11}}
\put(5,3){\line(0,-1){1}}
\put(6,2.5){\line(1,0){0.11}}
\put(6,2.5){\line(-1,0){0.11}}
\put(1,1.5){\line(1,0){0.11}}
\put(1,1.5){\line(-1,0){0.11}}
\put(2,2){\line(0,-1){1}}
\put(3,1.5){\line(1,0){0.11}}
\put(3,1.5){\line(-1,0){0.11}}
\put(5,1.5){\line(1,0){0.11}}
\put(5,1.5){\line(-1,0){0.11}}
\put(6,1.5){\line(1,0){0.11}}
\put(6,1.5){\line(-1,0){0.11}}
\put(1,0.5){\line(1,0){0.11}}
\put(1,0.5){\line(-1,0){0.11}}
\put(2,1){\line(0,-1){1}}
\put(5,0.5){\line(1,0){0.11}}
\put(5,0.5){\line(-1,0){0.11}}
\put(6,0.5){\line(1,0){0.11}}
\put(6,0.5){\line(-1,0){0.11}}
\put(0.5,6){\line(0,1){0.11}}
\put(0.5,6){\line(0,-1){0.11}}
\put(0.5,5){\line(0,1){0.11}}
\put(0.5,5){\line(0,-1){0.11}}
\put(0.5,4){\line(0,1){0.11}}
\put(0.5,4){\line(0,-1){0.11}}
\put(0.5,3){\line(0,1){0.11}}
\put(0.5,3){\line(0,-1){0.11}}
\put(0.5,2){\line(0,1){0.11}}
\put(0.5,2){\line(0,-1){0.11}}
\put(0.5,1){\line(0,1){0.11}}
\put(0.5,1){\line(0,-1){0.11}}
\put(1.5,6){\line(0,1){0.11}}
\put(1.5,6){\line(0,-1){0.11}}
\put(1.5,5){\line(0,1){0.11}}
\put(1.5,5){\line(0,-1){0.11}}
\put(1.5,4){\line(0,1){0.11}}
\put(1.5,4){\line(0,-1){0.11}}
\put(1.5,3){\line(0,1){0.11}}
\put(1.5,3){\line(0,-1){0.11}}
\put(1.5,2){\line(0,1){0.11}}
\put(1.5,2){\line(0,-1){0.11}}
\put(1.5,1){\line(0,1){0.11}}
\put(1.5,1){\line(0,-1){0.11}}
\put(2.5,6){\line(0,1){0.11}}
\put(2.5,6){\line(0,-1){0.11}}
\put(2,5){\line(1,0){1}}
\put(2,4){\line(1,0){1}}
\put(2.5,3){\line(0,1){0.11}}
\put(2.5,3){\line(0,-1){0.11}}
\put(2.5,2){\line(0,1){0.11}}
\put(2.5,2){\line(0,-1){0.11}}
\put(2.5,1){\line(0,1){0.11}}
\put(2.5,1){\line(0,-1){0.11}}
\put(3.5,6){\line(0,1){0.11}}
\put(3.5,6){\line(0,-1){0.11}}
\put(3,5){\line(1,0){1}}
\put(3.5,4){\line(0,1){0.11}}
\put(3.5,4){\line(0,-1){0.11}}
\put(3,3){\line(1,0){1}}
\put(3.5,1){\line(0,1){0.11}}
\put(3.5,1){\line(0,-1){0.11}}
\put(4.5,6){\line(0,1){0.11}}
\put(4.5,6){\line(0,-1){0.11}}
\put(4,5){\line(1,0){1}}
\put(4.5,4){\line(0,1){0.11}}
\put(4.5,4){\line(0,-1){0.11}}
\put(4,3){\line(1,0){1}}
\put(4.5,2){\line(0,1){0.11}}
\put(4.5,2){\line(0,-1){0.11}}
\put(4.5,1){\line(0,1){0.11}}
\put(4.5,1){\line(0,-1){0.11}}
\put(5.5,6){\line(0,1){0.11}}
\put(5.5,6){\line(0,-1){0.11}}
\put(5,5){\line(1,0){1}}
\put(5.5,4){\line(0,1){0.11}}
\put(5.5,4){\line(0,-1){0.11}}
\put(5.5,3){\line(0,1){0.11}}
\put(5.5,3){\line(0,-1){0.11}}
\put(5,2){\line(1,0){1}}
\put(5.5,1){\line(0,1){0.11}}
\put(5.5,1){\line(0,-1){0.11}}
\put(6.5,6){\line(0,1){0.11}}
\put(6.5,6){\line(0,-1){0.11}}
\put(6.5,5){\line(0,1){0.11}}
\put(6.5,5){\line(0,-1){0.11}}
\put(6.5,4){\line(0,1){0.11}}
\put(6.5,4){\line(0,-1){0.11}}
\put(6,2){\line(1,0){1}}
\put(6.5,1){\line(0,1){0.11}}
\put(6.5,1){\line(0,-1){0.11}}
\end{picture}
    }
    &
    \subfigure[]{
      \setlength{\unitlength}{0.5cm}
\begin{picture}(7,7)(0,0)
\setcoordinatesystem units <0.5cm,0.5cm> point at 0 0
\setplotarea x from 0 to 7, y from 0 to 7
\put(0.5,0.5){\makebox(0,0)[c]{0}}
\put(0.5,0.5){\circle{0.8}}
\put(1.5,0.5){\makebox(0,0)[c]{2}}
\put(1.5,0.5){\circle{0.8}}
\put(2.5,0.5){\makebox(0,0)[c]{9}}
\put(2.5,0.5){\circle{0.8}}
\put(3.5,0.5){\makebox(0,0)[c]{3}}
\put(3.5,0.5){\circle{0.8}}
\put(4.5,0.5){\makebox(0,0)[c]{8}}
\put(4.5,0.5){\circle{0.8}}
\put(5.5,0.5){\makebox(0,0)[c]{5}}
\put(5.5,0.5){\circle{0.8}}
\put(6.5,0.5){\makebox(0,0)[c]{9}}
\put(6.5,0.5){\circle{0.8}}
\put(0.5,1.5){\makebox(0,0)[c]{1}}
\put(0.5,1.5){\circle{0.8}}
\put(1.5,1.5){\makebox(0,0)[c]{0}}
\put(1.5,1.5){\circle{0.8}}
\put(2.5,1.5){\makebox(0,0)[c]{8}}
\put(2.5,1.5){\circle{0.8}}
\put(3.5,1.5){\makebox(0,0)[c]{4}}
\put(3.5,1.5){\circle{0.8}}
\put(4.5,1.5){\makebox(0,0)[c]{9}}
\put(4.5,1.5){\circle{0.8}}
\put(5.5,1.5){\makebox(0,0)[c]{6}}
\put(5.5,1.5){\circle{0.8}}
\put(6.5,1.5){\makebox(0,0)[c]{7}}
\put(6.5,1.5){\circle{0.8}}
\put(0.5,2.5){\makebox(0,0)[c]{3}}
\put(0.5,2.5){\circle{0.8}}
\put(1.5,2.5){\makebox(0,0)[c]{2}}
\put(1.5,2.5){\circle{0.8}}
\put(2.5,2.5){\makebox(0,0)[c]{7}}
\put(2.5,2.5){\circle{0.8}}
\put(3.5,2.5){\makebox(0,0)[c]{9}}
\put(3.5,2.5){\circle{0.8}}
\put(4.5,2.5){\makebox(0,0)[c]{9}}
\put(4.5,2.5){\circle{0.8}}
\put(5.5,2.5){\makebox(0,0)[c]{1}}
\put(5.5,2.5){\circle{0.8}}
\put(6.5,2.5){\makebox(0,0)[c]{1}}
\put(6.5,2.5){\circle{0.8}}
\put(0.5,3.5){\makebox(0,0)[c]{1}}
\put(0.5,3.5){\circle{0.8}}
\put(1.5,3.5){\makebox(0,0)[c]{1}}
\put(1.5,3.5){\circle{0.8}}
\put(2.5,3.5){\makebox(0,0)[c]{9}}
\put(2.5,3.5){\circle{0.8}}
\put(3.5,3.5){\makebox(0,0)[c]{3}}
\put(3.5,3.5){\circle{0.8}}
\put(4.5,3.5){\makebox(0,0)[c]{4}}
\put(4.5,3.5){\circle{0.8}}
\put(5.5,3.5){\makebox(0,0)[c]{2}}
\put(5.5,3.5){\circle{0.8}}
\put(6.5,3.5){\makebox(0,0)[c]{6}}
\put(6.5,3.5){\circle{0.8}}
\put(0.5,4.5){\makebox(0,0)[c]{1}}
\put(0.5,4.5){\circle{0.8}}
\put(1.5,4.5){\makebox(0,0)[c]{0}}
\put(1.5,4.5){\circle{0.8}}
\put(2.5,4.5){\makebox(0,0)[c]{4}}
\put(2.5,4.5){\circle{0.8}}
\put(3.5,4.5){\makebox(0,0)[c]{1}}
\put(3.5,4.5){\circle{0.8}}
\put(4.5,4.5){\makebox(0,0)[c]{1}}
\put(4.5,4.5){\circle{0.8}}
\put(5.5,4.5){\makebox(0,0)[c]{2}}
\put(5.5,4.5){\circle{0.8}}
\put(6.5,4.5){\makebox(0,0)[c]{5}}
\put(6.5,4.5){\circle{0.8}}
\put(0.5,5.5){\makebox(0,0)[c]{2}}
\put(0.5,5.5){\circle{0.8}}
\put(1.5,5.5){\makebox(0,0)[c]{1}}
\put(1.5,5.5){\circle{0.8}}
\put(2.5,5.5){\makebox(0,0)[c]{9}}
\put(2.5,5.5){\circle{0.8}}
\put(3.5,5.5){\makebox(0,0)[c]{8}}
\put(3.5,5.5){\circle{0.8}}
\put(4.5,5.5){\makebox(0,0)[c]{8}}
\put(4.5,5.5){\circle{0.8}}
\put(5.5,5.5){\makebox(0,0)[c]{9}}
\put(5.5,5.5){\circle{0.8}}
\put(6.5,5.5){\makebox(0,0)[c]{1}}
\put(6.5,5.5){\circle{0.8}}
\put(0.5,6.5){\makebox(0,0)[c]{1}}
\put(0.5,6.5){\circle{0.8}}
\put(1.5,6.5){\makebox(0,0)[c]{3}}
\put(1.5,6.5){\circle{0.8}}
\put(2.5,6.5){\makebox(0,0)[c]{8}}
\put(2.5,6.5){\circle{0.8}}
\put(3.5,6.5){\makebox(0,0)[c]{7}}
\put(3.5,6.5){\circle{0.8}}
\put(4.5,6.5){\makebox(0,0)[c]{8}}
\put(4.5,6.5){\circle{0.8}}
\put(5.5,6.5){\makebox(0,0)[c]{8}}
\put(5.5,6.5){\circle{0.8}}
\put(6.5,6.5){\makebox(0,0)[c]{2}}
\put(6.5,6.5){\circle{0.8}}
\allinethickness{1.5pt}
\put(0,0){\line(1,0){7}}
\put(0,0){\line(0,1){7}}
\put(7,7){\line(-1,0){7}}
\put(7,7){\line(0,-1){7}}
\put(1,6.5){\line(1,0){0.11}}
\put(1,6.5){\line(-1,0){0.11}}
\put(2,6.5){\line(1,0){0.11}}
\put(2,6.5){\line(-1,0){0.11}}
\put(3,6.5){\line(1,0){0.11}}
\put(3,6.5){\line(-1,0){0.11}}
\put(4,6.5){\line(1,0){0.11}}
\put(4,6.5){\line(-1,0){0.11}}
\put(5,6.5){\line(1,0){0.11}}
\put(5,6.5){\line(-1,0){0.11}}
\put(1,5.5){\line(1,0){0.11}}
\put(1,5.5){\line(-1,0){0.11}}
\put(3,5.5){\line(1,0){0.11}}
\put(3,5.5){\line(-1,0){0.11}}
\put(4,5.5){\line(1,0){0.11}}
\put(4,5.5){\line(-1,0){0.11}}
\put(5,5.5){\line(1,0){0.11}}
\put(5,5.5){\line(-1,0){0.11}}
\put(1,4.5){\line(1,0){0.11}}
\put(1,4.5){\line(-1,0){0.11}}
\put(2,4.5){\line(1,0){0.11}}
\put(2,4.5){\line(-1,0){0.11}}
\put(3,4.5){\line(1,0){0.11}}
\put(3,4.5){\line(-1,0){0.11}}
\put(4,4.5){\line(1,0){0.11}}
\put(4,4.5){\line(-1,0){0.11}}
\put(5,4.5){\line(1,0){0.11}}
\put(5,4.5){\line(-1,0){0.11}}
\put(6,4.5){\line(1,0){0.11}}
\put(6,4.5){\line(-1,0){0.11}}
\put(1,3.5){\line(1,0){0.11}}
\put(1,3.5){\line(-1,0){0.11}}
\put(4,3.5){\line(1,0){0.11}}
\put(4,3.5){\line(-1,0){0.11}}
\put(5,3.5){\line(1,0){0.11}}
\put(5,3.5){\line(-1,0){0.11}}
\put(6,3.5){\line(1,0){0.11}}
\put(6,3.5){\line(-1,0){0.11}}
\put(1,2.5){\line(1,0){0.11}}
\put(1,2.5){\line(-1,0){0.11}}
\put(2,2.5){\line(1,0){0.11}}
\put(2,2.5){\line(-1,0){0.11}}
\put(3,2.5){\line(1,0){0.11}}
\put(3,2.5){\line(-1,0){0.11}}
\put(4,2.5){\line(1,0){0.11}}
\put(4,2.5){\line(-1,0){0.11}}
\put(6,2.5){\line(1,0){0.11}}
\put(6,2.5){\line(-1,0){0.11}}
\put(1,1.5){\line(1,0){0.11}}
\put(1,1.5){\line(-1,0){0.11}}
\put(3,1.5){\line(1,0){0.11}}
\put(3,1.5){\line(-1,0){0.11}}
\put(4,1.5){\line(1,0){0.11}}
\put(4,1.5){\line(-1,0){0.11}}
\put(5,1.5){\line(1,0){0.11}}
\put(5,1.5){\line(-1,0){0.11}}
\put(6,1.5){\line(1,0){0.11}}
\put(6,1.5){\line(-1,0){0.11}}
\put(1,0.5){\line(1,0){0.11}}
\put(1,0.5){\line(-1,0){0.11}}
\put(4,0.5){\line(1,0){0.11}}
\put(4,0.5){\line(-1,0){0.11}}
\put(5,0.5){\line(1,0){0.11}}
\put(5,0.5){\line(-1,0){0.11}}
\put(6,0.5){\line(1,0){0.11}}
\put(6,0.5){\line(-1,0){0.11}}
\put(0.5,6){\line(0,1){0.11}}
\put(0.5,6){\line(0,-1){0.11}}
\put(0.5,5){\line(0,1){0.11}}
\put(0.5,5){\line(0,-1){0.11}}
\put(0.5,4){\line(0,1){0.11}}
\put(0.5,4){\line(0,-1){0.11}}
\put(0.5,3){\line(0,1){0.11}}
\put(0.5,3){\line(0,-1){0.11}}
\put(0.5,2){\line(0,1){0.11}}
\put(0.5,2){\line(0,-1){0.11}}
\put(0.5,1){\line(0,1){0.11}}
\put(0.5,1){\line(0,-1){0.11}}
\put(1.5,6){\line(0,1){0.11}}
\put(1.5,6){\line(0,-1){0.11}}
\put(1.5,5){\line(0,1){0.11}}
\put(1.5,5){\line(0,-1){0.11}}
\put(1.5,4){\line(0,1){0.11}}
\put(1.5,4){\line(0,-1){0.11}}
\put(1.5,3){\line(0,1){0.11}}
\put(1.5,3){\line(0,-1){0.11}}
\put(1.5,2){\line(0,1){0.11}}
\put(1.5,2){\line(0,-1){0.11}}
\put(1.5,1){\line(0,1){0.11}}
\put(1.5,1){\line(0,-1){0.11}}
\put(2.5,6){\line(0,1){0.11}}
\put(2.5,6){\line(0,-1){0.11}}
\put(2.5,5){\line(0,1){0.11}}
\put(2.5,5){\line(0,-1){0.11}}
\put(2.5,4){\line(0,1){0.11}}
\put(2.5,4){\line(0,-1){0.11}}
\put(2.5,3){\line(0,1){0.11}}
\put(2.5,3){\line(0,-1){0.11}}
\put(2.5,2){\line(0,1){0.11}}
\put(2.5,2){\line(0,-1){0.11}}
\put(2.5,1){\line(0,1){0.11}}
\put(2.5,1){\line(0,-1){0.11}}
\put(3.5,6){\line(0,1){0.11}}
\put(3.5,6){\line(0,-1){0.11}}
\put(3.5,4){\line(0,1){0.11}}
\put(3.5,4){\line(0,-1){0.11}}
\put(3.5,2){\line(0,1){0.11}}
\put(3.5,2){\line(0,-1){0.11}}
\put(3.5,1){\line(0,1){0.11}}
\put(3.5,1){\line(0,-1){0.11}}
\put(4.5,6){\line(0,1){0.11}}
\put(4.5,6){\line(0,-1){0.11}}
\put(4.5,4){\line(0,1){0.11}}
\put(4.5,4){\line(0,-1){0.11}}
\put(4.5,3){\line(0,1){0.11}}
\put(4.5,3){\line(0,-1){0.11}}
\put(4.5,2){\line(0,1){0.11}}
\put(4.5,2){\line(0,-1){0.11}}
\put(4.5,1){\line(0,1){0.11}}
\put(4.5,1){\line(0,-1){0.11}}
\put(5.5,6){\line(0,1){0.11}}
\put(5.5,6){\line(0,-1){0.11}}
\put(5.5,4){\line(0,1){0.11}}
\put(5.5,4){\line(0,-1){0.11}}
\put(5.5,3){\line(0,1){0.11}}
\put(5.5,3){\line(0,-1){0.11}}
\put(5.5,2){\line(0,1){0.11}}
\put(5.5,2){\line(0,-1){0.11}}
\put(5.5,1){\line(0,1){0.11}}
\put(5.5,1){\line(0,-1){0.11}}
\put(6.5,6){\line(0,1){0.11}}
\put(6.5,6){\line(0,-1){0.11}}
\put(6.5,5){\line(0,1){0.11}}
\put(6.5,5){\line(0,-1){0.11}}
\put(6.5,4){\line(0,1){0.11}}
\put(6.5,4){\line(0,-1){0.11}}
\put(6.5,3){\line(0,1){0.11}}
\put(6.5,3){\line(0,-1){0.11}}
\put(6.5,1){\line(0,1){0.11}}
\put(6.5,1){\line(0,-1){0.11}}
\end{picture}
    }
  \end{tabular}
\end{center}
 \caption{Example from~\cite{Soille2007} of a 7x7 image 
 and its partitions into $\alpha$-connected components for $\alpha$
 ranging from 0 to 5. 
(a) 0-CCs. (b) 1-CCs. (c) 2-CCs. (d) 3-CCs.
(e) 4-CCs. (f) 5-CCs. 
}
 \label{fig:alphapami}
\end{figure*}
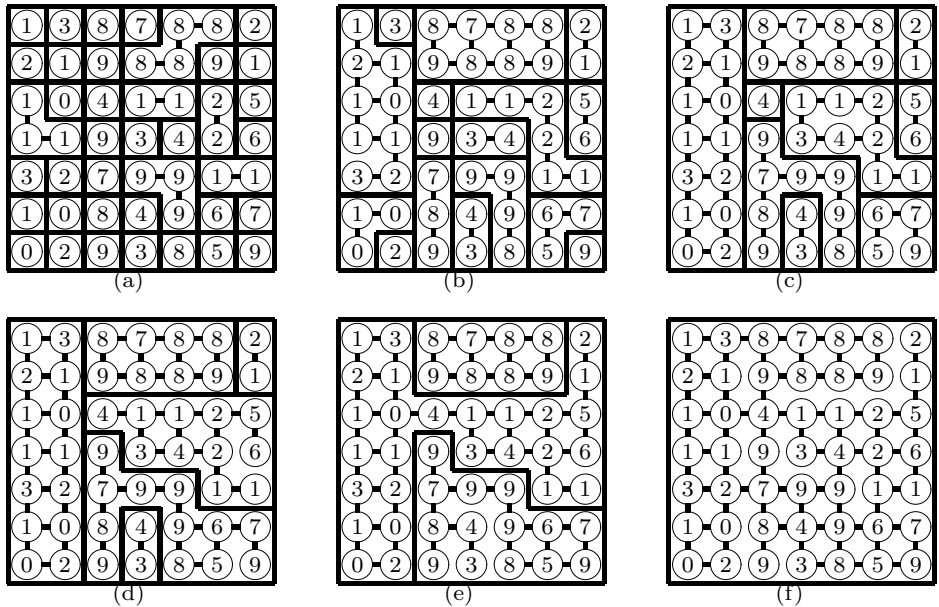

We now define the $(\alpha,\omega)$-connected component of an
arbitrary point $x$ as the largest $\alpha$-connected component of $x$
whose range is lower that $\omega$; 
more precisely,
\begin{eqnarray}
\label{eq:alphaomegaCC}
(\alpha,\omega)\mbox{-CC}(x)  = \sup\{ \beta\mbox{-CC}(x) & | & \beta\leq\alpha\mbox{~and~}\nonumber\\ 
& & R_f(\beta\mbox{-CC}(x))\leq\omega \}
\end{eqnarray}

The $(\alpha,\omega)$-CCs also define a hierarchy, that is called a
{\em constrained connectivity hierarchy}. We have:
\begin{equation}
(\alpha,\omega)\mbox{-CC}(x) \subseteq (\alpha',\omega')\mbox{-CC}(x)
\end{equation}
whenever $\alpha'\geq\alpha$ and $\forall \omega'\geq\omega$.  In
practice~\cite{Soille2007}, we are interested in this hierarchy for
$\alpha=\omega$, {\em i.e.}, for any $x\in V$ and any $\lambda\geq 0$,
we are looking for ($\lambda,\lambda$)-CC(x). 

Thus, informally, a hierarchy of $\alpha$-connected components is
given by connectivity relations constraining gray-level variations
along connected paths; a constrained connectivity hierarchy is given
by connectivity relations constraining gray-level variations both
along connected paths and within entire connected components.

An example of a constrained-connectivity hierarchy is given in
Fig.~\ref{fig:alphaomegapami}. 

\begin{figure*}[htbp]
\begin{center}
 \begin{tabular}{ccc}	
    \subfigure[]{
      \setlength{\unitlength}{0.5cm}
\begin{picture}(7,7)(0,0)
\setcoordinatesystem units <0.5cm,0.5cm> point at 0 0
\setplotarea x from 0 to 7, y from 0 to 7
\put(0.5,0.5){\makebox(0,0)[c]{0}}
\put(0.5,0.5){\circle{0.8}}
\put(1.5,0.5){\makebox(0,0)[c]{2}}
\put(1.5,0.5){\circle{0.8}}
\put(2.5,0.5){\makebox(0,0)[c]{9}}
\put(2.5,0.5){\circle{0.8}}
\put(3.5,0.5){\makebox(0,0)[c]{3}}
\put(3.5,0.5){\circle{0.8}}
\put(4.5,0.5){\makebox(0,0)[c]{8}}
\put(4.5,0.5){\circle{0.8}}
\put(5.5,0.5){\makebox(0,0)[c]{5}}
\put(5.5,0.5){\circle{0.8}}
\put(6.5,0.5){\makebox(0,0)[c]{9}}
\put(6.5,0.5){\circle{0.8}}
\put(0.5,1.5){\makebox(0,0)[c]{1}}
\put(0.5,1.5){\circle{0.8}}
\put(1.5,1.5){\makebox(0,0)[c]{0}}
\put(1.5,1.5){\circle{0.8}}
\put(2.5,1.5){\makebox(0,0)[c]{8}}
\put(2.5,1.5){\circle{0.8}}
\put(3.5,1.5){\makebox(0,0)[c]{4}}
\put(3.5,1.5){\circle{0.8}}
\put(4.5,1.5){\makebox(0,0)[c]{9}}
\put(4.5,1.5){\circle{0.8}}
\put(5.5,1.5){\makebox(0,0)[c]{6}}
\put(5.5,1.5){\circle{0.8}}
\put(6.5,1.5){\makebox(0,0)[c]{7}}
\put(6.5,1.5){\circle{0.8}}
\put(0.5,2.5){\makebox(0,0)[c]{3}}
\put(0.5,2.5){\circle{0.8}}
\put(1.5,2.5){\makebox(0,0)[c]{2}}
\put(1.5,2.5){\circle{0.8}}
\put(2.5,2.5){\makebox(0,0)[c]{7}}
\put(2.5,2.5){\circle{0.8}}
\put(3.5,2.5){\makebox(0,0)[c]{9}}
\put(3.5,2.5){\circle{0.8}}
\put(4.5,2.5){\makebox(0,0)[c]{9}}
\put(4.5,2.5){\circle{0.8}}
\put(5.5,2.5){\makebox(0,0)[c]{1}}
\put(5.5,2.5){\circle{0.8}}
\put(6.5,2.5){\makebox(0,0)[c]{1}}
\put(6.5,2.5){\circle{0.8}}
\put(0.5,3.5){\makebox(0,0)[c]{1}}
\put(0.5,3.5){\circle{0.8}}
\put(1.5,3.5){\makebox(0,0)[c]{1}}
\put(1.5,3.5){\circle{0.8}}
\put(2.5,3.5){\makebox(0,0)[c]{9}}
\put(2.5,3.5){\circle{0.8}}
\put(3.5,3.5){\makebox(0,0)[c]{3}}
\put(3.5,3.5){\circle{0.8}}
\put(4.5,3.5){\makebox(0,0)[c]{4}}
\put(4.5,3.5){\circle{0.8}}
\put(5.5,3.5){\makebox(0,0)[c]{2}}
\put(5.5,3.5){\circle{0.8}}
\put(6.5,3.5){\makebox(0,0)[c]{6}}
\put(6.5,3.5){\circle{0.8}}
\put(0.5,4.5){\makebox(0,0)[c]{1}}
\put(0.5,4.5){\circle{0.8}}
\put(1.5,4.5){\makebox(0,0)[c]{0}}
\put(1.5,4.5){\circle{0.8}}
\put(2.5,4.5){\makebox(0,0)[c]{4}}
\put(2.5,4.5){\circle{0.8}}
\put(3.5,4.5){\makebox(0,0)[c]{1}}
\put(3.5,4.5){\circle{0.8}}
\put(4.5,4.5){\makebox(0,0)[c]{1}}
\put(4.5,4.5){\circle{0.8}}
\put(5.5,4.5){\makebox(0,0)[c]{2}}
\put(5.5,4.5){\circle{0.8}}
\put(6.5,4.5){\makebox(0,0)[c]{5}}
\put(6.5,4.5){\circle{0.8}}
\put(0.5,5.5){\makebox(0,0)[c]{2}}
\put(0.5,5.5){\circle{0.8}}
\put(1.5,5.5){\makebox(0,0)[c]{1}}
\put(1.5,5.5){\circle{0.8}}
\put(2.5,5.5){\makebox(0,0)[c]{9}}
\put(2.5,5.5){\circle{0.8}}
\put(3.5,5.5){\makebox(0,0)[c]{8}}
\put(3.5,5.5){\circle{0.8}}
\put(4.5,5.5){\makebox(0,0)[c]{8}}
\put(4.5,5.5){\circle{0.8}}
\put(5.5,5.5){\makebox(0,0)[c]{9}}
\put(5.5,5.5){\circle{0.8}}
\put(6.5,5.5){\makebox(0,0)[c]{1}}
\put(6.5,5.5){\circle{0.8}}
\put(0.5,6.5){\makebox(0,0)[c]{1}}
\put(0.5,6.5){\circle{0.8}}
\put(1.5,6.5){\makebox(0,0)[c]{3}}
\put(1.5,6.5){\circle{0.8}}
\put(2.5,6.5){\makebox(0,0)[c]{8}}
\put(2.5,6.5){\circle{0.8}}
\put(3.5,6.5){\makebox(0,0)[c]{7}}
\put(3.5,6.5){\circle{0.8}}
\put(4.5,6.5){\makebox(0,0)[c]{8}}
\put(4.5,6.5){\circle{0.8}}
\put(5.5,6.5){\makebox(0,0)[c]{8}}
\put(5.5,6.5){\circle{0.8}}
\put(6.5,6.5){\makebox(0,0)[c]{2}}
\put(6.5,6.5){\circle{0.8}}
\allinethickness{1.5pt}
\put(0,0){\line(1,0){7}}
\put(0,0){\line(0,1){7}}
\put(7,7){\line(-1,0){7}}
\put(7,7){\line(0,-1){7}}
\put(1,7){\line(0,-1){1}}
\put(2,7){\line(0,-1){1}}
\put(3,7){\line(0,-1){1}}
\put(4,7){\line(0,-1){1}}
\put(5,6.5){\line(1,0){0.11}}
\put(5,6.5){\line(-1,0){0.11}}
\put(6,7){\line(0,-1){1}}
\put(1,6){\line(0,-1){1}}
\put(2,6){\line(0,-1){1}}
\put(3,6){\line(0,-1){1}}
\put(4,5.5){\line(1,0){0.11}}
\put(4,5.5){\line(-1,0){0.11}}
\put(5,6){\line(0,-1){1}}
\put(6,6){\line(0,-1){1}}
\put(1,5){\line(0,-1){1}}
\put(2,5){\line(0,-1){1}}
\put(3,5){\line(0,-1){1}}
\put(4,4.5){\line(1,0){0.11}}
\put(4,4.5){\line(-1,0){0.11}}
\put(5,4.5){\line(1,0){0.11}}
\put(5,4.5){\line(-1,0){0.11}}
\put(6,5){\line(0,-1){1}}
\put(1,3.5){\line(1,0){0.11}}
\put(1,3.5){\line(-1,0){0.11}}
\put(2,4){\line(0,-1){1}}
\put(3,4){\line(0,-1){1}}
\put(4,3.5){\line(1,0){0.11}}
\put(4,3.5){\line(-1,0){0.11}}
\put(5,4){\line(0,-1){1}}
\put(6,4){\line(0,-1){1}}
\put(1,3){\line(0,-1){1}}
\put(2,3){\line(0,-1){1}}
\put(3,3){\line(0,-1){1}}
\put(4,2.5){\line(1,0){0.11}}
\put(4,2.5){\line(-1,0){0.11}}
\put(5,3){\line(0,-1){1}}
\put(6,2.5){\line(1,0){0.11}}
\put(6,2.5){\line(-1,0){0.11}}
\put(1,1.5){\line(1,0){0.11}}
\put(1,1.5){\line(-1,0){0.11}}
\put(2,2){\line(0,-1){1}}
\put(3,2){\line(0,-1){1}}
\put(4,2){\line(0,-1){1}}
\put(5,2){\line(0,-1){1}}
\put(6,2){\line(0,-1){1}}
\put(1,1){\line(0,-1){1}}
\put(2,1){\line(0,-1){1}}
\put(3,1){\line(0,-1){1}}
\put(4,1){\line(0,-1){1}}
\put(5,1){\line(0,-1){1}}
\put(6,1){\line(0,-1){1}}
\put(0,6){\line(1,0){1}}
\put(0,5){\line(1,0){1}}
\put(0.5,4){\line(0,1){0.11}}
\put(0.5,4){\line(0,-1){0.11}}
\put(0,3){\line(1,0){1}}
\put(0,2){\line(1,0){1}}
\put(0.5,1){\line(0,1){0.11}}
\put(0.5,1){\line(0,-1){0.11}}
\put(1,6){\line(1,0){1}}
\put(1,5){\line(1,0){1}}
\put(1,4){\line(1,0){1}}
\put(1,3){\line(1,0){1}}
\put(1,2){\line(1,0){1}}
\put(1,1){\line(1,0){1}}
\put(2,6){\line(1,0){1}}
\put(2,5){\line(1,0){1}}
\put(2,4){\line(1,0){1}}
\put(2,3){\line(1,0){1}}
\put(2,2){\line(1,0){1}}
\put(2,1){\line(1,0){1}}
\put(3,6){\line(1,0){1}}
\put(3,5){\line(1,0){1}}
\put(3,4){\line(1,0){1}}
\put(3,3){\line(1,0){1}}
\put(3,2){\line(1,0){1}}
\put(3.5,1){\line(0,1){0.11}}
\put(3.5,1){\line(0,-1){0.11}}
\put(4.5,6){\line(0,1){0.11}}
\put(4.5,6){\line(0,-1){0.11}}
\put(4,5){\line(1,0){1}}
\put(4,4){\line(1,0){1}}
\put(4,3){\line(1,0){1}}
\put(4.5,2){\line(0,1){0.11}}
\put(4.5,2){\line(0,-1){0.11}}
\put(4.5,1){\line(0,1){0.11}}
\put(4.5,1){\line(0,-1){0.11}}
\put(5,6){\line(1,0){1}}
\put(5,5){\line(1,0){1}}
\put(5.5,4){\line(0,1){0.11}}
\put(5.5,4){\line(0,-1){0.11}}
\put(5.5,3){\line(0,1){0.11}}
\put(5.5,3){\line(0,-1){0.11}}
\put(5,2){\line(1,0){1}}
\put(5,1){\line(1,0){1}}
\put(6.5,6){\line(0,1){0.11}}
\put(6.5,6){\line(0,-1){0.11}}
\put(6,5){\line(1,0){1}}
\put(6.5,4){\line(0,1){0.11}}
\put(6.5,4){\line(0,-1){0.11}}
\put(6,3){\line(1,0){1}}
\put(6,2){\line(1,0){1}}
\put(6,1){\line(1,0){1}}
\end{picture}
    }
    &
    \subfigure[]{
      \setlength{\unitlength}{0.5cm}
\begin{picture}(7,7)(0,0)
\setcoordinatesystem units <0.5cm,0.5cm> point at 0 0
\setplotarea x from 0 to 7, y from 0 to 7
\put(0.5,0.5){\makebox(0,0)[c]{0}}
\put(0.5,0.5){\circle{0.8}}
\put(1.5,0.5){\makebox(0,0)[c]{2}}
\put(1.5,0.5){\circle{0.8}}
\put(2.5,0.5){\makebox(0,0)[c]{9}}
\put(2.5,0.5){\circle{0.8}}
\put(3.5,0.5){\makebox(0,0)[c]{3}}
\put(3.5,0.5){\circle{0.8}}
\put(4.5,0.5){\makebox(0,0)[c]{8}}
\put(4.5,0.5){\circle{0.8}}
\put(5.5,0.5){\makebox(0,0)[c]{5}}
\put(5.5,0.5){\circle{0.8}}
\put(6.5,0.5){\makebox(0,0)[c]{9}}
\put(6.5,0.5){\circle{0.8}}
\put(0.5,1.5){\makebox(0,0)[c]{1}}
\put(0.5,1.5){\circle{0.8}}
\put(1.5,1.5){\makebox(0,0)[c]{0}}
\put(1.5,1.5){\circle{0.8}}
\put(2.5,1.5){\makebox(0,0)[c]{8}}
\put(2.5,1.5){\circle{0.8}}
\put(3.5,1.5){\makebox(0,0)[c]{4}}
\put(3.5,1.5){\circle{0.8}}
\put(4.5,1.5){\makebox(0,0)[c]{9}}
\put(4.5,1.5){\circle{0.8}}
\put(5.5,1.5){\makebox(0,0)[c]{6}}
\put(5.5,1.5){\circle{0.8}}
\put(6.5,1.5){\makebox(0,0)[c]{7}}
\put(6.5,1.5){\circle{0.8}}
\put(0.5,2.5){\makebox(0,0)[c]{3}}
\put(0.5,2.5){\circle{0.8}}
\put(1.5,2.5){\makebox(0,0)[c]{2}}
\put(1.5,2.5){\circle{0.8}}
\put(2.5,2.5){\makebox(0,0)[c]{7}}
\put(2.5,2.5){\circle{0.8}}
\put(3.5,2.5){\makebox(0,0)[c]{9}}
\put(3.5,2.5){\circle{0.8}}
\put(4.5,2.5){\makebox(0,0)[c]{9}}
\put(4.5,2.5){\circle{0.8}}
\put(5.5,2.5){\makebox(0,0)[c]{1}}
\put(5.5,2.5){\circle{0.8}}
\put(6.5,2.5){\makebox(0,0)[c]{1}}
\put(6.5,2.5){\circle{0.8}}
\put(0.5,3.5){\makebox(0,0)[c]{1}}
\put(0.5,3.5){\circle{0.8}}
\put(1.5,3.5){\makebox(0,0)[c]{1}}
\put(1.5,3.5){\circle{0.8}}
\put(2.5,3.5){\makebox(0,0)[c]{9}}
\put(2.5,3.5){\circle{0.8}}
\put(3.5,3.5){\makebox(0,0)[c]{3}}
\put(3.5,3.5){\circle{0.8}}
\put(4.5,3.5){\makebox(0,0)[c]{4}}
\put(4.5,3.5){\circle{0.8}}
\put(5.5,3.5){\makebox(0,0)[c]{2}}
\put(5.5,3.5){\circle{0.8}}
\put(6.5,3.5){\makebox(0,0)[c]{6}}
\put(6.5,3.5){\circle{0.8}}
\put(0.5,4.5){\makebox(0,0)[c]{1}}
\put(0.5,4.5){\circle{0.8}}
\put(1.5,4.5){\makebox(0,0)[c]{0}}
\put(1.5,4.5){\circle{0.8}}
\put(2.5,4.5){\makebox(0,0)[c]{4}}
\put(2.5,4.5){\circle{0.8}}
\put(3.5,4.5){\makebox(0,0)[c]{1}}
\put(3.5,4.5){\circle{0.8}}
\put(4.5,4.5){\makebox(0,0)[c]{1}}
\put(4.5,4.5){\circle{0.8}}
\put(5.5,4.5){\makebox(0,0)[c]{2}}
\put(5.5,4.5){\circle{0.8}}
\put(6.5,4.5){\makebox(0,0)[c]{5}}
\put(6.5,4.5){\circle{0.8}}
\put(0.5,5.5){\makebox(0,0)[c]{2}}
\put(0.5,5.5){\circle{0.8}}
\put(1.5,5.5){\makebox(0,0)[c]{1}}
\put(1.5,5.5){\circle{0.8}}
\put(2.5,5.5){\makebox(0,0)[c]{9}}
\put(2.5,5.5){\circle{0.8}}
\put(3.5,5.5){\makebox(0,0)[c]{8}}
\put(3.5,5.5){\circle{0.8}}
\put(4.5,5.5){\makebox(0,0)[c]{8}}
\put(4.5,5.5){\circle{0.8}}
\put(5.5,5.5){\makebox(0,0)[c]{9}}
\put(5.5,5.5){\circle{0.8}}
\put(6.5,5.5){\makebox(0,0)[c]{1}}
\put(6.5,5.5){\circle{0.8}}
\put(0.5,6.5){\makebox(0,0)[c]{1}}
\put(0.5,6.5){\circle{0.8}}
\put(1.5,6.5){\makebox(0,0)[c]{3}}
\put(1.5,6.5){\circle{0.8}}
\put(2.5,6.5){\makebox(0,0)[c]{8}}
\put(2.5,6.5){\circle{0.8}}
\put(3.5,6.5){\makebox(0,0)[c]{7}}
\put(3.5,6.5){\circle{0.8}}
\put(4.5,6.5){\makebox(0,0)[c]{8}}
\put(4.5,6.5){\circle{0.8}}
\put(5.5,6.5){\makebox(0,0)[c]{8}}
\put(5.5,6.5){\circle{0.8}}
\put(6.5,6.5){\makebox(0,0)[c]{2}}
\put(6.5,6.5){\circle{0.8}}
\allinethickness{1.5pt}
\put(0,0){\line(1,0){7}}
\put(0,0){\line(0,1){7}}
\put(7,7){\line(-1,0){7}}
\put(7,7){\line(0,-1){7}}
\put(1,7){\line(0,-1){1}}
\put(2,7){\line(0,-1){1}}
\put(3,6.5){\line(1,0){0.11}}
\put(3,6.5){\line(-1,0){0.11}}
\put(4,6.5){\line(1,0){0.11}}
\put(4,6.5){\line(-1,0){0.11}}
\put(5,6.5){\line(1,0){0.11}}
\put(5,6.5){\line(-1,0){0.11}}
\put(6,7){\line(0,-1){1}}
\put(1,6){\line(0,-1){1}}
\put(2,6){\line(0,-1){1}}
\put(3,5.5){\line(1,0){0.11}}
\put(3,5.5){\line(-1,0){0.11}}
\put(4,5.5){\line(1,0){0.11}}
\put(4,5.5){\line(-1,0){0.11}}
\put(5,5.5){\line(1,0){0.11}}
\put(5,5.5){\line(-1,0){0.11}}
\put(6,6){\line(0,-1){1}}
\put(1,5){\line(0,-1){1}}
\put(2,5){\line(0,-1){1}}
\put(3,5){\line(0,-1){1}}
\put(4,4.5){\line(1,0){0.11}}
\put(4,4.5){\line(-1,0){0.11}}
\put(5,4.5){\line(1,0){0.11}}
\put(5,4.5){\line(-1,0){0.11}}
\put(6,5){\line(0,-1){1}}
\put(1,3.5){\line(1,0){0.11}}
\put(1,3.5){\line(-1,0){0.11}}
\put(2,4){\line(0,-1){1}}
\put(3,4){\line(0,-1){1}}
\put(4,3.5){\line(1,0){0.11}}
\put(4,3.5){\line(-1,0){0.11}}
\put(5,4){\line(0,-1){1}}
\put(6,4){\line(0,-1){1}}
\put(1,3){\line(0,-1){1}}
\put(2,3){\line(0,-1){1}}
\put(3,2.5){\line(1,0){0.11}}
\put(3,2.5){\line(-1,0){0.11}}
\put(4,2.5){\line(1,0){0.11}}
\put(4,2.5){\line(-1,0){0.11}}
\put(5,3){\line(0,-1){1}}
\put(6,2.5){\line(1,0){0.11}}
\put(6,2.5){\line(-1,0){0.11}}
\put(1,1.5){\line(1,0){0.11}}
\put(1,1.5){\line(-1,0){0.11}}
\put(2,2){\line(0,-1){1}}
\put(3,2){\line(0,-1){1}}
\put(4,2){\line(0,-1){1}}
\put(5,2){\line(0,-1){1}}
\put(6,1.5){\line(1,0){0.11}}
\put(6,1.5){\line(-1,0){0.11}}
\put(1,1){\line(0,-1){1}}
\put(2,1){\line(0,-1){1}}
\put(3,1){\line(0,-1){1}}
\put(4,1){\line(0,-1){1}}
\put(5,1){\line(0,-1){1}}
\put(6,1){\line(0,-1){1}}
\put(0,6){\line(1,0){1}}
\put(0,5){\line(1,0){1}}
\put(0.5,4){\line(0,1){0.11}}
\put(0.5,4){\line(0,-1){0.11}}
\put(0,3){\line(1,0){1}}
\put(0,2){\line(1,0){1}}
\put(0.5,1){\line(0,1){0.11}}
\put(0.5,1){\line(0,-1){0.11}}
\put(1,6){\line(1,0){1}}
\put(1,5){\line(1,0){1}}
\put(1,4){\line(1,0){1}}
\put(1,3){\line(1,0){1}}
\put(1,2){\line(1,0){1}}
\put(1,1){\line(1,0){1}}
\put(2.5,6){\line(0,1){0.11}}
\put(2.5,6){\line(0,-1){0.11}}
\put(2,5){\line(1,0){1}}
\put(2,4){\line(1,0){1}}
\put(2.5,3){\line(0,1){0.11}}
\put(2.5,3){\line(0,-1){0.11}}
\put(2.5,2){\line(0,1){0.11}}
\put(2.5,2){\line(0,-1){0.11}}
\put(2.5,1){\line(0,1){0.11}}
\put(2.5,1){\line(0,-1){0.11}}
\put(3.5,6){\line(0,1){0.11}}
\put(3.5,6){\line(0,-1){0.11}}
\put(3,5){\line(1,0){1}}
\put(3,4){\line(1,0){1}}
\put(3,3){\line(1,0){1}}
\put(3,2){\line(1,0){1}}
\put(3.5,1){\line(0,1){0.11}}
\put(3.5,1){\line(0,-1){0.11}}
\put(4.5,6){\line(0,1){0.11}}
\put(4.5,6){\line(0,-1){0.11}}
\put(4,5){\line(1,0){1}}
\put(4,4){\line(1,0){1}}
\put(4,3){\line(1,0){1}}
\put(4.5,2){\line(0,1){0.11}}
\put(4.5,2){\line(0,-1){0.11}}
\put(4.5,1){\line(0,1){0.11}}
\put(4.5,1){\line(0,-1){0.11}}
\put(5.5,6){\line(0,1){0.11}}
\put(5.5,6){\line(0,-1){0.11}}
\put(5,5){\line(1,0){1}}
\put(5.5,4){\line(0,1){0.11}}
\put(5.5,4){\line(0,-1){0.11}}
\put(5.5,3){\line(0,1){0.11}}
\put(5.5,3){\line(0,-1){0.11}}
\put(5,2){\line(1,0){1}}
\put(5.5,1){\line(0,1){0.11}}
\put(5.5,1){\line(0,-1){0.11}}
\put(6.5,6){\line(0,1){0.11}}
\put(6.5,6){\line(0,-1){0.11}}
\put(6,5){\line(1,0){1}}
\put(6.5,4){\line(0,1){0.11}}
\put(6.5,4){\line(0,-1){0.11}}
\put(6,3){\line(1,0){1}}
\put(6,2){\line(1,0){1}}
\put(6,1){\line(1,0){1}}
\end{picture}
    }
    &
    \subfigure[]{
      \setlength{\unitlength}{0.5cm}
\begin{picture}(7,7)(0,0)
\setcoordinatesystem units <0.5cm,0.5cm> point at 0 0
\setplotarea x from 0 to 7, y from 0 to 7
\put(0.5,0.5){\makebox(0,0)[c]{0}}
\put(0.5,0.5){\circle{0.8}}
\put(1.5,0.5){\makebox(0,0)[c]{2}}
\put(1.5,0.5){\circle{0.8}}
\put(2.5,0.5){\makebox(0,0)[c]{9}}
\put(2.5,0.5){\circle{0.8}}
\put(3.5,0.5){\makebox(0,0)[c]{3}}
\put(3.5,0.5){\circle{0.8}}
\put(4.5,0.5){\makebox(0,0)[c]{8}}
\put(4.5,0.5){\circle{0.8}}
\put(5.5,0.5){\makebox(0,0)[c]{5}}
\put(5.5,0.5){\circle{0.8}}
\put(6.5,0.5){\makebox(0,0)[c]{9}}
\put(6.5,0.5){\circle{0.8}}
\put(0.5,1.5){\makebox(0,0)[c]{1}}
\put(0.5,1.5){\circle{0.8}}
\put(1.5,1.5){\makebox(0,0)[c]{0}}
\put(1.5,1.5){\circle{0.8}}
\put(2.5,1.5){\makebox(0,0)[c]{8}}
\put(2.5,1.5){\circle{0.8}}
\put(3.5,1.5){\makebox(0,0)[c]{4}}
\put(3.5,1.5){\circle{0.8}}
\put(4.5,1.5){\makebox(0,0)[c]{9}}
\put(4.5,1.5){\circle{0.8}}
\put(5.5,1.5){\makebox(0,0)[c]{6}}
\put(5.5,1.5){\circle{0.8}}
\put(6.5,1.5){\makebox(0,0)[c]{7}}
\put(6.5,1.5){\circle{0.8}}
\put(0.5,2.5){\makebox(0,0)[c]{3}}
\put(0.5,2.5){\circle{0.8}}
\put(1.5,2.5){\makebox(0,0)[c]{2}}
\put(1.5,2.5){\circle{0.8}}
\put(2.5,2.5){\makebox(0,0)[c]{7}}
\put(2.5,2.5){\circle{0.8}}
\put(3.5,2.5){\makebox(0,0)[c]{9}}
\put(3.5,2.5){\circle{0.8}}
\put(4.5,2.5){\makebox(0,0)[c]{9}}
\put(4.5,2.5){\circle{0.8}}
\put(5.5,2.5){\makebox(0,0)[c]{1}}
\put(5.5,2.5){\circle{0.8}}
\put(6.5,2.5){\makebox(0,0)[c]{1}}
\put(6.5,2.5){\circle{0.8}}
\put(0.5,3.5){\makebox(0,0)[c]{1}}
\put(0.5,3.5){\circle{0.8}}
\put(1.5,3.5){\makebox(0,0)[c]{1}}
\put(1.5,3.5){\circle{0.8}}
\put(2.5,3.5){\makebox(0,0)[c]{9}}
\put(2.5,3.5){\circle{0.8}}
\put(3.5,3.5){\makebox(0,0)[c]{3}}
\put(3.5,3.5){\circle{0.8}}
\put(4.5,3.5){\makebox(0,0)[c]{4}}
\put(4.5,3.5){\circle{0.8}}
\put(5.5,3.5){\makebox(0,0)[c]{2}}
\put(5.5,3.5){\circle{0.8}}
\put(6.5,3.5){\makebox(0,0)[c]{6}}
\put(6.5,3.5){\circle{0.8}}
\put(0.5,4.5){\makebox(0,0)[c]{1}}
\put(0.5,4.5){\circle{0.8}}
\put(1.5,4.5){\makebox(0,0)[c]{0}}
\put(1.5,4.5){\circle{0.8}}
\put(2.5,4.5){\makebox(0,0)[c]{4}}
\put(2.5,4.5){\circle{0.8}}
\put(3.5,4.5){\makebox(0,0)[c]{1}}
\put(3.5,4.5){\circle{0.8}}
\put(4.5,4.5){\makebox(0,0)[c]{1}}
\put(4.5,4.5){\circle{0.8}}
\put(5.5,4.5){\makebox(0,0)[c]{2}}
\put(5.5,4.5){\circle{0.8}}
\put(6.5,4.5){\makebox(0,0)[c]{5}}
\put(6.5,4.5){\circle{0.8}}
\put(0.5,5.5){\makebox(0,0)[c]{2}}
\put(0.5,5.5){\circle{0.8}}
\put(1.5,5.5){\makebox(0,0)[c]{1}}
\put(1.5,5.5){\circle{0.8}}
\put(2.5,5.5){\makebox(0,0)[c]{9}}
\put(2.5,5.5){\circle{0.8}}
\put(3.5,5.5){\makebox(0,0)[c]{8}}
\put(3.5,5.5){\circle{0.8}}
\put(4.5,5.5){\makebox(0,0)[c]{8}}
\put(4.5,5.5){\circle{0.8}}
\put(5.5,5.5){\makebox(0,0)[c]{9}}
\put(5.5,5.5){\circle{0.8}}
\put(6.5,5.5){\makebox(0,0)[c]{1}}
\put(6.5,5.5){\circle{0.8}}
\put(0.5,6.5){\makebox(0,0)[c]{1}}
\put(0.5,6.5){\circle{0.8}}
\put(1.5,6.5){\makebox(0,0)[c]{3}}
\put(1.5,6.5){\circle{0.8}}
\put(2.5,6.5){\makebox(0,0)[c]{8}}
\put(2.5,6.5){\circle{0.8}}
\put(3.5,6.5){\makebox(0,0)[c]{7}}
\put(3.5,6.5){\circle{0.8}}
\put(4.5,6.5){\makebox(0,0)[c]{8}}
\put(4.5,6.5){\circle{0.8}}
\put(5.5,6.5){\makebox(0,0)[c]{8}}
\put(5.5,6.5){\circle{0.8}}
\put(6.5,6.5){\makebox(0,0)[c]{2}}
\put(6.5,6.5){\circle{0.8}}
\allinethickness{1.5pt}
\put(0,0){\line(1,0){7}}
\put(0,0){\line(0,1){7}}
\put(7,7){\line(-1,0){7}}
\put(7,7){\line(0,-1){7}}
\put(1,6.5){\line(1,0){0.11}}
\put(1,6.5){\line(-1,0){0.11}}
\put(2,7){\line(0,-1){1}}
\put(3,6.5){\line(1,0){0.11}}
\put(3,6.5){\line(-1,0){0.11}}
\put(4,6.5){\line(1,0){0.11}}
\put(4,6.5){\line(-1,0){0.11}}
\put(5,6.5){\line(1,0){0.11}}
\put(5,6.5){\line(-1,0){0.11}}
\put(6,7){\line(0,-1){1}}
\put(1,5.5){\line(1,0){0.11}}
\put(1,5.5){\line(-1,0){0.11}}
\put(2,6){\line(0,-1){1}}
\put(3,5.5){\line(1,0){0.11}}
\put(3,5.5){\line(-1,0){0.11}}
\put(4,5.5){\line(1,0){0.11}}
\put(4,5.5){\line(-1,0){0.11}}
\put(5,5.5){\line(1,0){0.11}}
\put(5,5.5){\line(-1,0){0.11}}
\put(6,6){\line(0,-1){1}}
\put(1,4.5){\line(1,0){0.11}}
\put(1,4.5){\line(-1,0){0.11}}
\put(2,5){\line(0,-1){1}}
\put(3,5){\line(0,-1){1}}
\put(4,4.5){\line(1,0){0.11}}
\put(4,4.5){\line(-1,0){0.11}}
\put(5,4.5){\line(1,0){0.11}}
\put(5,4.5){\line(-1,0){0.11}}
\put(6,5){\line(0,-1){1}}
\put(1,3.5){\line(1,0){0.11}}
\put(1,3.5){\line(-1,0){0.11}}
\put(2,4){\line(0,-1){1}}
\put(3,4){\line(0,-1){1}}
\put(4,3.5){\line(1,0){0.11}}
\put(4,3.5){\line(-1,0){0.11}}
\put(5,3.5){\line(1,0){0.11}}
\put(5,3.5){\line(-1,0){0.11}}
\put(6,4){\line(0,-1){1}}
\put(1,2.5){\line(1,0){0.11}}
\put(1,2.5){\line(-1,0){0.11}}
\put(2,3){\line(0,-1){1}}
\put(3,2.5){\line(1,0){0.11}}
\put(3,2.5){\line(-1,0){0.11}}
\put(4,2.5){\line(1,0){0.11}}
\put(4,2.5){\line(-1,0){0.11}}
\put(5,3){\line(0,-1){1}}
\put(6,2.5){\line(1,0){0.11}}
\put(6,2.5){\line(-1,0){0.11}}
\put(1,1.5){\line(1,0){0.11}}
\put(1,1.5){\line(-1,0){0.11}}
\put(2,2){\line(0,-1){1}}
\put(3,2){\line(0,-1){1}}
\put(4,2){\line(0,-1){1}}
\put(5,2){\line(0,-1){1}}
\put(6,1.5){\line(1,0){0.11}}
\put(6,1.5){\line(-1,0){0.11}}
\put(1,0.5){\line(1,0){0.11}}
\put(1,0.5){\line(-1,0){0.11}}
\put(2,1){\line(0,-1){1}}
\put(3,1){\line(0,-1){1}}
\put(4,1){\line(0,-1){1}}
\put(5,1){\line(0,-1){1}}
\put(6,1){\line(0,-1){1}}
\put(0.5,6){\line(0,1){0.11}}
\put(0.5,6){\line(0,-1){0.11}}
\put(0.5,5){\line(0,1){0.11}}
\put(0.5,5){\line(0,-1){0.11}}
\put(0.5,4){\line(0,1){0.11}}
\put(0.5,4){\line(0,-1){0.11}}
\put(0.5,3){\line(0,1){0.11}}
\put(0.5,3){\line(0,-1){0.11}}
\put(0.5,2){\line(0,1){0.11}}
\put(0.5,2){\line(0,-1){0.11}}
\put(0.5,1){\line(0,1){0.11}}
\put(0.5,1){\line(0,-1){0.11}}
\put(1.5,6){\line(0,1){0.11}}
\put(1.5,6){\line(0,-1){0.11}}
\put(1.5,5){\line(0,1){0.11}}
\put(1.5,5){\line(0,-1){0.11}}
\put(1.5,4){\line(0,1){0.11}}
\put(1.5,4){\line(0,-1){0.11}}
\put(1.5,3){\line(0,1){0.11}}
\put(1.5,3){\line(0,-1){0.11}}
\put(1.5,2){\line(0,1){0.11}}
\put(1.5,2){\line(0,-1){0.11}}
\put(1.5,1){\line(0,1){0.11}}
\put(1.5,1){\line(0,-1){0.11}}
\put(2.5,6){\line(0,1){0.11}}
\put(2.5,6){\line(0,-1){0.11}}
\put(2,5){\line(1,0){1}}
\put(2,4){\line(1,0){1}}
\put(2.5,3){\line(0,1){0.11}}
\put(2.5,3){\line(0,-1){0.11}}
\put(2.5,2){\line(0,1){0.11}}
\put(2.5,2){\line(0,-1){0.11}}
\put(2.5,1){\line(0,1){0.11}}
\put(2.5,1){\line(0,-1){0.11}}
\put(3.5,6){\line(0,1){0.11}}
\put(3.5,6){\line(0,-1){0.11}}
\put(3,5){\line(1,0){1}}
\put(3.5,4){\line(0,1){0.11}}
\put(3.5,4){\line(0,-1){0.11}}
\put(3,3){\line(1,0){1}}
\put(3,2){\line(1,0){1}}
\put(3.5,1){\line(0,1){0.11}}
\put(3.5,1){\line(0,-1){0.11}}
\put(4.5,6){\line(0,1){0.11}}
\put(4.5,6){\line(0,-1){0.11}}
\put(4,5){\line(1,0){1}}
\put(4,3){\line(1,0){1}}
\put(4.5,2){\line(0,1){0.11}}
\put(4.5,2){\line(0,-1){0.11}}
\put(4.5,1){\line(0,1){0.11}}
\put(4.5,1){\line(0,-1){0.11}}
\put(5.5,6){\line(0,1){0.11}}
\put(5.5,6){\line(0,-1){0.11}}
\put(5,5){\line(1,0){1}}
\put(5.5,4){\line(0,1){0.11}}
\put(5.5,4){\line(0,-1){0.11}}
\put(5.5,3){\line(0,1){0.11}}
\put(5.5,3){\line(0,-1){0.11}}
\put(5,2){\line(1,0){1}}
\put(5.5,1){\line(0,1){0.11}}
\put(5.5,1){\line(0,-1){0.11}}
\put(6.5,6){\line(0,1){0.11}}
\put(6.5,6){\line(0,-1){0.11}}
\put(6,5){\line(1,0){1}}
\put(6.5,4){\line(0,1){0.11}}
\put(6.5,4){\line(0,-1){0.11}}
\put(6,3){\line(1,0){1}}
\put(6,2){\line(1,0){1}}
\put(6,1){\line(1,0){1}}
\end{picture}
    }
\\
    \subfigure[]{
      \setlength{\unitlength}{0.5cm}
\begin{picture}(7,7)(0,0)
\setcoordinatesystem units <0.5cm,0.5cm> point at 0 0
\setplotarea x from 0 to 7, y from 0 to 7
\put(0.5,0.5){\makebox(0,0)[c]{0}}
\put(0.5,0.5){\circle{0.8}}
\put(1.5,0.5){\makebox(0,0)[c]{2}}
\put(1.5,0.5){\circle{0.8}}
\put(2.5,0.5){\makebox(0,0)[c]{9}}
\put(2.5,0.5){\circle{0.8}}
\put(3.5,0.5){\makebox(0,0)[c]{3}}
\put(3.5,0.5){\circle{0.8}}
\put(4.5,0.5){\makebox(0,0)[c]{8}}
\put(4.5,0.5){\circle{0.8}}
\put(5.5,0.5){\makebox(0,0)[c]{5}}
\put(5.5,0.5){\circle{0.8}}
\put(6.5,0.5){\makebox(0,0)[c]{9}}
\put(6.5,0.5){\circle{0.8}}
\put(0.5,1.5){\makebox(0,0)[c]{1}}
\put(0.5,1.5){\circle{0.8}}
\put(1.5,1.5){\makebox(0,0)[c]{0}}
\put(1.5,1.5){\circle{0.8}}
\put(2.5,1.5){\makebox(0,0)[c]{8}}
\put(2.5,1.5){\circle{0.8}}
\put(3.5,1.5){\makebox(0,0)[c]{4}}
\put(3.5,1.5){\circle{0.8}}
\put(4.5,1.5){\makebox(0,0)[c]{9}}
\put(4.5,1.5){\circle{0.8}}
\put(5.5,1.5){\makebox(0,0)[c]{6}}
\put(5.5,1.5){\circle{0.8}}
\put(6.5,1.5){\makebox(0,0)[c]{7}}
\put(6.5,1.5){\circle{0.8}}
\put(0.5,2.5){\makebox(0,0)[c]{3}}
\put(0.5,2.5){\circle{0.8}}
\put(1.5,2.5){\makebox(0,0)[c]{2}}
\put(1.5,2.5){\circle{0.8}}
\put(2.5,2.5){\makebox(0,0)[c]{7}}
\put(2.5,2.5){\circle{0.8}}
\put(3.5,2.5){\makebox(0,0)[c]{9}}
\put(3.5,2.5){\circle{0.8}}
\put(4.5,2.5){\makebox(0,0)[c]{9}}
\put(4.5,2.5){\circle{0.8}}
\put(5.5,2.5){\makebox(0,0)[c]{1}}
\put(5.5,2.5){\circle{0.8}}
\put(6.5,2.5){\makebox(0,0)[c]{1}}
\put(6.5,2.5){\circle{0.8}}
\put(0.5,3.5){\makebox(0,0)[c]{1}}
\put(0.5,3.5){\circle{0.8}}
\put(1.5,3.5){\makebox(0,0)[c]{1}}
\put(1.5,3.5){\circle{0.8}}
\put(2.5,3.5){\makebox(0,0)[c]{9}}
\put(2.5,3.5){\circle{0.8}}
\put(3.5,3.5){\makebox(0,0)[c]{3}}
\put(3.5,3.5){\circle{0.8}}
\put(4.5,3.5){\makebox(0,0)[c]{4}}
\put(4.5,3.5){\circle{0.8}}
\put(5.5,3.5){\makebox(0,0)[c]{2}}
\put(5.5,3.5){\circle{0.8}}
\put(6.5,3.5){\makebox(0,0)[c]{6}}
\put(6.5,3.5){\circle{0.8}}
\put(0.5,4.5){\makebox(0,0)[c]{1}}
\put(0.5,4.5){\circle{0.8}}
\put(1.5,4.5){\makebox(0,0)[c]{0}}
\put(1.5,4.5){\circle{0.8}}
\put(2.5,4.5){\makebox(0,0)[c]{4}}
\put(2.5,4.5){\circle{0.8}}
\put(3.5,4.5){\makebox(0,0)[c]{1}}
\put(3.5,4.5){\circle{0.8}}
\put(4.5,4.5){\makebox(0,0)[c]{1}}
\put(4.5,4.5){\circle{0.8}}
\put(5.5,4.5){\makebox(0,0)[c]{2}}
\put(5.5,4.5){\circle{0.8}}
\put(6.5,4.5){\makebox(0,0)[c]{5}}
\put(6.5,4.5){\circle{0.8}}
\put(0.5,5.5){\makebox(0,0)[c]{2}}
\put(0.5,5.5){\circle{0.8}}
\put(1.5,5.5){\makebox(0,0)[c]{1}}
\put(1.5,5.5){\circle{0.8}}
\put(2.5,5.5){\makebox(0,0)[c]{9}}
\put(2.5,5.5){\circle{0.8}}
\put(3.5,5.5){\makebox(0,0)[c]{8}}
\put(3.5,5.5){\circle{0.8}}
\put(4.5,5.5){\makebox(0,0)[c]{8}}
\put(4.5,5.5){\circle{0.8}}
\put(5.5,5.5){\makebox(0,0)[c]{9}}
\put(5.5,5.5){\circle{0.8}}
\put(6.5,5.5){\makebox(0,0)[c]{1}}
\put(6.5,5.5){\circle{0.8}}
\put(0.5,6.5){\makebox(0,0)[c]{1}}
\put(0.5,6.5){\circle{0.8}}
\put(1.5,6.5){\makebox(0,0)[c]{3}}
\put(1.5,6.5){\circle{0.8}}
\put(2.5,6.5){\makebox(0,0)[c]{8}}
\put(2.5,6.5){\circle{0.8}}
\put(3.5,6.5){\makebox(0,0)[c]{7}}
\put(3.5,6.5){\circle{0.8}}
\put(4.5,6.5){\makebox(0,0)[c]{8}}
\put(4.5,6.5){\circle{0.8}}
\put(5.5,6.5){\makebox(0,0)[c]{8}}
\put(5.5,6.5){\circle{0.8}}
\put(6.5,6.5){\makebox(0,0)[c]{2}}
\put(6.5,6.5){\circle{0.8}}
\allinethickness{1.5pt}
\put(0,0){\line(1,0){7}}
\put(0,0){\line(0,1){7}}
\put(7,7){\line(-1,0){7}}
\put(7,7){\line(0,-1){7}}
\put(1,6.5){\line(1,0){0.11}}
\put(1,6.5){\line(-1,0){0.11}}
\put(2,7){\line(0,-1){1}}
\put(3,6.5){\line(1,0){0.11}}
\put(3,6.5){\line(-1,0){0.11}}
\put(4,6.5){\line(1,0){0.11}}
\put(4,6.5){\line(-1,0){0.11}}
\put(5,6.5){\line(1,0){0.11}}
\put(5,6.5){\line(-1,0){0.11}}
\put(6,7){\line(0,-1){1}}
\put(1,5.5){\line(1,0){0.11}}
\put(1,5.5){\line(-1,0){0.11}}
\put(2,6){\line(0,-1){1}}
\put(3,5.5){\line(1,0){0.11}}
\put(3,5.5){\line(-1,0){0.11}}
\put(4,5.5){\line(1,0){0.11}}
\put(4,5.5){\line(-1,0){0.11}}
\put(5,5.5){\line(1,0){0.11}}
\put(5,5.5){\line(-1,0){0.11}}
\put(6,6){\line(0,-1){1}}
\put(1,4.5){\line(1,0){0.11}}
\put(1,4.5){\line(-1,0){0.11}}
\put(2,5){\line(0,-1){1}}
\put(3,5){\line(0,-1){1}}
\put(4,4.5){\line(1,0){0.11}}
\put(4,4.5){\line(-1,0){0.11}}
\put(5,4.5){\line(1,0){0.11}}
\put(5,4.5){\line(-1,0){0.11}}
\put(6,5){\line(0,-1){1}}
\put(1,3.5){\line(1,0){0.11}}
\put(1,3.5){\line(-1,0){0.11}}
\put(2,4){\line(0,-1){1}}
\put(3,4){\line(0,-1){1}}
\put(4,3.5){\line(1,0){0.11}}
\put(4,3.5){\line(-1,0){0.11}}
\put(5,3.5){\line(1,0){0.11}}
\put(5,3.5){\line(-1,0){0.11}}
\put(6,4){\line(0,-1){1}}
\put(1,2.5){\line(1,0){0.11}}
\put(1,2.5){\line(-1,0){0.11}}
\put(2,3){\line(0,-1){1}}
\put(3,2.5){\line(1,0){0.11}}
\put(3,2.5){\line(-1,0){0.11}}
\put(4,2.5){\line(1,0){0.11}}
\put(4,2.5){\line(-1,0){0.11}}
\put(5,3){\line(0,-1){1}}
\put(6,2.5){\line(1,0){0.11}}
\put(6,2.5){\line(-1,0){0.11}}
\put(1,1.5){\line(1,0){0.11}}
\put(1,1.5){\line(-1,0){0.11}}
\put(2,2){\line(0,-1){1}}
\put(3,2){\line(0,-1){1}}
\put(4,2){\line(0,-1){1}}
\put(5,1.5){\line(1,0){0.11}}
\put(5,1.5){\line(-1,0){0.11}}
\put(6,1.5){\line(1,0){0.11}}
\put(6,1.5){\line(-1,0){0.11}}
\put(1,0.5){\line(1,0){0.11}}
\put(1,0.5){\line(-1,0){0.11}}
\put(2,1){\line(0,-1){1}}
\put(3,1){\line(0,-1){1}}
\put(4,1){\line(0,-1){1}}
\put(5,0.5){\line(1,0){0.11}}
\put(5,0.5){\line(-1,0){0.11}}
\put(0.5,6){\line(0,1){0.11}}
\put(0.5,6){\line(0,-1){0.11}}
\put(0.5,5){\line(0,1){0.11}}
\put(0.5,5){\line(0,-1){0.11}}
\put(0.5,4){\line(0,1){0.11}}
\put(0.5,4){\line(0,-1){0.11}}
\put(0.5,3){\line(0,1){0.11}}
\put(0.5,3){\line(0,-1){0.11}}
\put(0.5,2){\line(0,1){0.11}}
\put(0.5,2){\line(0,-1){0.11}}
\put(0.5,1){\line(0,1){0.11}}
\put(0.5,1){\line(0,-1){0.11}}
\put(1.5,6){\line(0,1){0.11}}
\put(1.5,6){\line(0,-1){0.11}}
\put(1.5,5){\line(0,1){0.11}}
\put(1.5,5){\line(0,-1){0.11}}
\put(1.5,4){\line(0,1){0.11}}
\put(1.5,4){\line(0,-1){0.11}}
\put(1.5,3){\line(0,1){0.11}}
\put(1.5,3){\line(0,-1){0.11}}
\put(1.5,2){\line(0,1){0.11}}
\put(1.5,2){\line(0,-1){0.11}}
\put(1.5,1){\line(0,1){0.11}}
\put(1.5,1){\line(0,-1){0.11}}
\put(2.5,6){\line(0,1){0.11}}
\put(2.5,6){\line(0,-1){0.11}}
\put(2,5){\line(1,0){1}}
\put(2,4){\line(1,0){1}}
\put(2.5,3){\line(0,1){0.11}}
\put(2.5,3){\line(0,-1){0.11}}
\put(2.5,2){\line(0,1){0.11}}
\put(2.5,2){\line(0,-1){0.11}}
\put(2.5,1){\line(0,1){0.11}}
\put(2.5,1){\line(0,-1){0.11}}
\put(3.5,6){\line(0,1){0.11}}
\put(3.5,6){\line(0,-1){0.11}}
\put(3,5){\line(1,0){1}}
\put(3.5,4){\line(0,1){0.11}}
\put(3.5,4){\line(0,-1){0.11}}
\put(3,3){\line(1,0){1}}
\put(3,2){\line(1,0){1}}
\put(3.5,1){\line(0,1){0.11}}
\put(3.5,1){\line(0,-1){0.11}}
\put(4.5,6){\line(0,1){0.11}}
\put(4.5,6){\line(0,-1){0.11}}
\put(4,5){\line(1,0){1}}
\put(4,3){\line(1,0){1}}
\put(4.5,2){\line(0,1){0.11}}
\put(4.5,2){\line(0,-1){0.11}}
\put(4.5,1){\line(0,1){0.11}}
\put(4.5,1){\line(0,-1){0.11}}
\put(5.5,6){\line(0,1){0.11}}
\put(5.5,6){\line(0,-1){0.11}}
\put(5,5){\line(1,0){1}}
\put(5.5,4){\line(0,1){0.11}}
\put(5.5,4){\line(0,-1){0.11}}
\put(5.5,3){\line(0,1){0.11}}
\put(5.5,3){\line(0,-1){0.11}}
\put(5,2){\line(1,0){1}}
\put(5.5,1){\line(0,1){0.11}}
\put(5.5,1){\line(0,-1){0.11}}
\put(6.5,6){\line(0,1){0.11}}
\put(6.5,6){\line(0,-1){0.11}}
\put(6,5){\line(1,0){1}}
\put(6.5,4){\line(0,1){0.11}}
\put(6.5,4){\line(0,-1){0.11}}
\put(6,3){\line(1,0){1}}
\put(6,2){\line(1,0){1}}
\put(6.5,1){\line(0,1){0.11}}
\put(6.5,1){\line(0,-1){0.11}}
\end{picture}
    }
    &
    \subfigure[]{
      \setlength{\unitlength}{0.5cm}
\begin{picture}(7,7)(0,0)
\setcoordinatesystem units <0.5cm,0.5cm> point at 0 0
\setplotarea x from 0 to 7, y from 0 to 7
\put(0.5,0.5){\makebox(0,0)[c]{0}}
\put(0.5,0.5){\circle{0.8}}
\put(1.5,0.5){\makebox(0,0)[c]{2}}
\put(1.5,0.5){\circle{0.8}}
\put(2.5,0.5){\makebox(0,0)[c]{9}}
\put(2.5,0.5){\circle{0.8}}
\put(3.5,0.5){\makebox(0,0)[c]{3}}
\put(3.5,0.5){\circle{0.8}}
\put(4.5,0.5){\makebox(0,0)[c]{8}}
\put(4.5,0.5){\circle{0.8}}
\put(5.5,0.5){\makebox(0,0)[c]{5}}
\put(5.5,0.5){\circle{0.8}}
\put(6.5,0.5){\makebox(0,0)[c]{9}}
\put(6.5,0.5){\circle{0.8}}
\put(0.5,1.5){\makebox(0,0)[c]{1}}
\put(0.5,1.5){\circle{0.8}}
\put(1.5,1.5){\makebox(0,0)[c]{0}}
\put(1.5,1.5){\circle{0.8}}
\put(2.5,1.5){\makebox(0,0)[c]{8}}
\put(2.5,1.5){\circle{0.8}}
\put(3.5,1.5){\makebox(0,0)[c]{4}}
\put(3.5,1.5){\circle{0.8}}
\put(4.5,1.5){\makebox(0,0)[c]{9}}
\put(4.5,1.5){\circle{0.8}}
\put(5.5,1.5){\makebox(0,0)[c]{6}}
\put(5.5,1.5){\circle{0.8}}
\put(6.5,1.5){\makebox(0,0)[c]{7}}
\put(6.5,1.5){\circle{0.8}}
\put(0.5,2.5){\makebox(0,0)[c]{3}}
\put(0.5,2.5){\circle{0.8}}
\put(1.5,2.5){\makebox(0,0)[c]{2}}
\put(1.5,2.5){\circle{0.8}}
\put(2.5,2.5){\makebox(0,0)[c]{7}}
\put(2.5,2.5){\circle{0.8}}
\put(3.5,2.5){\makebox(0,0)[c]{9}}
\put(3.5,2.5){\circle{0.8}}
\put(4.5,2.5){\makebox(0,0)[c]{9}}
\put(4.5,2.5){\circle{0.8}}
\put(5.5,2.5){\makebox(0,0)[c]{1}}
\put(5.5,2.5){\circle{0.8}}
\put(6.5,2.5){\makebox(0,0)[c]{1}}
\put(6.5,2.5){\circle{0.8}}
\put(0.5,3.5){\makebox(0,0)[c]{1}}
\put(0.5,3.5){\circle{0.8}}
\put(1.5,3.5){\makebox(0,0)[c]{1}}
\put(1.5,3.5){\circle{0.8}}
\put(2.5,3.5){\makebox(0,0)[c]{9}}
\put(2.5,3.5){\circle{0.8}}
\put(3.5,3.5){\makebox(0,0)[c]{3}}
\put(3.5,3.5){\circle{0.8}}
\put(4.5,3.5){\makebox(0,0)[c]{4}}
\put(4.5,3.5){\circle{0.8}}
\put(5.5,3.5){\makebox(0,0)[c]{2}}
\put(5.5,3.5){\circle{0.8}}
\put(6.5,3.5){\makebox(0,0)[c]{6}}
\put(6.5,3.5){\circle{0.8}}
\put(0.5,4.5){\makebox(0,0)[c]{1}}
\put(0.5,4.5){\circle{0.8}}
\put(1.5,4.5){\makebox(0,0)[c]{0}}
\put(1.5,4.5){\circle{0.8}}
\put(2.5,4.5){\makebox(0,0)[c]{4}}
\put(2.5,4.5){\circle{0.8}}
\put(3.5,4.5){\makebox(0,0)[c]{1}}
\put(3.5,4.5){\circle{0.8}}
\put(4.5,4.5){\makebox(0,0)[c]{1}}
\put(4.5,4.5){\circle{0.8}}
\put(5.5,4.5){\makebox(0,0)[c]{2}}
\put(5.5,4.5){\circle{0.8}}
\put(6.5,4.5){\makebox(0,0)[c]{5}}
\put(6.5,4.5){\circle{0.8}}
\put(0.5,5.5){\makebox(0,0)[c]{2}}
\put(0.5,5.5){\circle{0.8}}
\put(1.5,5.5){\makebox(0,0)[c]{1}}
\put(1.5,5.5){\circle{0.8}}
\put(2.5,5.5){\makebox(0,0)[c]{9}}
\put(2.5,5.5){\circle{0.8}}
\put(3.5,5.5){\makebox(0,0)[c]{8}}
\put(3.5,5.5){\circle{0.8}}
\put(4.5,5.5){\makebox(0,0)[c]{8}}
\put(4.5,5.5){\circle{0.8}}
\put(5.5,5.5){\makebox(0,0)[c]{9}}
\put(5.5,5.5){\circle{0.8}}
\put(6.5,5.5){\makebox(0,0)[c]{1}}
\put(6.5,5.5){\circle{0.8}}
\put(0.5,6.5){\makebox(0,0)[c]{1}}
\put(0.5,6.5){\circle{0.8}}
\put(1.5,6.5){\makebox(0,0)[c]{3}}
\put(1.5,6.5){\circle{0.8}}
\put(2.5,6.5){\makebox(0,0)[c]{8}}
\put(2.5,6.5){\circle{0.8}}
\put(3.5,6.5){\makebox(0,0)[c]{7}}
\put(3.5,6.5){\circle{0.8}}
\put(4.5,6.5){\makebox(0,0)[c]{8}}
\put(4.5,6.5){\circle{0.8}}
\put(5.5,6.5){\makebox(0,0)[c]{8}}
\put(5.5,6.5){\circle{0.8}}
\put(6.5,6.5){\makebox(0,0)[c]{2}}
\put(6.5,6.5){\circle{0.8}}
\allinethickness{1.5pt}
\put(0,0){\line(1,0){7}}
\put(0,0){\line(0,1){7}}
\put(7,7){\line(-1,0){7}}
\put(7,7){\line(0,-1){7}}
\put(1,6.5){\line(1,0){0.11}}
\put(1,6.5){\line(-1,0){0.11}}
\put(2,7){\line(0,-1){1}}
\put(3,6.5){\line(1,0){0.11}}
\put(3,6.5){\line(-1,0){0.11}}
\put(4,6.5){\line(1,0){0.11}}
\put(4,6.5){\line(-1,0){0.11}}
\put(5,6.5){\line(1,0){0.11}}
\put(5,6.5){\line(-1,0){0.11}}
\put(6,7){\line(0,-1){1}}
\put(1,5.5){\line(1,0){0.11}}
\put(1,5.5){\line(-1,0){0.11}}
\put(2,6){\line(0,-1){1}}
\put(3,5.5){\line(1,0){0.11}}
\put(3,5.5){\line(-1,0){0.11}}
\put(4,5.5){\line(1,0){0.11}}
\put(4,5.5){\line(-1,0){0.11}}
\put(5,5.5){\line(1,0){0.11}}
\put(5,5.5){\line(-1,0){0.11}}
\put(6,6){\line(0,-1){1}}
\put(1,4.5){\line(1,0){0.11}}
\put(1,4.5){\line(-1,0){0.11}}
\put(2,5){\line(0,-1){1}}
\put(3,4.5){\line(1,0){0.11}}
\put(3,4.5){\line(-1,0){0.11}}
\put(4,4.5){\line(1,0){0.11}}
\put(4,4.5){\line(-1,0){0.11}}
\put(5,4.5){\line(1,0){0.11}}
\put(5,4.5){\line(-1,0){0.11}}
\put(6,4.5){\line(1,0){0.11}}
\put(6,4.5){\line(-1,0){0.11}}
\put(1,3.5){\line(1,0){0.11}}
\put(1,3.5){\line(-1,0){0.11}}
\put(2,4){\line(0,-1){1}}
\put(3,4){\line(0,-1){1}}
\put(4,3.5){\line(1,0){0.11}}
\put(4,3.5){\line(-1,0){0.11}}
\put(5,3.5){\line(1,0){0.11}}
\put(5,3.5){\line(-1,0){0.11}}
\put(1,2.5){\line(1,0){0.11}}
\put(1,2.5){\line(-1,0){0.11}}
\put(2,3){\line(0,-1){1}}
\put(3,2.5){\line(1,0){0.11}}
\put(3,2.5){\line(-1,0){0.11}}
\put(4,2.5){\line(1,0){0.11}}
\put(4,2.5){\line(-1,0){0.11}}
\put(5,3){\line(0,-1){1}}
\put(6,2.5){\line(1,0){0.11}}
\put(6,2.5){\line(-1,0){0.11}}
\put(1,1.5){\line(1,0){0.11}}
\put(1,1.5){\line(-1,0){0.11}}
\put(2,2){\line(0,-1){1}}
\put(3,2){\line(0,-1){1}}
\put(4,2){\line(0,-1){1}}
\put(5,1.5){\line(1,0){0.11}}
\put(5,1.5){\line(-1,0){0.11}}
\put(6,1.5){\line(1,0){0.11}}
\put(6,1.5){\line(-1,0){0.11}}
\put(1,0.5){\line(1,0){0.11}}
\put(1,0.5){\line(-1,0){0.11}}
\put(2,1){\line(0,-1){1}}
\put(3,1){\line(0,-1){1}}
\put(4,1){\line(0,-1){1}}
\put(5,0.5){\line(1,0){0.11}}
\put(5,0.5){\line(-1,0){0.11}}
\put(0.5,6){\line(0,1){0.11}}
\put(0.5,6){\line(0,-1){0.11}}
\put(0.5,5){\line(0,1){0.11}}
\put(0.5,5){\line(0,-1){0.11}}
\put(0.5,4){\line(0,1){0.11}}
\put(0.5,4){\line(0,-1){0.11}}
\put(0.5,3){\line(0,1){0.11}}
\put(0.5,3){\line(0,-1){0.11}}
\put(0.5,2){\line(0,1){0.11}}
\put(0.5,2){\line(0,-1){0.11}}
\put(0.5,1){\line(0,1){0.11}}
\put(0.5,1){\line(0,-1){0.11}}
\put(1.5,6){\line(0,1){0.11}}
\put(1.5,6){\line(0,-1){0.11}}
\put(1.5,5){\line(0,1){0.11}}
\put(1.5,5){\line(0,-1){0.11}}
\put(1.5,4){\line(0,1){0.11}}
\put(1.5,4){\line(0,-1){0.11}}
\put(1.5,3){\line(0,1){0.11}}
\put(1.5,3){\line(0,-1){0.11}}
\put(1.5,2){\line(0,1){0.11}}
\put(1.5,2){\line(0,-1){0.11}}
\put(1.5,1){\line(0,1){0.11}}
\put(1.5,1){\line(0,-1){0.11}}
\put(2.5,6){\line(0,1){0.11}}
\put(2.5,6){\line(0,-1){0.11}}
\put(2,5){\line(1,0){1}}
\put(2,4){\line(1,0){1}}
\put(2.5,3){\line(0,1){0.11}}
\put(2.5,3){\line(0,-1){0.11}}
\put(2.5,2){\line(0,1){0.11}}
\put(2.5,2){\line(0,-1){0.11}}
\put(2.5,1){\line(0,1){0.11}}
\put(2.5,1){\line(0,-1){0.11}}
\put(3.5,6){\line(0,1){0.11}}
\put(3.5,6){\line(0,-1){0.11}}
\put(3,5){\line(1,0){1}}
\put(3.5,4){\line(0,1){0.11}}
\put(3.5,4){\line(0,-1){0.11}}
\put(3,3){\line(1,0){1}}
\put(3,2){\line(1,0){1}}
\put(3.5,1){\line(0,1){0.11}}
\put(3.5,1){\line(0,-1){0.11}}
\put(4.5,6){\line(0,1){0.11}}
\put(4.5,6){\line(0,-1){0.11}}
\put(4,5){\line(1,0){1}}
\put(4.5,4){\line(0,1){0.11}}
\put(4.5,4){\line(0,-1){0.11}}
\put(4,3){\line(1,0){1}}
\put(4.5,2){\line(0,1){0.11}}
\put(4.5,2){\line(0,-1){0.11}}
\put(4.5,1){\line(0,1){0.11}}
\put(4.5,1){\line(0,-1){0.11}}
\put(5.5,6){\line(0,1){0.11}}
\put(5.5,6){\line(0,-1){0.11}}
\put(5,5){\line(1,0){1}}
\put(5.5,4){\line(0,1){0.11}}
\put(5.5,4){\line(0,-1){0.11}}
\put(5.5,3){\line(0,1){0.11}}
\put(5.5,3){\line(0,-1){0.11}}
\put(5,2){\line(1,0){1}}
\put(5.5,1){\line(0,1){0.11}}
\put(5.5,1){\line(0,-1){0.11}}
\put(6.5,6){\line(0,1){0.11}}
\put(6.5,6){\line(0,-1){0.11}}
\put(6,5){\line(1,0){1}}
\put(6.5,4){\line(0,1){0.11}}
\put(6.5,4){\line(0,-1){0.11}}
\put(6,2){\line(1,0){1}}
\put(6.5,1){\line(0,1){0.11}}
\put(6.5,1){\line(0,-1){0.11}}
\end{picture}
    }
    &
    \subfigure[]{
      \setlength{\unitlength}{0.5cm}
\begin{picture}(7,7)(0,0)
\setcoordinatesystem units <0.5cm,0.5cm> point at 0 0
\setplotarea x from 0 to 7, y from 0 to 7
\put(0.5,0.5){\makebox(0,0)[c]{0}}
\put(0.5,0.5){\circle{0.8}}
\put(1.5,0.5){\makebox(0,0)[c]{2}}
\put(1.5,0.5){\circle{0.8}}
\put(2.5,0.5){\makebox(0,0)[c]{9}}
\put(2.5,0.5){\circle{0.8}}
\put(3.5,0.5){\makebox(0,0)[c]{3}}
\put(3.5,0.5){\circle{0.8}}
\put(4.5,0.5){\makebox(0,0)[c]{8}}
\put(4.5,0.5){\circle{0.8}}
\put(5.5,0.5){\makebox(0,0)[c]{5}}
\put(5.5,0.5){\circle{0.8}}
\put(6.5,0.5){\makebox(0,0)[c]{9}}
\put(6.5,0.5){\circle{0.8}}
\put(0.5,1.5){\makebox(0,0)[c]{1}}
\put(0.5,1.5){\circle{0.8}}
\put(1.5,1.5){\makebox(0,0)[c]{0}}
\put(1.5,1.5){\circle{0.8}}
\put(2.5,1.5){\makebox(0,0)[c]{8}}
\put(2.5,1.5){\circle{0.8}}
\put(3.5,1.5){\makebox(0,0)[c]{4}}
\put(3.5,1.5){\circle{0.8}}
\put(4.5,1.5){\makebox(0,0)[c]{9}}
\put(4.5,1.5){\circle{0.8}}
\put(5.5,1.5){\makebox(0,0)[c]{6}}
\put(5.5,1.5){\circle{0.8}}
\put(6.5,1.5){\makebox(0,0)[c]{7}}
\put(6.5,1.5){\circle{0.8}}
\put(0.5,2.5){\makebox(0,0)[c]{3}}
\put(0.5,2.5){\circle{0.8}}
\put(1.5,2.5){\makebox(0,0)[c]{2}}
\put(1.5,2.5){\circle{0.8}}
\put(2.5,2.5){\makebox(0,0)[c]{7}}
\put(2.5,2.5){\circle{0.8}}
\put(3.5,2.5){\makebox(0,0)[c]{9}}
\put(3.5,2.5){\circle{0.8}}
\put(4.5,2.5){\makebox(0,0)[c]{9}}
\put(4.5,2.5){\circle{0.8}}
\put(5.5,2.5){\makebox(0,0)[c]{1}}
\put(5.5,2.5){\circle{0.8}}
\put(6.5,2.5){\makebox(0,0)[c]{1}}
\put(6.5,2.5){\circle{0.8}}
\put(0.5,3.5){\makebox(0,0)[c]{1}}
\put(0.5,3.5){\circle{0.8}}
\put(1.5,3.5){\makebox(0,0)[c]{1}}
\put(1.5,3.5){\circle{0.8}}
\put(2.5,3.5){\makebox(0,0)[c]{9}}
\put(2.5,3.5){\circle{0.8}}
\put(3.5,3.5){\makebox(0,0)[c]{3}}
\put(3.5,3.5){\circle{0.8}}
\put(4.5,3.5){\makebox(0,0)[c]{4}}
\put(4.5,3.5){\circle{0.8}}
\put(5.5,3.5){\makebox(0,0)[c]{2}}
\put(5.5,3.5){\circle{0.8}}
\put(6.5,3.5){\makebox(0,0)[c]{6}}
\put(6.5,3.5){\circle{0.8}}
\put(0.5,4.5){\makebox(0,0)[c]{1}}
\put(0.5,4.5){\circle{0.8}}
\put(1.5,4.5){\makebox(0,0)[c]{0}}
\put(1.5,4.5){\circle{0.8}}
\put(2.5,4.5){\makebox(0,0)[c]{4}}
\put(2.5,4.5){\circle{0.8}}
\put(3.5,4.5){\makebox(0,0)[c]{1}}
\put(3.5,4.5){\circle{0.8}}
\put(4.5,4.5){\makebox(0,0)[c]{1}}
\put(4.5,4.5){\circle{0.8}}
\put(5.5,4.5){\makebox(0,0)[c]{2}}
\put(5.5,4.5){\circle{0.8}}
\put(6.5,4.5){\makebox(0,0)[c]{5}}
\put(6.5,4.5){\circle{0.8}}
\put(0.5,5.5){\makebox(0,0)[c]{2}}
\put(0.5,5.5){\circle{0.8}}
\put(1.5,5.5){\makebox(0,0)[c]{1}}
\put(1.5,5.5){\circle{0.8}}
\put(2.5,5.5){\makebox(0,0)[c]{9}}
\put(2.5,5.5){\circle{0.8}}
\put(3.5,5.5){\makebox(0,0)[c]{8}}
\put(3.5,5.5){\circle{0.8}}
\put(4.5,5.5){\makebox(0,0)[c]{8}}
\put(4.5,5.5){\circle{0.8}}
\put(5.5,5.5){\makebox(0,0)[c]{9}}
\put(5.5,5.5){\circle{0.8}}
\put(6.5,5.5){\makebox(0,0)[c]{1}}
\put(6.5,5.5){\circle{0.8}}
\put(0.5,6.5){\makebox(0,0)[c]{1}}
\put(0.5,6.5){\circle{0.8}}
\put(1.5,6.5){\makebox(0,0)[c]{3}}
\put(1.5,6.5){\circle{0.8}}
\put(2.5,6.5){\makebox(0,0)[c]{8}}
\put(2.5,6.5){\circle{0.8}}
\put(3.5,6.5){\makebox(0,0)[c]{7}}
\put(3.5,6.5){\circle{0.8}}
\put(4.5,6.5){\makebox(0,0)[c]{8}}
\put(4.5,6.5){\circle{0.8}}
\put(5.5,6.5){\makebox(0,0)[c]{8}}
\put(5.5,6.5){\circle{0.8}}
\put(6.5,6.5){\makebox(0,0)[c]{2}}
\put(6.5,6.5){\circle{0.8}}
\allinethickness{1.5pt}
\put(0,0){\line(1,0){7}}
\put(0,0){\line(0,1){7}}
\put(7,7){\line(-1,0){7}}
\put(7,7){\line(0,-1){7}}
\put(1,6.5){\line(1,0){0.11}}
\put(1,6.5){\line(-1,0){0.11}}
\put(2,7){\line(0,-1){1}}
\put(3,6.5){\line(1,0){0.11}}
\put(3,6.5){\line(-1,0){0.11}}
\put(4,6.5){\line(1,0){0.11}}
\put(4,6.5){\line(-1,0){0.11}}
\put(5,6.5){\line(1,0){0.11}}
\put(5,6.5){\line(-1,0){0.11}}
\put(6,7){\line(0,-1){1}}
\put(1,5.5){\line(1,0){0.11}}
\put(1,5.5){\line(-1,0){0.11}}
\put(2,6){\line(0,-1){1}}
\put(3,5.5){\line(1,0){0.11}}
\put(3,5.5){\line(-1,0){0.11}}
\put(4,5.5){\line(1,0){0.11}}
\put(4,5.5){\line(-1,0){0.11}}
\put(5,5.5){\line(1,0){0.11}}
\put(5,5.5){\line(-1,0){0.11}}
\put(6,6){\line(0,-1){1}}
\put(1,4.5){\line(1,0){0.11}}
\put(1,4.5){\line(-1,0){0.11}}
\put(2,4.5){\line(1,0){0.11}}
\put(2,4.5){\line(-1,0){0.11}}
\put(3,4.5){\line(1,0){0.11}}
\put(3,4.5){\line(-1,0){0.11}}
\put(4,4.5){\line(1,0){0.11}}
\put(4,4.5){\line(-1,0){0.11}}
\put(5,4.5){\line(1,0){0.11}}
\put(5,4.5){\line(-1,0){0.11}}
\put(6,4.5){\line(1,0){0.11}}
\put(6,4.5){\line(-1,0){0.11}}
\put(1,3.5){\line(1,0){0.11}}
\put(1,3.5){\line(-1,0){0.11}}
\put(2,4){\line(0,-1){1}}
\put(3,4){\line(0,-1){1}}
\put(4,3.5){\line(1,0){0.11}}
\put(4,3.5){\line(-1,0){0.11}}
\put(5,3.5){\line(1,0){0.11}}
\put(5,3.5){\line(-1,0){0.11}}
\put(6,3.5){\line(1,0){0.11}}
\put(6,3.5){\line(-1,0){0.11}}
\put(1,2.5){\line(1,0){0.11}}
\put(1,2.5){\line(-1,0){0.11}}
\put(2,3){\line(0,-1){1}}
\put(3,2.5){\line(1,0){0.11}}
\put(3,2.5){\line(-1,0){0.11}}
\put(4,2.5){\line(1,0){0.11}}
\put(4,2.5){\line(-1,0){0.11}}
\put(5,3){\line(0,-1){1}}
\put(6,2.5){\line(1,0){0.11}}
\put(6,2.5){\line(-1,0){0.11}}
\put(1,1.5){\line(1,0){0.11}}
\put(1,1.5){\line(-1,0){0.11}}
\put(2,2){\line(0,-1){1}}
\put(3,1.5){\line(1,0){0.11}}
\put(3,1.5){\line(-1,0){0.11}}
\put(5,1.5){\line(1,0){0.11}}
\put(5,1.5){\line(-1,0){0.11}}
\put(6,1.5){\line(1,0){0.11}}
\put(6,1.5){\line(-1,0){0.11}}
\put(1,0.5){\line(1,0){0.11}}
\put(1,0.5){\line(-1,0){0.11}}
\put(2,1){\line(0,-1){1}}
\put(5,0.5){\line(1,0){0.11}}
\put(5,0.5){\line(-1,0){0.11}}
\put(6,0.5){\line(1,0){0.11}}
\put(6,0.5){\line(-1,0){0.11}}
\put(0.5,6){\line(0,1){0.11}}
\put(0.5,6){\line(0,-1){0.11}}
\put(0.5,5){\line(0,1){0.11}}
\put(0.5,5){\line(0,-1){0.11}}
\put(0.5,4){\line(0,1){0.11}}
\put(0.5,4){\line(0,-1){0.11}}
\put(0.5,3){\line(0,1){0.11}}
\put(0.5,3){\line(0,-1){0.11}}
\put(0.5,2){\line(0,1){0.11}}
\put(0.5,2){\line(0,-1){0.11}}
\put(0.5,1){\line(0,1){0.11}}
\put(0.5,1){\line(0,-1){0.11}}
\put(1.5,6){\line(0,1){0.11}}
\put(1.5,6){\line(0,-1){0.11}}
\put(1.5,5){\line(0,1){0.11}}
\put(1.5,5){\line(0,-1){0.11}}
\put(1.5,4){\line(0,1){0.11}}
\put(1.5,4){\line(0,-1){0.11}}
\put(1.5,3){\line(0,1){0.11}}
\put(1.5,3){\line(0,-1){0.11}}
\put(1.5,2){\line(0,1){0.11}}
\put(1.5,2){\line(0,-1){0.11}}
\put(1.5,1){\line(0,1){0.11}}
\put(1.5,1){\line(0,-1){0.11}}
\put(2.5,6){\line(0,1){0.11}}
\put(2.5,6){\line(0,-1){0.11}}
\put(2,5){\line(1,0){1}}
\put(2,4){\line(1,0){1}}
\put(2.5,3){\line(0,1){0.11}}
\put(2.5,3){\line(0,-1){0.11}}
\put(2.5,2){\line(0,1){0.11}}
\put(2.5,2){\line(0,-1){0.11}}
\put(2.5,1){\line(0,1){0.11}}
\put(2.5,1){\line(0,-1){0.11}}
\put(3.5,6){\line(0,1){0.11}}
\put(3.5,6){\line(0,-1){0.11}}
\put(3,5){\line(1,0){1}}
\put(3.5,4){\line(0,1){0.11}}
\put(3.5,4){\line(0,-1){0.11}}
\put(3,3){\line(1,0){1}}
\put(3.5,1){\line(0,1){0.11}}
\put(3.5,1){\line(0,-1){0.11}}
\put(4.5,6){\line(0,1){0.11}}
\put(4.5,6){\line(0,-1){0.11}}
\put(4,5){\line(1,0){1}}
\put(4.5,4){\line(0,1){0.11}}
\put(4.5,4){\line(0,-1){0.11}}
\put(4,3){\line(1,0){1}}
\put(4.5,2){\line(0,1){0.11}}
\put(4.5,2){\line(0,-1){0.11}}
\put(4.5,1){\line(0,1){0.11}}
\put(4.5,1){\line(0,-1){0.11}}
\put(5.5,6){\line(0,1){0.11}}
\put(5.5,6){\line(0,-1){0.11}}
\put(5,5){\line(1,0){1}}
\put(5.5,4){\line(0,1){0.11}}
\put(5.5,4){\line(0,-1){0.11}}
\put(5.5,3){\line(0,1){0.11}}
\put(5.5,3){\line(0,-1){0.11}}
\put(5,2){\line(1,0){1}}
\put(5.5,1){\line(0,1){0.11}}
\put(5.5,1){\line(0,-1){0.11}}
\put(6.5,6){\line(0,1){0.11}}
\put(6.5,6){\line(0,-1){0.11}}
\put(6.5,5){\line(0,1){0.11}}
\put(6.5,5){\line(0,-1){0.11}}
\put(6.5,4){\line(0,1){0.11}}
\put(6.5,4){\line(0,-1){0.11}}
\put(6,2){\line(1,0){1}}
\put(6.5,1){\line(0,1){0.11}}
\put(6.5,1){\line(0,-1){0.11}}
\end{picture}
    }
  \end{tabular}
\end{center}
 \caption{
Example from~\cite{Soille2007} of a 7x7 image and its 
partitions into $(\alpha,\omega)$-connected components
using identical values for the local and global range parameters ranging
from 1 to 6. 
(a) (1, 1)-CCs. (b) (2, 2)-CCs. (c) (3, 3)-CCs. (d) (4, 4)-CCs.
(e) (5, 5)-CCs. (f) (6, 6)-CCs. 
}
 \label{fig:alphaomegapami}
\end{figure*}
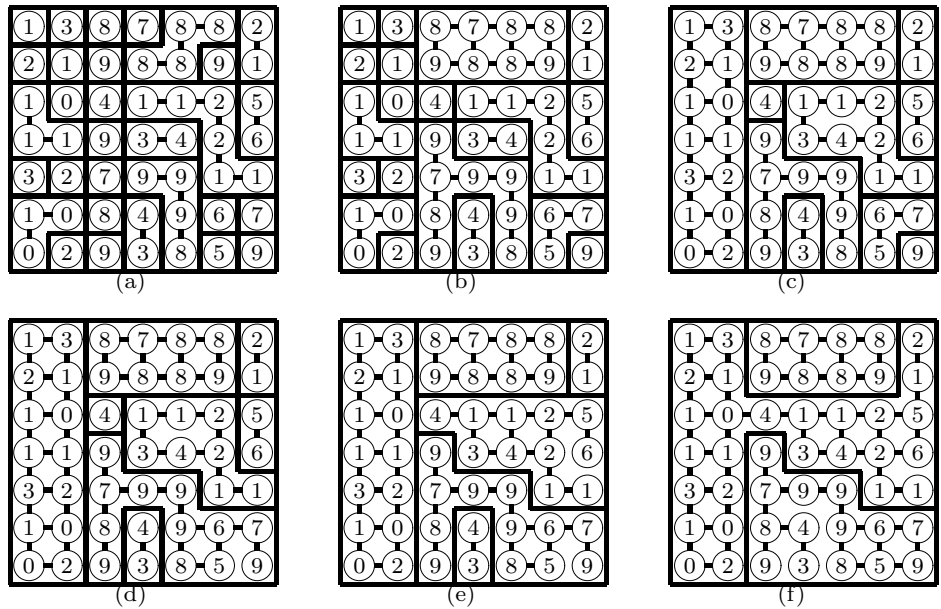

\subsection{Ultrametric watershed for constrained connectivity}
In that section, we show how to build a weighted graph on which the
ultrametric watershed corresponding to the hierarchy of constrained
connectivity can be computed. Intuitively, this weighted graph can be
seen as the gradient of the original image. We compute an ultrametric
watershed for the hierarchy of $\alpha$-connected components. We
filter that watershed to obtain the family of
$(\alpha,\omega)$-connected components. We then show how to directly
compute the ultrametric watershed corresponding to the hierarchy of
$(\alpha,\omega)$-connected components.

Constrained connectivity is a hierarchy of flat zones of $f$, in the
sense where the $0$-connected components of $f$ are the zones of $f$
where the intensity of $f$ does not change. In a continuous world,
such zones would be the ones where the gradient is null, {\em i.e.}
$\nabla f=0$. However, the space we are working with is discrete, and
a flat zone of $f$ can consist in a single point. In general, it is
not possible to compute a gradient on the points or on the edges such
that this gradient is null on the flat zones. To compute a gradient on
the edges such that the gradient is null on the flat zones, we need to
``double'' the graph, for example we can do that by doubling the
number of points of $V$ and adding one edge between each new point and
the old one (see Fig.~\ref{fig:doubling}(b)).
\begin{figure}[htbp]
\begin{center}
  \begin{tabular}{ccc}
    \subfigure[]{
      \includegraphics[width=.2\columnwidth]{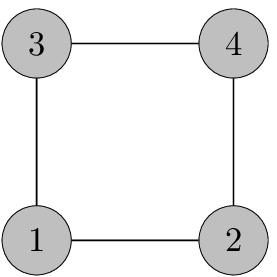}
    }
    &
    ~~~~~~
    &
    \subfigure[]{
      \includegraphics[width=.3\columnwidth]{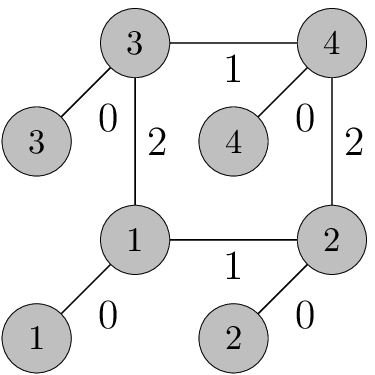}
    }
  \end{tabular}
\end{center}
 \caption{Doubling the graph. (a) Original graph with weights $f$ on
 the vertices. (b)~Double graph, with weigths $f$ on the vertices and
 the gradient $F$ on the edges (see text).} 
\label{fig:doubling}
\end{figure}

More precisely, if we denote the points of $V$ by
$V=\{x_0,\ldots,x_n\}$, we set $V'=\{x'_0,\ldots,x'_n\}$ (with $V\cap
V'=\emptyset$), and $E'=\{\{x_i,x'_i\}\st 0\leq i\leq n\}$. We then
set $V_1=V\cup V'$ and $E_1=E\cup E'$.
By construction, as $G=(V,E)$ is a connected graph, the graph
$G_1=(V_1,E_1)$ is a connected graph.
We also extend $f$ to $V'$, by setting, for any $x'\in V'$,
$f(x')=f(x)$, where $\{x,x'\}\in E'$. 

Let $(V_1,E_1,F)$ be the weighted graph obtained from $f$ by setting,
for any $\{x,y\}\in E_1$, $F(\{x,y\})=|f(x)-f(y)|$. The map $F$ can be
seen as the ``natural gradient'' of $f$~\cite{Mattiussi-2000}. It is
easy to see that the flat zones of $f$, {\em i.e.} the $0$-connected
components of $f$ are (in bijection with) the connected components of
the set $\{v=\{x,y\}\in E_1\st F(\{x,y\})=0\}$.

\begin{figure}[htbp]
\begin{center}
  \begin{tabular}{ccc}
    \subfigure[]{
      \includegraphics[width=.3\columnwidth]{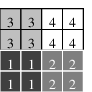}
    }
    &
    ~~~~~~
    &
    \subfigure[]{
      \includegraphics[width=.3\columnwidth]{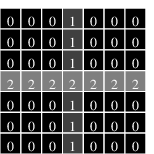}
    }
  \end{tabular}
\end{center}
 \caption{Doubling the graph as an image. (b) Doubling the graph of Fig.~\ref{fig:doubling}.a. 
	(b) The gradient of (a) (see text).}
  \label{fig:doublingimage}
\end{figure}

Let us note that it is also possible, for the purpose of
visualisation, to double the graph as an image, {\em i.e.}, to
multiply the size of the image by 2. On the graph of
Fig.~\ref{fig:doubling}.a, that gives the image of
Fig.~\ref{fig:doublingimage}.a. Then the gradient can be seen as an
image (Fig.~\ref{fig:doublingimage}.b) as described in
section~\ref{sec:representation}. This representation will be adopted
in all the subsequent figures of the paper.

\begin{figure*}[htbp]
\begin{center}
  \begin{tabular}{c}	
  \begin{tabular}{cc}	
    \subfigure[Original image]{
      \includegraphics[width=.3\textwidth]{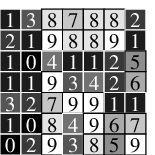}
    }
    &
    \subfigure[Doubled image]{
      \includegraphics[width=.3\textwidth]{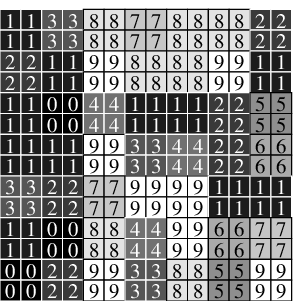}
    }
  \end{tabular}
    \\
    \subfigure[Gradient]{
      \includegraphics[width=.45\textwidth]{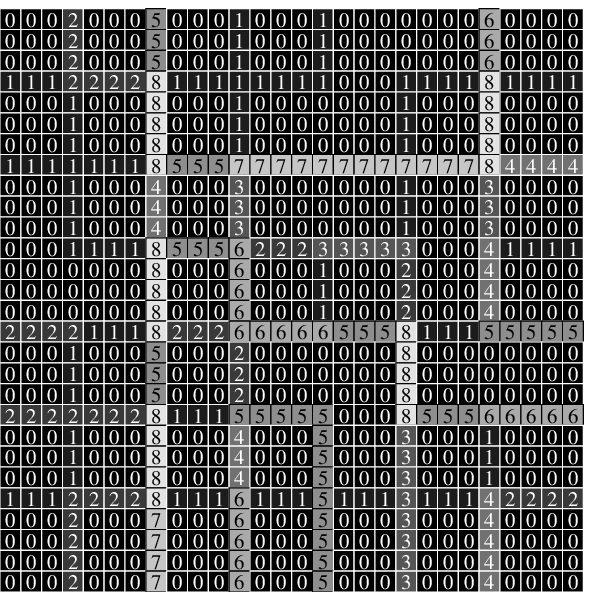}
    }
    \\
  \begin{tabular}{cc}
    \subfigure[Ultrametric watershed for the $\alpha$-connectivity]{
      \includegraphics[width=.45\textwidth]{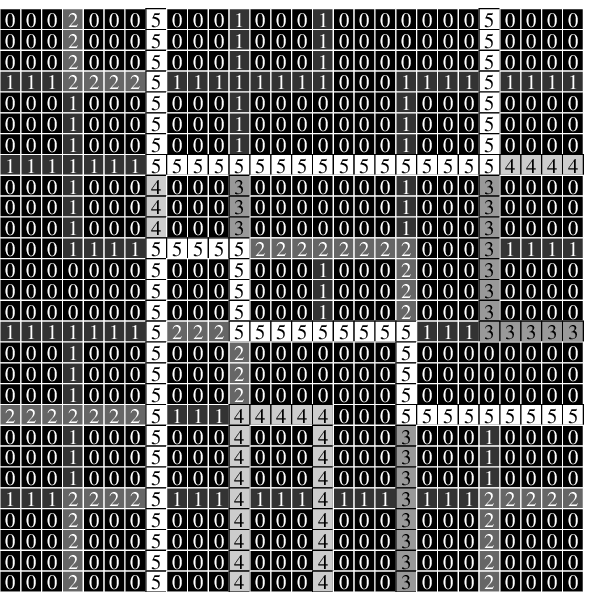}
    }
    &
    \subfigure[Ultrametric watershed for the constrained connectivity]{
      \includegraphics[width=.45\textwidth]{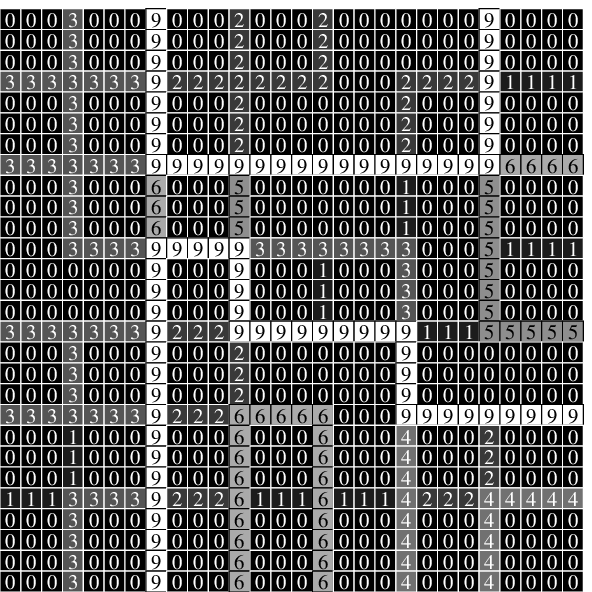}
    }

  \end{tabular}
  \end{tabular}
\end{center}
 \caption{Example of a constrained connectivity
 hierarchy. (a) Original image (the one of Fig.~\ref{fig:alphapami})
 from~\cite{Soille2007}. (b) Doubling of (a). (c) Gradient of (b). (d)
 Topological watershed of the gradient, that is the ultrametric watershed
 $W^1$ for the $\alpha$-connectivity that. (e)
 Ultrametric watershed $W^2$ for the constrained connectivity. (see text)}
 \label{fig:pamiSoille}
\end{figure*}

Let $W^1$ be a topological watershed of $F$. From
Th.~\ref{th:wtopocarac} and Eq.~\ref{eq:connection}, if
$W^1(\{x,y\})=\lambda$, there exists a path $\pi=\{x_0=x, \ldots,
x_n=y\}$ linking $x$ to $y$ such that the altitude of any edge along
$\pi$ is below $\lambda$, {\em i.e.} we have, for any $0\leq i<n$,
$F(\{x_i,x_{i+1}\})=|f(x_i)-f(x_{i+1})|\leq\lambda$. The following
property, the proof of which is left to the reader, states that the
hierarchy of $\alpha$-connected components is given by $W^1$.
\begin{proper}
\label{pr:alpha}
We have
\begin{itemize}
\item $W^1$ is an ultrametric watershed; 
\item $W^1$ is uniquely defined (if $W'$ is a topological watershed of
$F$, then $W'=W^1$);
\item let $\lambda\geq 0$ and let $X$ be a connected component of the
cross-section $W^1[\lambda]$; then for any $x\in V(X)\setminus V'$,
$\lambda\mbox{-CC}(x)=V(X)\setminus V'$.
\end{itemize}
\end{proper}
Pr.~\ref{pr:alpha} is illustrated on Fig.~\ref{fig:pamiSoille}.d. Let
us stress that Fig.~\ref{fig:pamiSoille}.d sums up in one image all
the images of Fig.~\ref{fig:alphapami}.

\begin{figure*}[htbp]
\begin{center}
  \begin{tabular}{cc}
    \subfigure[Original image]{
      \includegraphics[width=.45\textwidth]{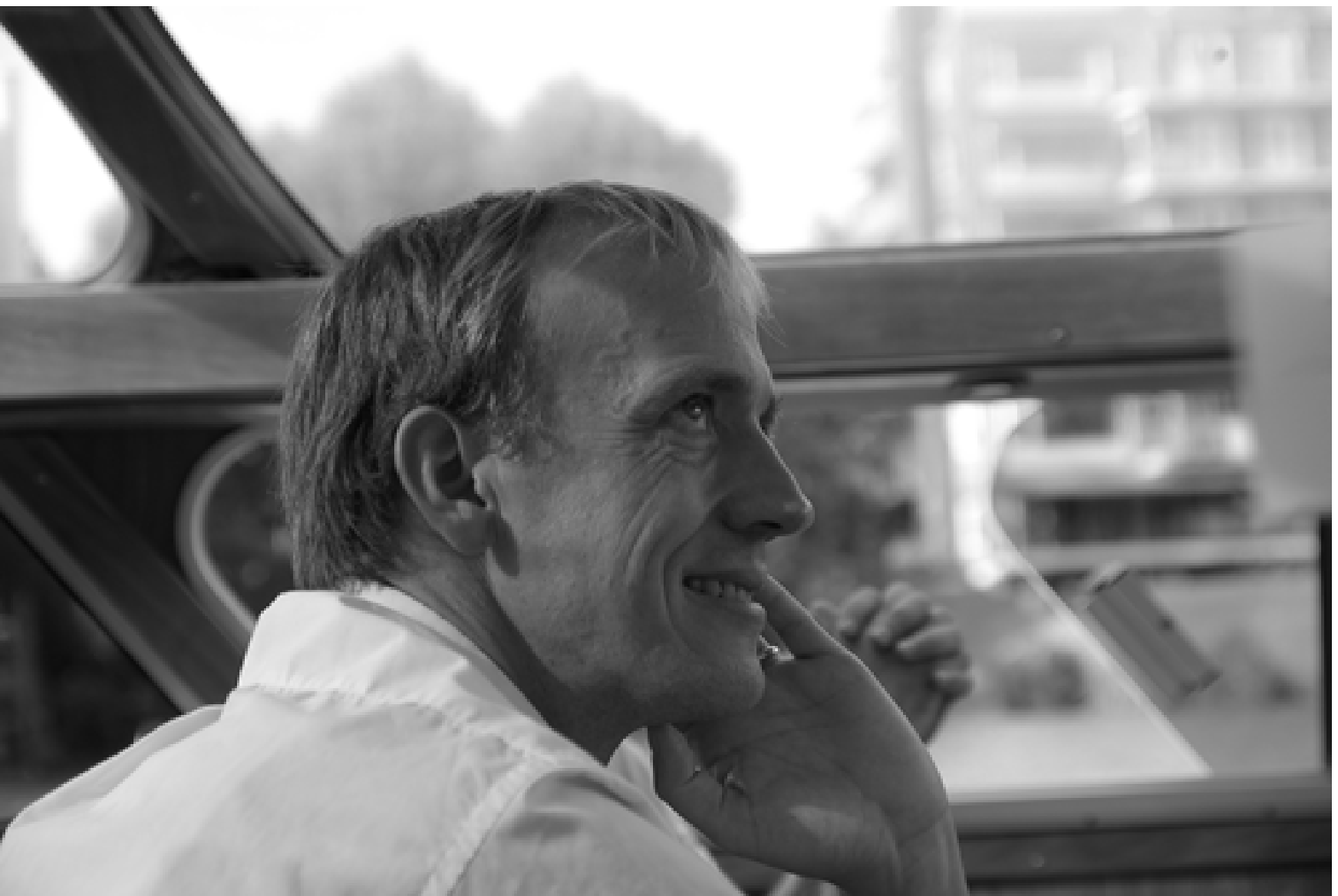}
    }
    &
    \subfigure[$W^1$(logarithmic grey-scale)]{
      \includegraphics[width=.45\textwidth]{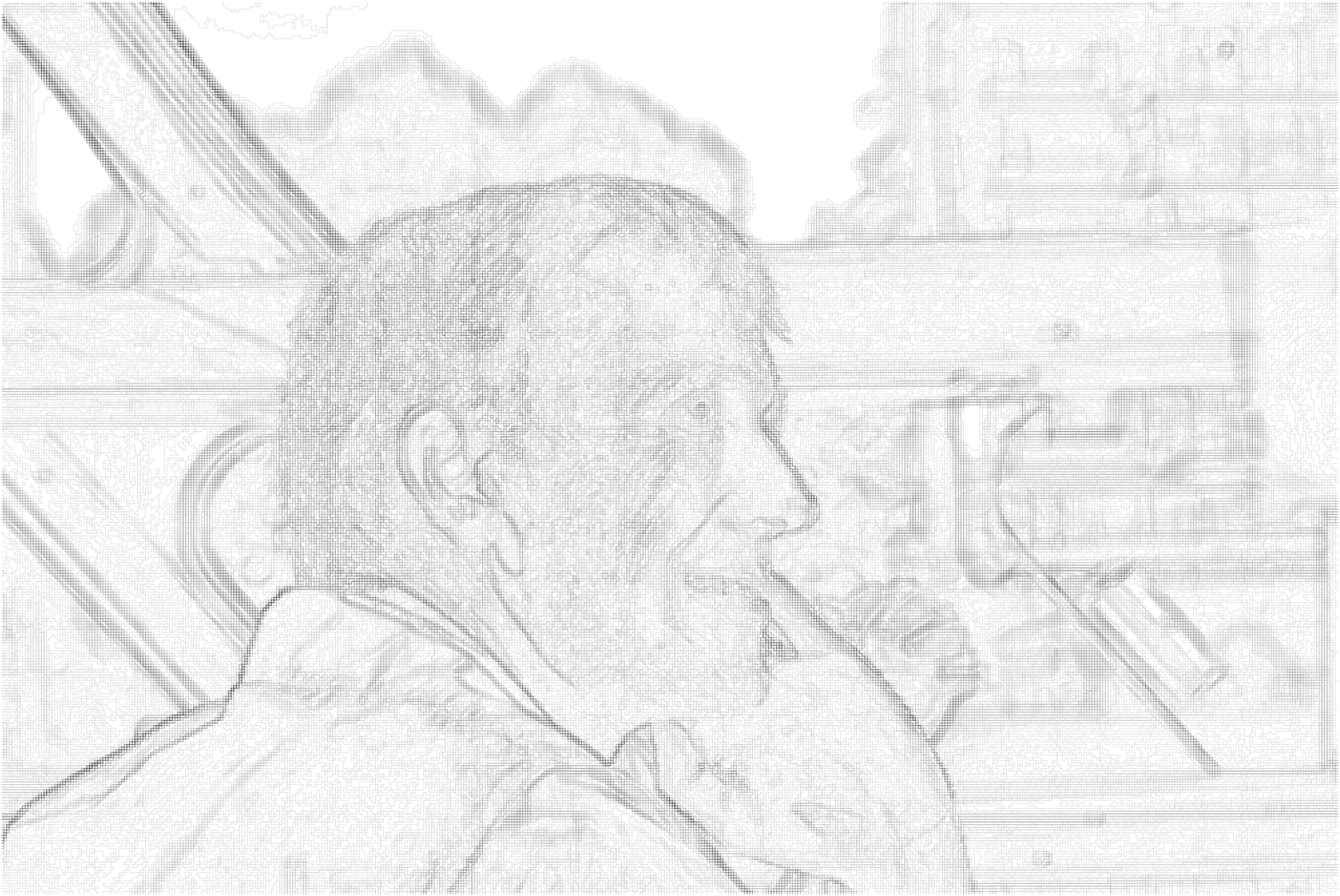}
    }
    \\
    \subfigure[$W^2$]{
      \includegraphics[width=.45\textwidth]{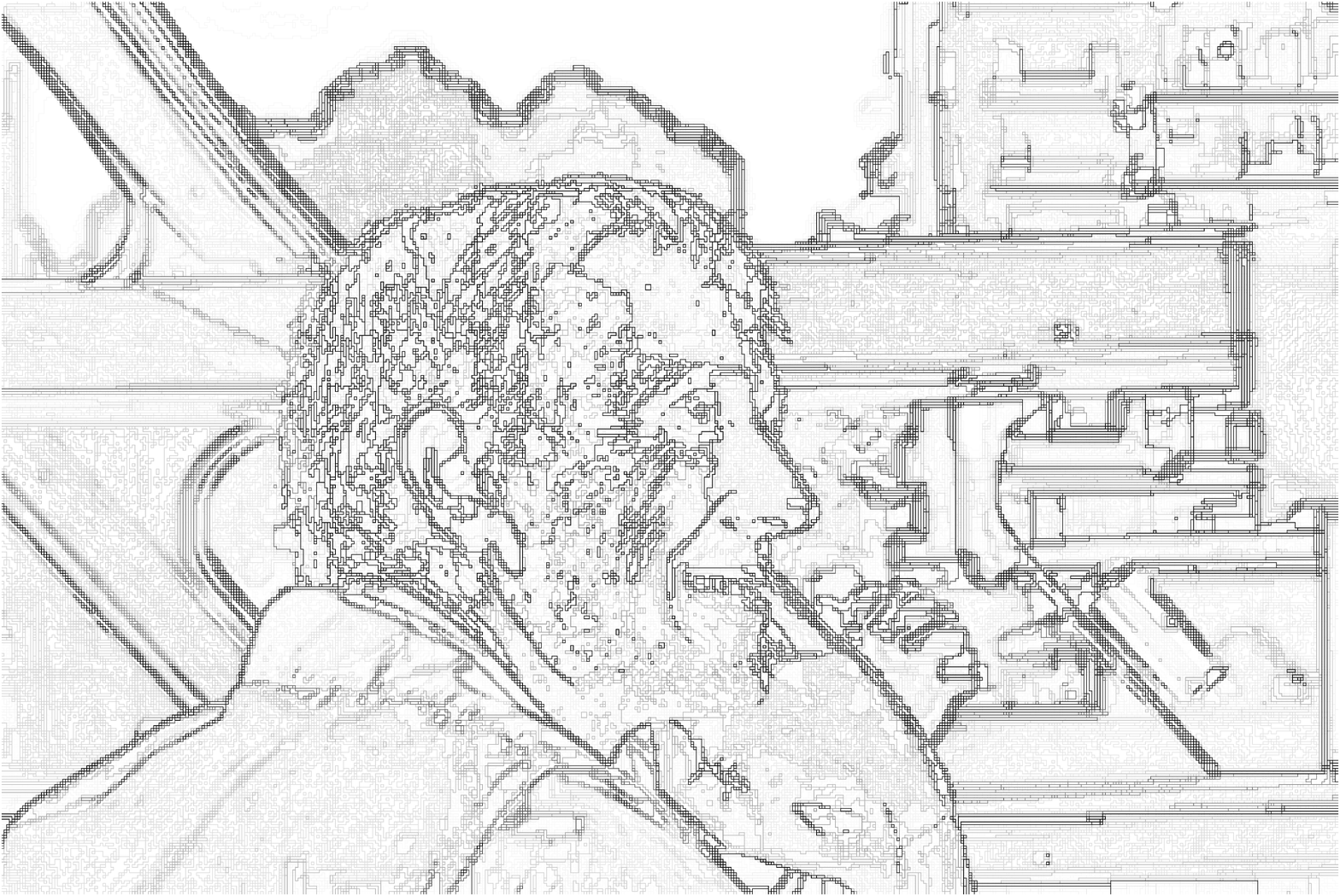}
    }
    &
    \subfigure[Area-filtering ultrametric watershed]{
      \includegraphics[width=.45\textwidth]{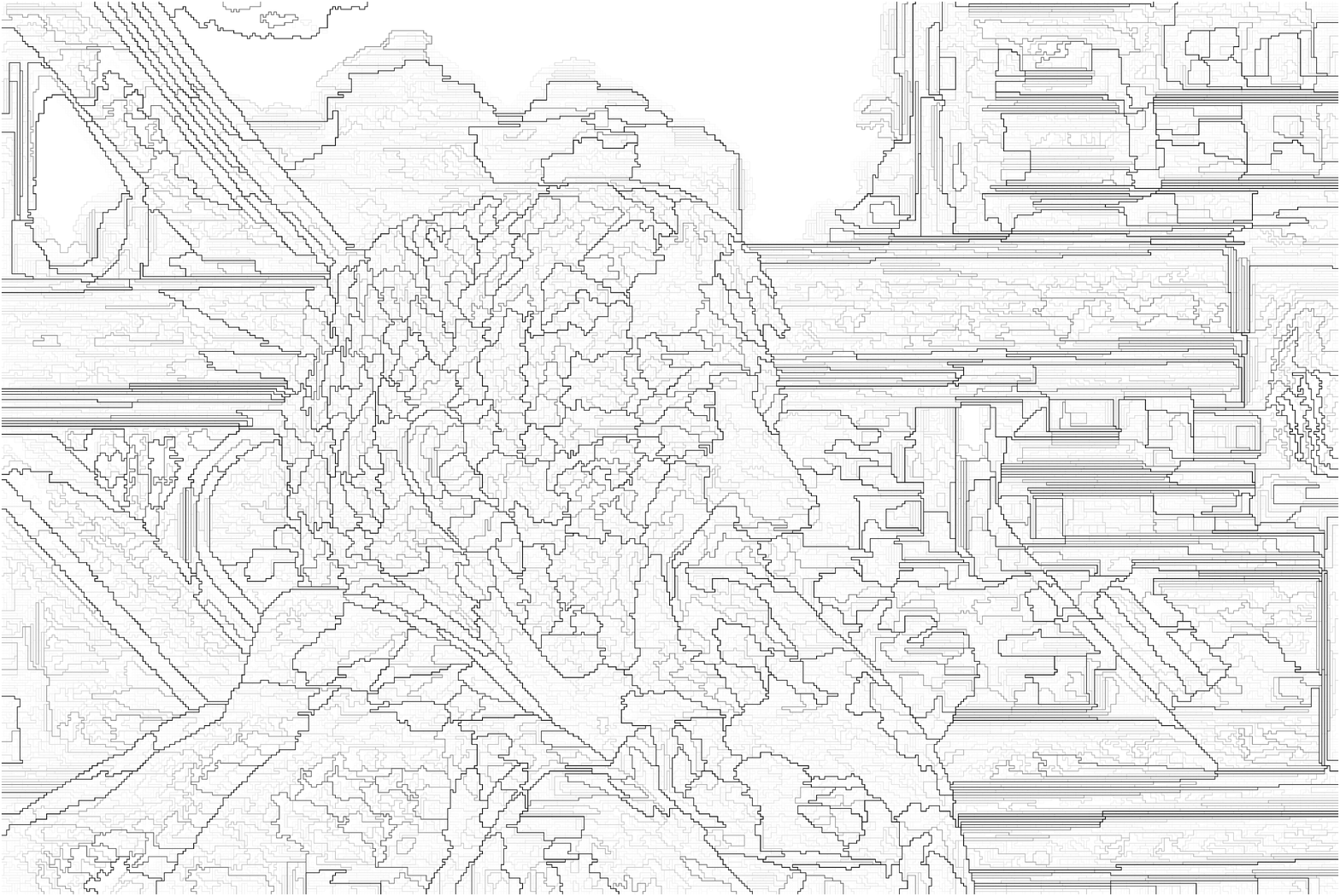}
    }
  \end{tabular}
\end{center}
 \caption{Soille's ($\alpha,\omega$)-constrained connectivity
 hierarchy. (a) Original image. (b) Ultrametric watershed $W^1$ for
 the $\alpha$-connectivity. (c) Ultrametric watershed $W^2$ for the
 constrained connectivity. (d) Ultrametric watersheds corresponding to
 one of the possible hierarchies of area-filterings on $W^2$.}
 \label{fig:alphaomega}
\end{figure*}

One can notice that $R_f$ is increasing on $2^V$, {\em i.e.}
$R_f(X)\subset R_f(Y)$ whenever $X\subseteq Y$. Thus $R_f$ is
increasing on ${\mathcal C}(W^1)$, and by removing the connected
components of ${\cal C}(W^1)$ that are below a threshold $\omega$ for
$R_f$, we have an attribute filtering which is idempotent (the values
on the points do not change), thus it is a closing.  More precisely,
we denote by $(R_\lambda)_{\lambda\geq 0)}$ the family of maps
obtained by applying this closing on $W^1$ for varying $\lambda$, {\em
  i.e.}, for any $\lambda \geq 0$ and any $\{x,y\}\in E_1$, we set
\begin{eqnarray}
R_\lambda(\{x,y\}) = \min\{\lambda' & \st & [\lambda',C]\in{\cal C}(W^1),\nonumber\\
& & x\in V(C), y\in V(C), \nonumber\\
& & R_f(V(C))\geq \lambda\}
\end{eqnarray}
In other words, the altitude for $R_\lambda$ of the edge $\{x,y\}$ is
the altitude of the lowest component of ${\cal C}(W^1)$ that contains
both $x$ and $y$ and such that the range of that component is greater
than $\lambda$.

The family $(R_\lambda)_{\lambda\geq 0}$ allows us to retrieve the
$(\alpha,\omega)$-CCs of $f$: surprisingly, it can be shown
that any $R_\lambda$ is a topological watershed, and thus
$\Minima{R_\lambda}$ is a segmentation from which it is easy to
extract the $(\lambda,\lambda)$-connected component of a point, as the
minimum of $\Minima{R_\lambda}$ that contains that point (See
Pr.~\ref{pr:alphaomega} below for a more formal setting).

Moreover, one can directly compute the ultrametric watershed
associated to the hierarchy of $(\alpha,\omega)$-constrai\-ned
connectivity. We set:
\begin{eqnarray}
\label{eq:alphaomega}
W^2(\{x,y\}) = \min\{R_f(V(C)) & \st& [\lambda,C]\in{\cal C}(W^1),\nonumber\\
& &  x\in V(C), \nonumber\\
& &  y\in V(C)\}
\end{eqnarray}
In other words, the altitude for $W^2$ of the edge $\{x,y\}$ is the
range of the lowest component of ${\cal C}(W^1)$ that contains both
$x$ and $y$.  One can remark that Eq.~\ref{eq:alphaomega} corresponds
to Eq.~\ref{eq:alphaomegaCC} for the framework of segmentation.

The following property, the proof of which is left to the reader,
states that the hierarchy of $(\alpha,\omega)$-connected components is
given by $W^2$.
\begin{proper}
\label{pr:alphaomega}
We have 
\begin{itemize}
\item $\forall\lambda\geq0$, $R_\lambda$ is a topological watershed;
\item $\forall\lambda\geq 0$,  $W^2[\lambda]=\Minima{R_\lambda}$ ;
\item $W^2$ is an ultrametric watershed;
\item $W^2$ is uniquely defined;
\item let $\lambda\geq 0$ and let $X$ be a connected component of the
cross-section $W^2[\lambda]$; then for any $x\in V(X)\setminus
V'$, $(\lambda,\lambda)\mbox{-CC}(x)=V(X)\setminus V'$.
\end{itemize}
\end{proper}
Prop.~\ref{pr:alphaomega}, illustrated on Fig.~\ref{fig:pamiSoille}.e,
thus gives an efficient algorithm to compute the hierarchy of
$(\alpha,\omega)$-constrained connectivity. Indeed,
Eq.~\ref{eq:alphaomega} can be computed in constant
time~\cite{BenderFarach-Colton-2000} on ${\cal C}(W^1)$, which itself
can be computed in quasi-linear time~\cite{NajCou2006}. Such an
algorithm is much faster than the one proposed in~\cite{Soille2007},
that computes only one level of the hierarchy.

Let us stress that for an algorith\-mic/implementation point of view, it
is not necessary in practice to double the image. Furthermore, for an
efficient computation of the hierarchy, a minimum spanning tree or a
component tree of the gradient can also be used instead of an
ultrametric watershed, without changing the overall theoretical
complexity of the algorithm. But for visualisation purpose, the
ultrametric watershed is necessary.  Moreover, those tools can be
combined; indeed, one can compute a topological watershed on the graph
of a minimum spanning tree.  In a forthcomming paper, we will propose
various data structures, including but not limited to component tree
and minimum spanning tree, that allows an efficient computation of
hierarchical segmentations. We will also study how to extend
Prop.~\ref{pr:alphaomega} in order to compute any granulometry of
operators (strong hierarchies in the sense of~\cite{Serra-2006}).

A example of the application of the properties of this section to a
real image is given in Fig.~\ref{fig:alphaomega}.
Visualising $W^2$ allows to assess some of the qualities of the
hierarchy of constrained connectivity. One can notice in
Fig.~\ref{fig:alphaomega}.c a large number of transition regions
(small undersirable regions that persist in the hierarchy), and this
problem is known~\cite{Soille-Grazzini-2009}. As $W^2$ is an image, a
number of classical morphological schemes ({\em e.g.}, area filtering
that produces a hierarchy of regions classified according to their
size or area; see~\cite{meyer.najman:segmentation} for more details)
can be used to remove those transition zones (see
Fig.~\ref{fig:alphaomega}.d for an example). Studying the usefulness
of such schemes is the subject of future research.


\subsection{Links with other hierarchical schemes}
As we have shown, and as stated by Th.~\ref{th:onetoone}, any
hierarchical scheme can be represented by {\em and computed through}
an ultrametric watershed. This is in particular true for the classical
watershed-based segmentation algorithms.

Fig.~\ref{fig:salience} is an illustration of the application of the
framework developped in this paper to a classical hierarchical
segmentation scheme based on attribute
opening~\cite{NS96,Salembier-Oliveras-Garrido-1998,meyer.najman:segmentation}.
The attribute opening tends to produce large plateaus where a
watershed can be located anywhere; in particular, the contours at a
given level of the hierarchy can be choosen differently depending on
the filtering level, and one has to take care of indeed producing a
hierarchy. In contrast, the ultrametric watershed will always choose
a contour that is present at a lower level of the hierarchy.

Fig.~\ref{fig:salienceextract} shows some of the differences between
applying an ultrametric-watershed scheme and applying a classical
watershed-based segmentation scheme, {\em e.g.}  attribute opening
followed by a watershed~\cite{meyer-beucher90}. As watershed
algorithms generally place watershed lines in the middle of plateaus,
the contours produced by the classical watershed-based segmentation
scheme do not lead to a hierarchy, and the two schemes give quite
different results.

\begin{figure*}[htbp]
\begin{center}
  \begin{tabular}{ccc}
    \subfigure[Original image]{
      \includegraphics[width=0.3\textwidth]{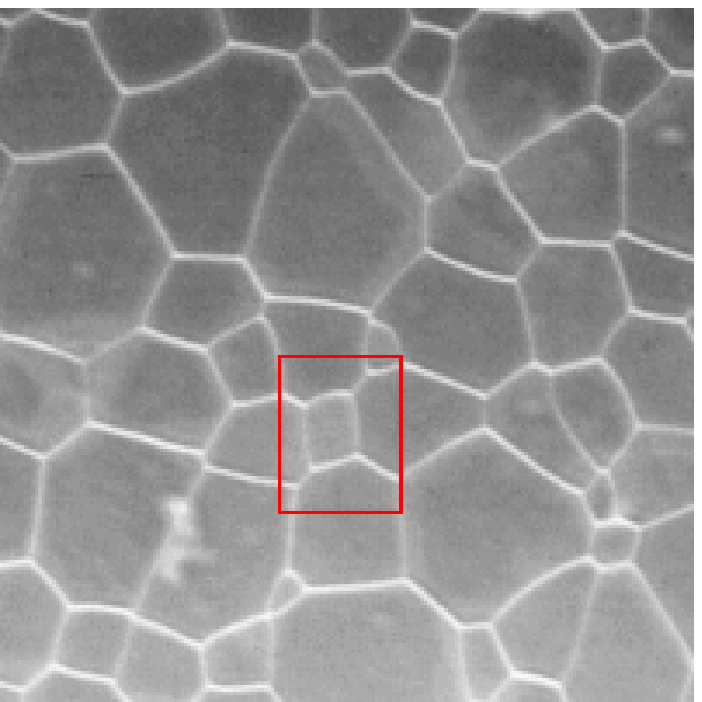}
    }
    &
    \subfigure[Ultrametric watershed]{
      \includegraphics[width=0.3\textwidth]{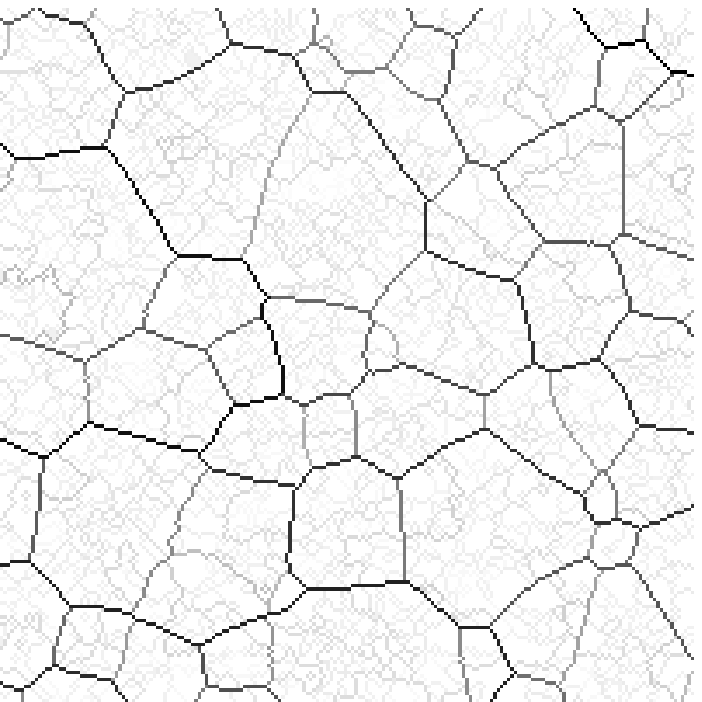}
    }
    &
    \subfigure[Cross section of (b)]{
      \includegraphics[width=0.3\textwidth]{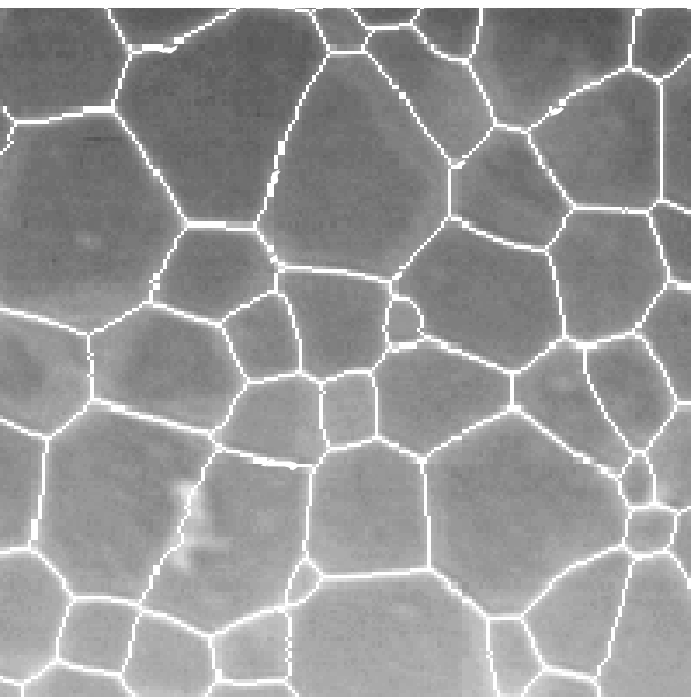}
    }
  \end{tabular}
\end{center}
 \caption{Example of ultrametric watershed. On (a), a box is drawn
   around the zoomed part used in Fig.~\ref{fig:salienceextract}}
 \label{fig:salience}
\end{figure*}
 
\begin{figure}[htbp]
\begin{center}
  \begin{tabular}{cc}
    \subfigure[]{
      \includegraphics[width=3.cm]{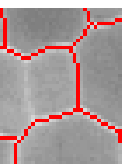}
    }
    &
    \subfigure[]{
      \includegraphics[width=3.cm]{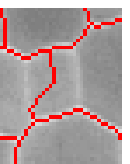}
    }
  \end{tabular}
\end{center}
 \caption{Zoom on a comparison between two watersheds of a filtered
 version of the image~\ref{fig:salience}.a. Morphological filtering
 tends to create large plateaus, and both watersheds (a) and (b) are
 possible, but only (a) is a subset of a watershed
 of~\ref{fig:salience}.a. No hierarchical scheme will ever give a
 result as~(b). }
 \label{fig:salienceextract}
\end{figure}

\section{Conclusion}
In this paper, we have shown (Th.~\ref{th:onetoone}) that any
hierarchical segmentation can be represented by an ultrametric
watershed, and conversely that any ultrametric watershed leads to a
hierarchical segmentation.  Th.~\ref{th:onetoone} thus offers an
alternative way of thinking hierarchical segmentation that complete
existing ones (ultrametric distances, minimum spanning tree, $\ldots$)
We have seen how to apply Th.~\ref{th:onetoone} to directly compute
the constrained connectivity hierachy as an ultrametric watershed,
leading to a fast algorithm. An important research direction is to
provide a generalization of this scheme for computing any hierachical
segmentation.

As a step in this direction, future work will propose novel algorithms
(based on the topological watershed algorithm~\cite{CNB05}) to compute
ultrametric watersheds, with proof of correctness.  It is important to
note that most of the algorithms proposed in the literature to compute
saliency maps are not correct, often because they rely on wrong
connection values or because they rely on thick watersheds where
merging regions is difficult~\cite{CouBerCou2008}.

On a more theoretical level, this work can be pursued in several
directions.  
\begin{itemize}
\item We will study lattices of watersheds~\cite{IGMI_CouNajSer08} and
will bring to that framework recent approaches like
scale-sets~\cite{GuiguesCM06} and other metric approaches to
segmentation~\cite{Arbelaez-Cohen-2006}. For example, scale-sets
theory considers a rather general formulation of the partitioning
problem which involves minimizing a two-term-based energy, of the form
$\lambda C + D$, where $D$ is a goodness-of-fit term and $C$ is a
regularization term, and proposes an algorithm to compute the
hierarchical segmentation we obtain by varying the $\lambda$
parameter. As in the case of constrained connectivity (see
section~\ref{sec:constcon} above), we can hope that the topological
watershed algorithm~\cite{CNB05} can be used on a specific energy
function to directly obtain the hierarchy.

\item Subdominant theory (mentionned at the end of section~\ref{sec:hier})
links hierarchical classification and optimisation. In particular, the
subdominant ultrametric $d'$ of a dissimilarity $d$ is the solution to
the following optimisation problem for $p<\infty$:
\begin{eqnarray}
\min \{ || d - d' ||_p^p & \st & d' \mbox{~is an ultrametric
distance }\nonumber\\ 
& & \mbox{~~ and~}d'\leq d  \}
\end{eqnarray}
It is certainly of interest to search if topological watersheds can be
solutions of similar optimisation problems.

\item Several generalisations of hierarchical clustering have been proposed
in the literature~\cite{BBO2004}. An interesting direction of research
is to see how to extend in the same way the topological watershed
approach, for example for allowing regions to overlap.

\item Last, but not least, the links of hierarchical segmentation
with connective segmentation~\cite{Serra-2006} have to be studied.
\end{itemize}

\begin{acknowledgements}
The author would like to thank P.~Soille for the permission to use
Fig.~\ref{fig:alphapami}, \ref{fig:alphaomegapami}
and~\ref{fig:alphaomega}.a (that was shot at ISMM'09 by the author),
as well as J.-P. Coquerez for the permission to use
Fig.~\ref{fig:introduction}.

For numerous discussions, the author feels indebted to (by
alphabetical order): Gilles Bertrand, Michel Couprie, Jean Cousty,
Christian Ronse, Jean Serra and Hugues Talbot.
\end{acknowledgements}

%

\onecolumn

\parpic(1in,1.25in){\includegraphics[width=1in,height=1.25in,clip,keepaspectratio]{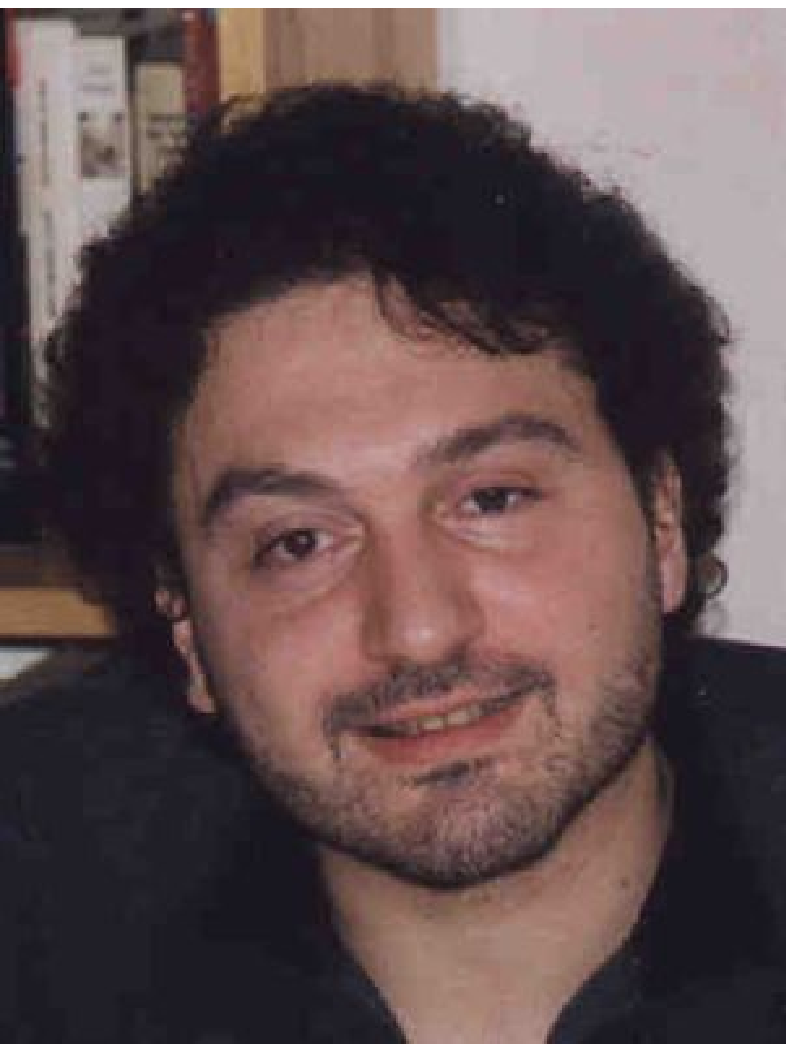}}
{\bf Laurent Najman}
received the Habilitation à Diriger les Recherches in 2006 from
University the University of Marne-la-Vallée, the Ph.D. degree in applied
mathematics from Universit\'e Paris-Dauphine in 1994 with the highest
honor (Félicitations du Jury) and the engineering degree from the
Ecole Nationale Sup\'erieure des Mines de Paris in 1991. He worked in the central research 
laboratories of Thomson-CSF for three years after his engineering degree, working on infrared image
segmentation problems using mathematical morphology. He then joined a start-up
company named Animation Science in 1995, as director of research and
development. The particle systems technology for computer graphics
and scientific visualization developed by the company under his
technical leadership received several awards, including the ``European
Information Technology Prize 1997'' awarded by the European Commission
(Esprit program\-me) and by the European Council for Applied Science and
Engineering as well as the ``Hottest Products of the Year 1996'' awarded by the
Computer Graphics World journal. In 1998, he joined OC\'E Print Logic
Technologies, as senior scientist. There he worked  on various  image analysis problems 
related to scanning and printing. In 2002, he
joined the Informatics Department of ESIEE, Paris, where he is
professor and a member of the Gaspard-Monge computer science research laboratory (LIGM), Universit\'e
Paris-Est. His current research interest is discrete mathematical morphology.


\end{document}